\documentclass[showpacs,twocolumn,aps,prx,longbibliography,superscriptaddress,notitlepage]{revtex4-2}
\usepackage{qcircuit}
\usepackage[dvips]{graphicx}
\usepackage{amsmath,amssymb,amsthm,mathrsfs,amsfonts,dsfont}
\usepackage{subfigure, epsfig}
\usepackage{braket}
\usepackage{bm}
\usepackage{enumerate}
\usepackage{algorithm}
\usepackage{algpseudocode}
\usepackage{diagbox}
\usepackage{physics}
\usepackage{color}
\usepackage{multirow}
\usepackage[marginal]{footmisc}
\usepackage{comment}
\usepackage{tikz}
\usepackage[]{qcircuit}
\usepackage{makecell}
\usetikzlibrary{arrows}
\usetikzlibrary{shapes,fadings,snakes}
\usetikzlibrary{decorations.pathmorphing,patterns}
\usetikzlibrary{calc}
\usetikzlibrary{positioning}
\usepackage[colorlinks = true]{hyperref}

\graphicspath{{./figure/}}

\newtheorem{theorem}{Theorem}

\newtheorem{fact}{Fact}
\newtheorem{lemma}{Lemma}
\newtheorem{corollary}{Corollary}
\newtheorem{proposition}{Proposition}

\newcommand{\comments}[1]{}

\begin{document}
\title{Realizing Non-Physical Actions through Hermitian-Preserving Map Exponentiation}
\begin{abstract}
Quantum mechanics features a variety of distinct properties such as coherence and entanglement, which could be explored to showcase potential advantages over classical counterparts in information processing. In general, legitimate quantum operations must adhere to principles of quantum mechanics, particularly the requirements of complete positivity and trace preservation. Nonetheless, non-physical maps, especially Hermitian-preserving maps, play a crucial role in quantum information science. To date, there exists {\it no} effective method for implementing these non-physical maps with quantum devices. In this work, we introduce the Hermitian-preserving map exponentiation algorithm, which can effectively realize the action of an arbitrary Hermitian-preserving map  by encoding its output into a quantum process. We analyze the performances of this algorithm, including its sample complexity and robustness, and prove its optimality in certain cases. When combined with algorithms such as the Hadamard test and quantum phase estimation, it allows for the extraction of information and generation of states from outputs of Hermitian-preserving maps, enabling various applications. Utilizing positive but not completely positive maps, this algorithm provides exponential advantages in entanglement detection and quantification compared to protocols based on single-copy operations. In addition, it facilitates the recovery of noiseless quantum states from multiple copies of noisy states by implementing the inverse map of the corresponding noise channel, offering an intriguing approach to handling quantum errors. Our findings present a pathway for systematically and efficiently implementing non-physical actions with quantum devices, thereby boosting the exploration of potential quantum advantages across a wide range of information processing tasks.
\end{abstract}

\date{\today}

\author{Fuchuan Wei}
\thanks{These authors contributed equally to this work.}
\affiliation{Yau Mathematical Sciences Center and Department of Mathematics, Tsinghua University, Beijing 100084, China}

\author{Zhenhuan Liu}
\thanks{These authors contributed equally to this work.}
\affiliation{Center for Quantum Information, Institute for Interdisciplinary Information Sciences, Tsinghua University, Beijing 100084, China}

\author{Guoding Liu}
\affiliation{Center for Quantum Information, Institute for Interdisciplinary Information Sciences, Tsinghua University, Beijing 100084, China}

\author{Zizhao Han}
\affiliation{Center for Quantum Information, Institute for Interdisciplinary Information Sciences, Tsinghua University, Beijing 100084, China}

\author{Xiongfeng Ma}
\email{xma@tsinghua.edu.cn}
\affiliation{Center for Quantum Information, Institute for Interdisciplinary Information Sciences, Tsinghua University, Beijing 100084, China}

\author{Dong-Ling Deng}
\email{dldeng@tsinghua.edu.cn}
\affiliation{Center for Quantum Information, Institute for Interdisciplinary Information Sciences, Tsinghua University, Beijing 100084, China}
\affiliation{Shanghai Qi Zhi Institute, 41th Floor, AI Tower, No. 701 Yunjin Road, Xuhui District, Shanghai 200232, China}
\affiliation{Hefei National Laboratory, Hefei 230088, People’s Republic of China}

\author{Zhengwei Liu}
\email{liuzhengwei@mail.tsinghua.edu.cn}
\affiliation{Yau Mathematical Sciences Center and Department of Mathematics, Tsinghua University, Beijing 100084, China}
\affiliation{Yanqi Lake Beijing Institute of Mathematical Sciences and Applications, Beijing 100407, China}

\maketitle

\section{Introduction}
\subsection{Backgrounds and Motivations}
Principles of quantum mechanics dictate that quantum operations must act on and output density matrices, which are positive matrices with a unit trace \cite{nielsen2010quantum}. Thus, a valid quantum operation, known as a quantum channel, must be completely positive and trace-preserving (CPTP). Compared to operations in classical physics, CPTP maps encompass a broader range of possibilities, enabling coherent and entangled operations that underpin quantum advantages observed in various quantum information processing tasks. However, as depicted in Fig.~\ref{fig:overview}(a), the set of CPTP maps represents only a small subset of all linear maps \cite{watrous2018theory}. In practical terms, the CPTP constraint limits the performance of many quantum tasks. Notably, entanglement detection \cite{gunhe2009entanglement} and quantum error mitigation \cite{endo2021hybrid,cai2022quantum} rely on some linear maps that lie beyond the scope of CPTP requirements. We coin such maps as non-physical maps in this context.

In entanglement detection and quantification, positive but not completely positive maps \cite{gunhe2009entanglement} serve as crucial tools. By deciding the positivity of output matrices generated by acting positive maps on subsystems, corresponding entanglement criteria exhibit strong detection capabilities. For instance, the positive partial transposition criterion \cite{peres1996ppt}, based on the transposition map, is widely employed for entanglement detection of mixed-states \cite{elben2020mixed,vidal2002negativity} and entanglement distillation \cite{Horodecki1998distillation}. Moreover, the entanglement negativity, which quantifies the violation of the positive partial transposition criterion, represents an easily computable and operationally meaningful entanglement measure \cite{vidal2002negativity}. However, due to their lack of complete positivity, verifying positive map criteria often requires highly joint operations or exponential repetition times \cite{gray2018machine,zhou2020Single,elben2020mixed}. 

In contrast to quantum error correction \cite{terhal2015qec}, which sacrifices a large number of ancillary qubits to protect noiseless states, quantum error mitigation aims to suppress errors using currently available quantum devices \cite{cai2022quantum}. The core idea of leading approaches for mitigating quantum errors is built on applying the inverse map of the noise channel to noisy states. However, since the inverse maps of noise channels are always non-physical and quantum error mitigation protocols are confined by restricted quantum operations, techniques such as probabilistic error cancellation \cite{temme2017mitigation,endo2018practical} are adopted to statistically realize these inverse maps. Consequently, quantum error mitigation protocols can only recover noiseless expectation values rather than noiseless states, limiting the range of applications.

Given the importance of non-physical maps, particularly Hermitian-preserving ones as introduced above, researchers have devoted substantial efforts to their physical realizations. However, existing approaches are restricted in the level of quantum states and aim to prepare outputs of non-physical maps in some indirect ways. Since the output of the non-physical map is not always a density matrix, these approaches face fundamental limitations in feasibility and efficiency. For example, methods based on structural approximation \cite{horodecki2002direct,korb2008structural} and Petz recovery map \cite{Petz1986,gilyen2022petz} employ quantum channels to approximate non-physical maps. However, the approximate channel may largely deviate from the target non-physical map. The multi-copy extension method \cite{dong2019positive} utilizes a joint quantum channel acting on multiple copies of input states to produce a single output of the non-physical map, making it feasible only when the output remains a density matrix. Other attempts, such as the probabilistic error cancellation mentioned above, decompose the non-physical map into a linear combination of some quantum channels, realizing the non-physical map only in a statistical manner. An effective and practical approach to implementing non-physical maps remains elusive. In this work, we address this crucial problem by proposing a systematic approach to efficiently realize the actions of all Hermitian-preserving maps.

\subsection{Overview and Structure of the Paper}
The implementation of non-physical maps poses a challenge mainly due to the requirement that the output of a valid quantum operation must be a density matrix. To overcome this restriction, our core idea is to change the carrier of density matrices. More concretely, although the output of a Hermitian-preserving map, $\mathcal{N}(\rho)$, might not necessarily be a density matrix but rather a general Hermitian matrix, Hermitian matrices still have physical meanings, such as Hamiltonians determining evolutions of physical systems. Therefore, by exponentiating a Hermitian-preserving map $\mathcal{N}(\cdot)$, we define a new map, $e^{-i\mathcal{N}(\cdot)t}$, which maps an input state to a unitary evolution. This new map contains all the information of the Hermitian-preserving map $\mathcal{N}$ and transforms this non-physical map from a purely mathematical object into a physical process. For simplicity, throughout this context, the phrase ``implement a non-physical map $\mathcal{N}$" generally denotes the realization of this new map $e^{-i\mathcal{N}(\cdot)t}$.

\begin{figure}[t]
\centering
\includegraphics[width=0.45\textwidth]{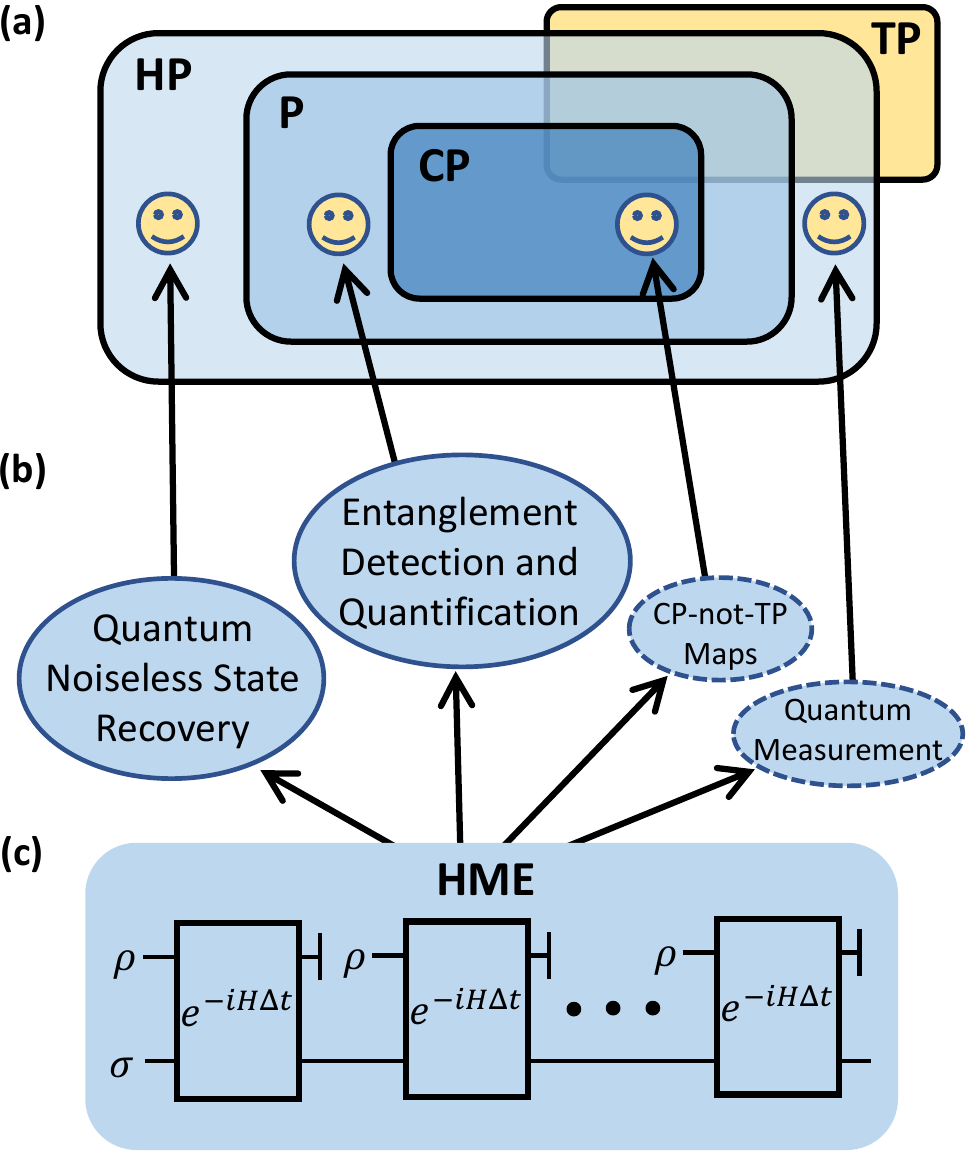}
\caption{(a) Diagrammatic representation of different types of maps, including CP (completely positive), P (positive), HP (Hermitian-preserving), and TP (trace-preserving) maps. Physical maps lie in the intersection of CP and TP. The smiling faces represent the maps used in the four applications listed in (b). These applications are all situated within the HP maps region, among which we extensively explore entanglement detection and quantification, as well as quantum noiseless state recovery. (c) Circuit diagram for HME, comprising three components: sequential input of identical states $\rho$, the evolved state $\sigma$ preserved using quantum memory, and joint Hamiltonian evolution.
}
\label{fig:overview}
\end{figure}

To realize the map $e^{-i\mathcal{N}(\cdot)t}$, we design a quantum algorithm called the \emph{Hermitian-preserving map exponentiation} (HME), as depicted in Fig.~\ref{fig:overview}(c). The quantum circuit of this algorithm involves sequentially preparing identical input states $\rho$ and repeatedly evolving the joint system under the evolution of $e^{-iH\Delta t}$, where the Hamitonian $H$ is determined by $\mathcal{N}$. The density matrix exponentiation algorithm \cite{Lloyd2014qpca,kjaergaard2022DME}, which has been proved to exhibit exponential speedup compared to single-copy strategies in certain tasks \cite{huang2022quantum}, can be regarded as a special case of HME by setting $\mathcal{N}$ to be the identity map. Given that Hermitian-preserving maps are significantly more general than CPTP maps and encompass a wide range of crucial non-physical maps, HME has the potential to be a key tool in various quantum information processing tasks. To assess the performance of HME, we conduct a comprehensive analysis, including its sample complexity for achieving the desired accuracy in realizing $e^{-i\mathcal{N}(\cdot)t}$ and its robustness against noises originating from the Hamiltonians and input states. We find that the infinite norm of the Hamiltonian $H$, the evolution time $t$, and the desired accuracy $\epsilon$ play pivotal roles in determining its performance. In addition, we prove that the HME algorithm achieves the lowest sample complexity for exponentiating certain Hermitian-preserving maps.

As mentioned earlier, one direct application of HME lies in entanglement detection and quantification, as positive maps are Hermitian-preserving. In entanglement detection, we incorporate HME into the quantum phase estimation algorithm \cite{kitaev1995quantum,nielsen2010quantum} to propose a new entanglement detection protocol. We demonstrate through an example that the HME-based protocol can offer exponential advantages compared to all single-copy approaches. Moreover, by combining the HME and Hadamard test algorithm \cite{knill1998power}, we develop a protocol to estimate entanglement negativity and compare it to conventional means, showcasing advantages in resource consumption. Another significant application is quantum noiseless state recovery, which relies on the ability of HME to implement the inverse map of a noise channel. We integrate HME into a simple quantum circuit to recover the noiseless state with arbitrary precision and analyze its performance. This protocol can handle any invertible noise when the description of the noise channel is known. In contrast to existing methods such as quantum error mitigation and quantum error correction, our protocol can recover the desired noiseless state from multiple noisy states, establishing a new approach for combating quantum noises. Furthermore, we explore the potential of HME in other scenarios, such as measuring expectation values and implementing completely positive but not trace-preserving maps, as listed in Fig.~\ref{fig:overview}(b). Notably, most of the quantum circuits used in these applications only require a small number of ancilla qubits, enhancing the feasibility of HME.

This paper is organized as follows. We introduce the HME algorithm in Sec.~\ref{sec:HME} and analyze its performances in Sec.~\ref{sec:performances}, including the upper bound on sample complexity in Sec.~\ref{subsec:complexity_upper_bound}, the robustness against different kinds of errors in Sec.~\ref{subsec:robustness}, and the optimality analysis in Sec.~\ref{subsec:optimality}. In Sec.~\ref{sec:ent_det}, we use two kinds of positive maps as examples to show how HME yields exponential speedups in the tasks of entanglement detection and quantification. In Sec.~\ref{sec:QNSR}, we show how to use HME to recover the noiseless quantum state and discuss its comparison with quantum error correction and mitigation. In Sec.~\ref{sec:other_app}, we discuss the applications of HME in the tasks of expectation value measurements and quantum algorithms including linear combination of unitaries and quantum imaginary time evolution. In Sec.~\ref{sec:conclusion}, we summarize the advantages of HME, discuss some possible generalizations, and pose some remaining open problems concerning the implementation of non-physical maps.

\subsection{Notations and Definitions}\label{sec:notation}
In this paper, we use letters $\mathcal{X}$, $\mathcal{Y}$, $\mathcal{H}$, $\mathcal{K}$, and $\mathcal{R}$ to represent Hilbert spaces. The symbol $\mathcal{N}$ is used to denote the  Hermitian-preserving map that we aim to implement. The symbols $\mathcal{U}$, $\mathcal{Q}$, and $\mathcal{I}$ are used to represent the CPTP maps that are directly implementable on physical systems. Scratch letters like $\mathscr{H}$ and $\mathscr{F}$ are used to represent sets. $\mathbb{I}$ denotes the identity matrix. We denote $L(\mathcal{H})$ as the set of all linear operators on the Hilbert space $\mathcal{H}$, and $ D(\mathcal{H})$ as the set of all density operators on $\mathcal{H}$. Furthermore, $ T(\mathcal{X},\mathcal{Y})$ represents the set of all linear maps from $L(\mathcal{X})$ to $L(\mathcal{Y})$. For a pure state $\ket{\psi}$, we sometimes omit the Dirac notation and denote the density matrix $\ketbra{\psi}{\psi}$ by $\psi$ for brevity. For a unitary operator, $U\in L(\mathcal{H})$, we denote the corresponding unitary channel as $\left[U\right]$, which satisfies $\left[U\right](\sigma)=U\sigma U^\dagger$ for $\sigma\in D(\mathcal{H})$. The notation of $\log(\cdot)$ in this work represents the natural logarithm. For a Hermitian matrix $A$ with spectral decomposition $A=\sum_{i}\lambda_i\ketbra{\alpha_i}{\alpha_i}$, we denote $A^+:=\sum_{i:\lambda_i>0}\lambda_i\ketbra{\alpha_i}{\alpha_i}$ as the positive part of $A$ and $A^-:=-\sum_{i:\lambda_i<0}\lambda_i\ketbra{\alpha_i}{\alpha_i}$ as the negative part of $A$.

When drawing tensor network diagrams, we use rounded rectangles to represent tensors and use lines of different colors to represent indices belonging to different systems. While drawing quantum circuits, we use regular rectangles to represent quantum operations and black lines to represent quantum systems.

We denote the operator norm as $\norm{\cdot}_{\infty}$, which represents the maximum singular value of a matrix, and the trace norm as $\norm{\cdot}_1$, which is the sum of singular values of a matrix. For two quantum states $\rho$ and $\sigma$, the trace distance and fidelity between them are $\norm{\rho-\sigma}_1$ and $F(\rho,\sigma):=\left(\tr\sqrt{\sqrt{\sigma}\rho\sqrt{\sigma}}\right)^2$, respectively. For a Hermitian-preserving map $\mathcal{Q}\in T(\mathcal{X},\mathcal{Y})$, its diamond norm is defined as $\norm{\mathcal{Q}}_{\diamond}
:=\sup_{\sigma_{\mathcal{X}\mathcal{R}}\in D(\mathcal{X}\mathcal{R}),\mathcal{R}}\norm{(\mathcal{Q}\otimes\mathcal{I}_{\mathcal{R}})(\sigma_{\mathcal{X}\mathcal{R}})}_1$,
where the symbol $\mathcal{X}\mathcal{R}$ represents the Hilbert space $\mathcal{X}\otimes\mathcal{R}$, $\mathcal{I}_\mathcal{R}$ is the identity channel on $\mathcal{R},$ and the supremum is taken over the reference system $\mathcal{R}$ of an arbitrary dimension and all the states $\sigma_{\mathcal{X}\mathcal{R}}\in D(\mathcal{X}\mathcal{R})$. The diamond distance of two Hermitian-preserving maps $\mathcal{Q}_1$ and $\mathcal{Q}_2$ is defined as $\norm{\mathcal{Q}_1-\mathcal{Q}_2}_{\diamond}$.

\section{Hermitian-preserving Map Exponentiation}\label{sec:HME}

The circuit for HME is shown in Fig.~\ref{fig:overview}(c). To implement the evolution of $e^{-i\mathcal{N}(\rho)t}$ on the state $\sigma$, the HME algorithm begins with preparing two quantum systems. The target state $\rho$ will be prepared on the first system, while the second system serves as a quantum memory to keep the state on which the evolution of $e^{-i\mathcal{N}(\rho)t}$ is applied. In the beginning, the quantum memory is prepared in the initial state $\sigma$. Based on the desired accuracy, the non-physical map $\mathcal{N}$, and the total evolution time $t$, one determines an appropriate Hamiltonian $H$ and a short time period $\Delta t$. These choices are guided by Theorem~\ref{theorem:maintheorem} and Theorem~\ref{theorem:cost}. Subsequently, one repeats the following steps for a total of $K=t/\Delta t$ times:
\begin{enumerate}
\item Prepare the target state $\rho$ on the first system.
\item Evolve the two systems jointly using $e^{-iH\Delta t}$.
\end{enumerate}

\begin{theorem}[Validation of HME]\label{theorem:maintheorem}
For a short time period $\Delta t$, we have
\begin{equation}\label{eq:maineq}
\Tr_1\left(e^{-iH\Delta t}(\rho\otimes\sigma)e^{iH\Delta t}\right)=e^{-i\mathcal{N}(\rho)\Delta t}\sigma e^{i\mathcal{N}(\rho)\Delta t}+\mathcal{O}(\Delta t^2).
\end{equation}
Here, $\Tr_1$ denotes the partial trace over the first system, $\mathcal{N}$ represents the target Hermitian-preserving map, $H=\Lambda_{\mathcal{N}}^{T_1}$ with $\Lambda_{\mathcal{N}}=(\mathcal{I}\otimes\mathcal{N})\ketbra{\Phi^+}{\Phi^+}$ being the Choi matrix for $\mathcal{N}$, $\ket{\Phi^+}=\sum_i\ket{ii}$ denotes the unnormalized maximally entangled state, and $T_1$ represents the partial transposition operation on the first system.
\end{theorem}

\begin{proof}
Substituting the Taylor expansion $e^{-iH\Delta t}=\mathbb{I}-iH\Delta t+\mathcal{O}(\Delta t^2)$ into Eq.~\eqref{eq:maineq}, the left-hand side becomes
\begin{equation}\label{eq:proof_of_main_left_hand_side}
\sigma -i\Delta t\Tr_1\big(\left[H,\rho\otimes\sigma\right]\big) + \mathcal{O}(\Delta t^2).
\end{equation}
According to the Choi–Jamiołkowski isomorphism \cite{choi1975cp}, the resulting state of a map $\mathcal{N}$ can be represented using the Choi matrix $\Lambda_\mathcal{N}$, as $\mathcal{N}(\rho)=\Tr_1\left(\Lambda_\mathcal{N}(\rho^T\otimes\mathbb{I})\right)$. Following this definition, we can rewrite the coefficient of the first-order term in Eq.~\eqref{eq:proof_of_main_left_hand_side} as
\begin{equation}
\Tr_1\big(\left[H,\rho\otimes\sigma\right]\big)=\left[\Tr_1\left(\Lambda_{\mathcal{N}}(\rho^T\otimes\mathbb{I})\right),\sigma\right],
\end{equation}
which by definition equals to $[\mathcal{N}(\rho),\sigma]$. It can be verified that $[\mathcal{N}(\rho),\sigma]$ is also the coefficient of the first-order term of $e^{-i\mathcal{N}(\rho)\Delta t}\sigma e^{i\mathcal{N}(\rho)\Delta t}$. To provide a clearer derivation, we present a graphical demonstration of this proof based on tensor network representation in Fig.~\ref{fig:proof}(b).

Since $H=\Lambda_\mathcal{N}^{T_1}$ and $H$ is the Hamiltonian of the composite system, $\Lambda_{\mathcal{N}}$ must be a Hermitian operator. Therefore, the only restriction on $\mathcal{N}$ is that it should be Hermitian-preserving.
\end{proof}

\begin{figure}[t]
\centering
\includegraphics[width=0.48\textwidth]{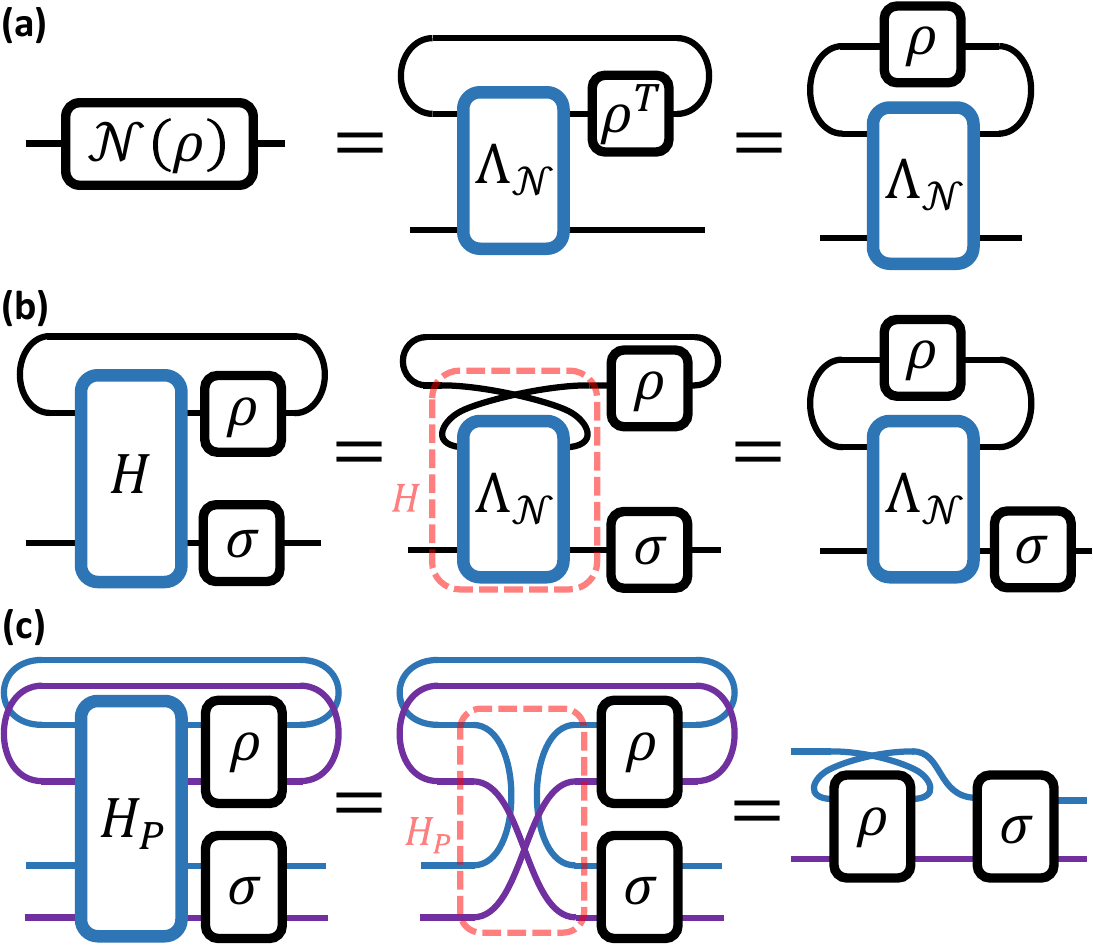}
\caption{Proof of Theorem~\ref{theorem:maintheorem} using tensor network calculations. A matrix is represented using a box with left and right legs, where the left legs represent row indices and the right legs represent column indices. The connection of legs indicates the contraction between indices. In (a), we illustrate the tensor network representation of the Choi matrix and how it can be used to represent the action of a linear map. In (b), we use the exchange of the left and right legs to represent the transposition operation. Consequently, it becomes evident that in order to ensure $\Tr_1[H(\rho\otimes\sigma)]=\mathcal{N}(\rho)\sigma$, we must have $H=\Lambda_\mathcal{N}^{T_1}$, which is highlighted by the red dashed box. In (c), we graphically demonstrate the validation of $H_P=\Phi^+_A\otimes S_B$ for realizing the evolution of $e^{-i\rho^{T_A}t}$. We use blue and purple legs to represent the indices of subsystems $A$ and $B$, respectively.}
\label{fig:proof}
\end{figure}

An illustrative example is the partial transposition map $e^{-i\rho^{T_A}t}$, where $\rho$ is a bipartite state with two subsystems $A$ and $B$. According to Theorem~\ref{theorem:maintheorem}, the corresponding Hamiltonian for HME is given by $H_P=\Phi^+_A\otimes S_B$, depicted within the red dashed box in Fig.~\ref{fig:proof}(c). the two blue half circles within the box represent the unnormalized maximally entangled state $\Phi^+_A$, while the cross of purple lines represents the SWAP operator $S_B$. We can deduce the validity of $\Tr_1\big(\left[H_P,\rho\otimes\sigma\right]\big)=[\rho^{T_A},\sigma]$ based on the connection rule of the legs. As a typical Hermitian-preserving map, the partial transposition map will be frequently used in the following discussion.

We note that the Hamiltonian used in HME is not the Choi matrix of $\mathcal{N}$, but the partial transposition of the Choi matrix. An intuitive way to grasp the role of the partial transposition operation in the Hamiltonian is by considering it as a result of concatenating the identity map. By rewriting $\mathcal{N}(\rho)=\mathcal{N}\circ\mathcal{I}(\rho)$, we can interpret HME as a process of first encoding $\rho$ into the evolution using the identity map $\mathcal{I}$, and then applying the map $\mathcal{N}$ to it. Since the state-encoding procedure is equivalent to the density matrix exponentiation algorithm, the Hamiltonian for implementing the identity map is the SWAP operator, which can also be derived using Theorem~\ref{theorem:maintheorem}. Then, to perform the map $\mathcal{N}$ on the encoded state, we concatenate the Choi matrix of $\mathcal{N}$ to the SWAP operator, resulting in the final Hamiltonian of HME, $\Lambda_\mathcal{N}^{T_1}$. The concatenation with the SWAP operator results in the partial transposition of $\Lambda_{\mathcal{N}}$, which can be visualized using the red dashed box in Fig.~\ref{fig:proof}(b), where the cross of the black lines above $\Lambda_{\mathcal{N}}$ represents the SWAP operator.

\section{Performance Analysis}\label{sec:performances}

As a quantum algorithm, we thoroughly analyze the performance of HME, including its sample complexity upper bound and robustness against various errors. Furthermore, we provide a fundamental lower bound of the sample complexity for exponentiating a Hermitian-preserving map and show the optimality of HME in certain cases.

\subsection{The Upper Bound of Sample Complexity}\label{subsec:complexity_upper_bound}
HME realizes the evolution of $e^{-i\mathcal{N}(\rho)t}$ by sequentially inputting the target state $\rho$. Thus, an essential indicator for analyzing the performance of HME is the number of copies needed for realizing the desired evolution with an error up to $\epsilon$. According to Theorem~\ref{theorem:maintheorem}, the difference between the ideal and real channels for a single step of the experiment is a second-order term. Thus the error could be suppressed by choosing a shorter time slice $\Delta t$, or equivalently, by using more copies of $\rho$. 

\begin{theorem}[Upper bound of sample complexity]\label{theorem:cost}
Let $\mathcal{N}$ be an arbitrary Hermitian-preserving map, the HME algorithm shown in Fig.~\ref{fig:overview}(c) requires at most $\mathcal{O}\left(\epsilon^{-1}\norm{H}_{\infty}^2t^2\right)$ copies of sequentially inputting state $\rho$ to ensure that $\norm{\mathcal{Q}_t-\mathcal{U}_t}_\diamond\le\epsilon$ holds for arbitrary $\rho$. Here, $H=\Lambda_{\mathcal{N}}^{T_1}$, $\mathcal{Q}_t=\mathcal{Q}_{\Delta t}^{\circ K}$ represents the real channel with $\mathcal{Q}_{\Delta t}(\sigma):=\Tr_1\left[e^{-iH\Delta t}(\rho\otimes\sigma)e^{iH\Delta t}\right]$ and $K=t/\Delta t$, $\mathcal{U}_t$ is the ideal evolution channel corresponding to $e^{-i\mathcal{N}(\rho)t}$, and $\norm{\cdot}_\diamond$ denotes the diamond norm.
\end{theorem}

Here, we provide a sketch of the main idea for the proof of Theorem~\ref{theorem:cost}. The complete proof is technically involved and thus left to Appendix~\ref{app:Proof_of_Thm2}. We divide the ideal evolution channel $\mathcal{U}_t$ into $K$ slices $\mathcal{U}_t=\mathcal{U}_{\Delta t}^{\circ K}$, where $\mathcal{U}_{\Delta t}$ refers to the unitary evolution of $e^{-i\mathcal{N}(\rho)\Delta t}$. Intuitively, the closeness of $\mathcal{Q}_{\Delta t}$ and $\mathcal{U}_{\Delta t}$ ensures the closeness of $\mathcal{Q}_{t}$ and $\mathcal{U}_{t}$. According to the subadditivity property of the diamond distance (see Lemma~\ref{lemma:subadditivity} in Appendix~\ref{app:Auxiliary_Lemmas}), we obtain
\begin{equation}
\norm{\mathcal{Q}_t-\mathcal{U}_t}_\diamond\le K\norm{\mathcal{Q}_{\Delta t}-\mathcal{U}_{\Delta t}}_\diamond,
\end{equation}
which implies that the error accumulates linearly during the sequential operations. By employing certain matrix inequalities, we arrive at 
\begin{equation}
\norm{\mathcal{U}_{\Delta t}-\mathcal{Q}_{\Delta t}}_\diamond\le\mathcal{O}\left(\norm{H}_{\infty}^2\Delta t^2\right)=\mathcal{O}\left(\norm{H}_{\infty}^2t^2/K^2\right).
\end{equation}
Hence, the total diamond distance scales as $\mathcal{O}\left(\norm{H}_{\infty}^2t^2/K\right)$. Therefore, in order to ensure that the total diamond distance is less than $\epsilon$, the number of steps $K$, or equivalently, the number of copies of $\rho$, needs to be at most $\mathcal{O}\left(\epsilon^{-1}\norm{H}_{\infty}^2t^2\right)$.

According to the definition of the diamond norm, Theorem~\ref{theorem:cost} implies that for any state $\sigma_{\mathcal{K}\mathcal{R}}\in\mathcal{K}\otimes\mathcal{R}$, where $\mathcal{K}$ is the Hilbert space on which the channels $\mathcal{Q}_t$ and $\mathcal{U}_t$ are defined, and $\mathcal{R}$ is a reference system with arbitrary dimension, if the number of steps in the HME algorithm shown in Fig.~\ref{fig:overview}(c) satisfies $K=\Theta\left(\epsilon^{-1}\norm{H}_{\infty}^2t^2\right)$, then
\begin{equation}
\norm{\left(\mathcal{U}_t-\mathcal{Q}_t\right)\otimes \mathcal{I}_\mathcal{R}(\sigma_{\mathcal{K}\mathcal{R}})}_1\le \epsilon.
\end{equation}
This implies that states resulting from the real and ideal evolutions are close in terms of trace distance. Consequently, when measuring the real post-evolution state, the measurement outcomes will be close to those obtained from the ideal post-evolution state. This property will be frequently utilized in our analysis. Another crucial insight provided by Theorem~\ref{theorem:cost} is that the complexity of HME highly depends on the operator norm of the evolution Hamiltonian. This finding enables us to identify scenarios that are suitable for HME and is instrumental in analyzing the sample complexities for applications we will explore in subsequent discussions.

By specializing $\mathcal{N}=\mathcal{I}$ and $H=S$, the sample complexity upper bound provided by Theorem~\ref{theorem:cost} recovers the previously obtained upper bound of $\mathcal{O}\left(\epsilon^{-1}t^2\right)$ for density matrix exponentiation \cite{Lloyd2014qpca, Kimmel2017hamiltonian}, as $\norm{S}_{\infty}=1$. This upper bound has been proven to be tight \cite{Kimmel2017hamiltonian}. While the upper bound $\mathcal{O}\left(\epsilon^{-1}\norm{H}_{\infty}^2t^2\right)$ provided by Theorem~\ref{theorem:cost} holds for all Hermitian-preserving maps, it may not be a tight bound in certain cases. For instance, when considering the partial transposition map, we observe that the sample complexity of implementing it with HME can be lower than what is predicted by Theorem~\ref{theorem:cost}, based on the specific properties of the partial transposition map. We provide the detailed analysis in Appendix~\ref{app:proof_of_err:PT} and will employ this result in the application of entanglement quantification.

In practical quantum information processing tasks, the HME algorithm is often employed in conjunction with other quantum algorithms such as quantum phase estimation or Hadamard test. These algorithms typically require a controlled version of the evolution $e^{-i\mathcal{N}(\rho)t}$, as illustrated in the circuits of Fig.~\ref{fig:ent_det}, Fig.~\ref{fig:QNSR}, and Fig.~\ref{fig:obs_circuit}. Thus, to analyze the performance of HME in practical applications, it is crucial to determine the cost of realizing a controlled-$e^{-i\mathcal{N}(\rho)t}$ evolution.

It is indeed straightforward to generalize Theorem~\ref{theorem:cost} to the controlled evolution $\mathrm{C}\text{-}e^{-i\mathcal{N}(\rho)t}:=\ketbra{0}{0}_c\otimes \mathbb{I}+\ketbra{1}{1}_c\otimes e^{-i\mathcal{N}(\rho)t}$,
where the subscript $c$ denotes the control qubit. The core idea of this generalization is based on the equality
\begin{equation}
\mathrm{C}\text{-}e^{-i\mathcal{N}(\rho)t}=e^{-i\ketbra{1}{1}_c\otimes\mathcal{N}(\rho)t},
\end{equation}
by which we can define a new Hermitian-preserving map $\mathcal{N}^\prime(\rho)=\ketbra{1}{1}_c\otimes\mathcal{N}(\rho)$.
Therefore, in the circuit of HME depicted in Fig.~\ref{fig:overview}(c), we can add a control qubit and regard the control qubit together with $\sigma$ as the evolved state of a new HME circuit, with a new Hamiltonian
$H^\prime=\ketbra{1}{1}_c\otimes H$, according to Theorem~\ref{theorem:maintheorem}.
Since $\ketbra{1}{1}_c\otimes H$ has the same operator norm as $H$, according to Theorem~\ref{theorem:cost}, we have:

\begin{corollary}\label{corollary:cost(controlled_version)}
We need at most $\mathcal{O}\left(\epsilon^{-1}\norm{H}_{\infty}^2t^2\right)$ copies of $\rho$ to perform $\mathrm{C}\text{-}e^{-i\mathcal{N}(\rho)t}$ to precision $\epsilon$ in diamond distance, where $H=\Lambda_{\mathcal{N}}^{T_1}$.
\end{corollary}

Based on the same reasoning, Corollary~\ref{corollary:cost(controlled_version)} can be straightforwardly extended to scenarios where the number of control qubits is greater than one.

\subsection{Robustness}\label{subsec:robustness}

In practical situations, the performance of HME can be affected by various factors. For example, quantum devices may exhibit unpredictable and unavoidable noise, and the realization of the Hamiltonian evolution $e^{-iH\Delta t}$ may be inaccurate due to approximation errors in Hamiltonian simulation \cite{suzuki1990fractal,suzuki1991fractal}. Hence, the robustness of HME is another important indicator that needs in-depth analysis. We mainly consider two kinds of errors, the input state error, and the Hamiltonian error. These errors can lead to biases and variations in the input states and the evolution Hamiltonians at each step. The error of HME caused by these factors can be bounded by:

\begin{theorem}[Robustness]\label{theorem:robustness}
Suppose the input states and Hamiltonians used in Fig.~\ref{fig:overview}(c) vary at each step, denoted as $\{\rho_k^\prime\}_{k=1}^K$ and $\{H_k^\prime\}_{k=1}^K$. If the number of steps satisfies $K\ge2t\norm{H}_{\infty}$ and $K\ge2t\norm{H^\prime_k}_{\infty}$ for $k=1,\cdots,K$,
then the diamond distance between the noisy channel $\mathcal{Q}_t^\prime$ constructed using $\{\rho_k^\prime\}_{k=1}^K$, and $\{H_k^\prime\}_{k=1}^K$ and the noiseless channel $\mathcal{Q}_t$ constructed using fixed $\rho$ and $H$, is bounded by
\begin{equation}\label{eq:robustness}
\norm{\mathcal{Q}^{\prime}_{t}-\mathcal{Q}_{t}}_{\diamond}\le 4t\left(D_H+\norm{H}_{\infty}D_S\right),
\end{equation}
where
$D_H:=\frac{1}{K}\sum_{k=1}^K\norm{H^{\prime}_k-H}_{\infty}$ and $D_S:=\frac{1}{K}\sum_{k=1}^K\norm{\rho^{\prime}_k-\rho}_1$ are the average divergences between the Hamiltonians and input states, respectively.
\end{theorem}

In practice, the difference between $\norm{H}_\infty$ and $\norm{H_k^\prime}_\infty$ is relatively small and the condition $\norm{H}_\infty t\ge 1$ is usually satisfied. Thus, if we choose $K=\Theta\left(\epsilon^{-1}\norm{H}_\infty^2t^2\right)$, as requires in Theorem~\ref{theorem:cost}, then the requirement of $K\ge2t\norm{H}_{\infty}$ and $K\ge2t\norm{H^\prime_k}_{\infty}$ for $k=1,\cdots,K$ 
can be naturally satisfied. The proof of the above theorem can be found in Appendix~\ref{app:proof_of_robustness}. In addition to the diversities of Hamiltonians and input states, the error also depends on the norm of the evolution Hamiltonian. According to Eq.~\eqref{eq:robustness}, $\norm{H}_\infty$ acts as a multiplier that amplifies the error caused by the diversity of input states. Similar to the sample complexity upper bound, we can derive a better estimation of the robustness of certain maps by leveraging their specific properties. For instance, in the case of the partial transposition map, the robustness is not related to $\norm{H_P}_\infty$. The proof of this specific case is provided in Appendix~\ref{app:proof_of_robustness:PT}.

\subsection{Analysis of Optimality}\label{subsec:optimality}

After analyzing the sample complexity upper bound and robustness of HME, another important question arises: is HME the optimal protocol for implementing a Hermitian-preserving map? Since $\mathcal{N}(\rho)$ is not always a physical state, the task of preparing $\mathcal{N}(\rho)$ is generally infeasible. Therefore, we aim to determine whether HME is the protocol for exponentiating specific Hermitian-preserving maps with asymptotically optimal sample complexity.

\begin{figure}[t]
\centering
\includegraphics[width=0.48\textwidth]{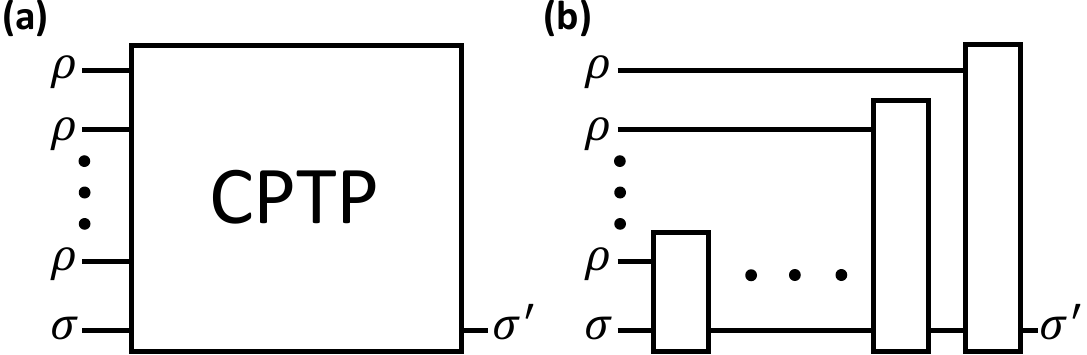}
\caption{General and sequential protocols to exponentiate a Hermitian-preserving map. (a) 
In the most general case, joint operations are performed on multiple copies of $\rho$ and the input state $\sigma$, which induces a channel $\sigma^{\prime}=\mathcal{Q}^{\rho}(\sigma)$ that acts only on $\sigma$. (b) A sequential protocol relies on sequential operations acting on single-copies of $\rho$ and the evolved state, where the black boxes represent distinct CPTP maps.}
\label{fig:LowerBoundModel}
\end{figure}

The most general approach to realize $e^{-i\mathcal{N}(\rho)t}$ is depicted in Fig.~\ref{fig:LowerBoundModel}(a). It involves the joint evolution of multiple copies of $\rho$ and a single copy of $\sigma$ under a quantum channel that is independent of $\rho$ and $\sigma$. The resulting output state $\sigma^\prime$ has the same dimension as $\sigma$. When focusing on the state $\sigma$, we can incorporate the multiple copies of $\rho$ into the quantum channel, resulting in a smaller channel that acts only on $\sigma$, denoted as $\mathcal{Q}^\rho$. When a sufficient number of $\rho$ is available, it is possible to select the original quantum channel in such a way that, for any choice of $\rho$, the induced evolution $\mathcal{Q}^\rho$ and the target evolution $e^{-i\mathcal{N}(\rho)t}$ are sufficiently close to each other. In the following analysis, we establish a fundamental lower bound for the sample complexity in this general setting.

In practice, performing joint operations on multiple copies of $\rho$ can be extremely challenging. Protocols that utilize sequential operations on each copy of $\rho$ along with a quantum memory, as depicted in Fig.~\ref{fig:LowerBoundModel}(b), can significantly alleviate the challenges faced by quantum devices. Since HME also adopts sequential operations, another essential question is whether HME is the most efficient sequential protocol to achieve the evolution of $e^{-i\mathcal{N}(\rho)t}$. Notice that sequential protocols are special cases of general protocols, a sample complexity lower bound derived for general protocols also serves as a lower bound for the complexity of sequential ones.

\begin{theorem}[Lower bound of sample complexity]\label{theorem:LowerBound}
Let $\mathcal{N}\in T(\mathcal{H},\mathcal{K})$ be a Hermitian-preserving map and set $0<\epsilon\le1/6$, and $t\ge \frac{15\pi\epsilon}{4R_*}$. The minimum number of $\rho$ needed to realize the evolution of $e^{-i\mathcal{N}(\rho)t}$ using protocols shown in Fig.~\ref{fig:LowerBoundModel}(a) with $\epsilon$ accuracy in diamond distance satisfies
\begin{equation}
f_{\mathcal{N}}(\epsilon,t)\ge\Omega\left(\epsilon^{-1}R_*^2t^2\right),
\end{equation}
where $R_*:=\underset{A\in \mathscr{F}}{\max}\;R\left[\mathcal{N}(A)\right]$.  $R[\cdot]=\lambda_{\text{max}}(\cdot)-\lambda_{\text{min}}(\cdot)$ denotes the spectral gap defined as the difference between the largest and the lowest eigenvalues of the processed matrix. The feasible region $\mathscr{F}$ is defined as
\begin{equation}\label{eq:feasible region}
\begin{aligned}
\mathscr{F}=\Big\{A\in L(\mathcal{H}):A^{\dagger}=A,\Tr(A)=0,\\
\norm{A}_1=1,\left[\mathcal{N}(A^+),\mathcal{N}(A^-)\right]=0\Big\},
\end{aligned}
\end{equation}
where $A^+$ and $A^-$ are the positive and negative parts of $A$.
\end{theorem}

We leave the in-depth proof of this theorem in Appendix~\ref{app:proof_of_LB}, and sketch the main idea here. The key principle for proving this theorem is the Holevo-Helstrom theorem \cite{Holevo1973decision,Helstrom1969detection,watrous2018theory}, which deals with the discriminations of states and channels. According to this theorem, the optimal success probability for discriminating two different states is $\frac{1}{2}+\frac{1}{4}\norm{\rho_0-\rho_1}_1$, and for discriminating two different channels is $\frac{1}{2}+\frac{1}{4}\norm{\mathcal{Q}_0-\mathcal{Q}_1}_\diamond$. When allowed to perform joint operations on $K$ copies of the unknown state $\rho$, discrimination becomes easier with more copies, as $\norm{\rho_0^{\otimes K}-\rho_1^{\otimes K}}_1$ increases with $K$.

For a given Hermitian-preserving map $\mathcal{N}$, there may exist two states $\rho_0$ and $\rho_1$ such that the trace distance between them can be much lower than the diamond distance between the two ideal evolution channels, $\norm{\rho_0-\rho_1}_1\ll\norm{[e^{-i\mathcal{N}(\rho_0) t}]-[e^{-i\mathcal{N}(\rho_1) t}]}_{\diamond}$, for some $t$. According to the Holevo-Helstrom theorem, this implies that the discrimination of these two states is more challenging compared to distinguishing the ideal evolution channels. Thus, if the induced channel of Fig.~\ref{fig:LowerBoundModel}(a), denoted as $\mathcal{Q}^\rho$, can approximate the ideal channel $[e^{-i\mathcal{N}(\rho)t}]$ for all states $\rho$, the discrimination between the channels $\mathcal{Q}^{\rho_0}$ and $\mathcal{Q}^{\rho_1}$ will also be easier. Since the CPTP map used in Fig.~\ref{fig:LowerBoundModel}(a) is independent of the input state $\rho$, the discrimination between the channels $\mathcal{Q}^{\rho_0}$ and $\mathcal{Q}^{\rho_1}$ guarantees the discrimination between the states $\rho_0$ and $\rho_1$. This implies that there exists a fundamental lower bound on the number of copies for $\rho$ required to realize $\mathcal{Q}^{\rho}$, given by $\norm{\rho_0^{\otimes K}-\rho_1^{\otimes K}}_1\ge \norm{\mathcal{Q}^{\rho_0}-\mathcal{Q}^{\rho_1}}_\diamond$, which serves as the key equation for deriving the lower bound on $K$. The choices of $\rho_0$ and $\rho_1$ are related to the feasible region $\mathscr{F}$, whose properties are discussed in Appendix~\ref{app:proof_of_LB}.

As mentioned earlier, Theorem~\ref{theorem:LowerBound} provides a lower bound on the sample complexities for general protocols, which might not be tight for sequential protocols. However, we find that for certain important Hermitian-preserving maps, the upper bound of HME, as given in Theorem~\ref{theorem:cost}, can indeed reach the lower bound derived for general protocols. This demonstrates the optimality of HME for specific tasks.

As an illustrative example, we consider the identity map $\mathcal{N}=\mathcal{I}$ again. According to the definition of the feasible region in Eq.~\eqref{eq:feasible region}, it is straightforward to verify that $\frac{1}{2}(\ketbra{\alpha}{\alpha}-\ketbra{\beta}{\beta})$, where $\ket{\alpha}$ and $\ket{\beta}$ are arbitrary orthonormal pure states, is a feasible point that maximizes $R[\mathcal{N}(A)]=R[A]$ with $R_*=1$. Thus, based on Theorem~\ref{theorem:LowerBound}, the lower bound for realizing the evolution of $e^{-i\rho t}$ is $\Omega(\epsilon^{-1}t^2)$, which meets the upper bound of the complexity of HME in Theorem~\ref{theorem:cost}. This fact demonstrates that Theorem~\ref{theorem:LowerBound} covers the sample complexity lower bound given in Ref.~\cite{Kimmel2017hamiltonian}. Another example is closely related to the quantum noiseless state recovery discussed in Sec.~\ref{sec:QNSR}, where we consider inverting the local amplitude damping noise. In this case, the non-physical map we aim to implement is $\mathcal{N}_{\gamma,n}=(\mathcal{E}_\gamma^{\otimes n})^{-1}$, where $n$ denotes the number of qubits and $\gamma$ is the damping rate. We show in Appendix~\ref{app:LowerBound_AD} that the lower bound of the sample complexity for implementing $\mathcal{N}_{\gamma,n}$ matches the upper bound of HME obtained in Theorem~\ref{theorem:cost}.

It is worth mentioning that the analysis of the upper and lower bounds for realizing $e^{-i\mathcal{N}(\rho)t}$ also partly answers the open problem raised in Ref.~\cite{anshu2023survey}. This problem seeks to determine the sample complexity required for implementing the evolution of $e^{-if(\rho)t}$ for a given function $f(\cdot)$.

\section{Entanglement Detection and Quantification}\label{sec:ent_det}

As mentioned before, positive but not completely positive maps are important tools for entanglement detection and quantification \cite{gunhe2009entanglement}. In this section, we will demonstrate how the HME algorithm enhances these tasks via reduction and transposition maps, which are Hermitian-preserving but non-physical.

\subsection{Exponential Advantage in Entanglement Detection}\label{subsec:ent_det}

Entanglement is a distinctive feature of quantum mechanics \cite{horodechi2009entanglement}, and the detection and quantification of entanglement play vital roles in quantum information processing tasks and theoretical studies \cite{pan2012multiphoton,Pasquale2004qft}. A variety of protocols have been proposed for entanglement detection \cite{gunhe2009entanglement}. In the literature, the positive map criteria, which include the partial transposition \cite{peres1996ppt} and reduction maps \cite{Horodecki1999reduction}, are considered to be among the most powerful entanglement detection criteria. Compared to widely used entanglement witness protocols \cite{gunhe2009entanglement}, positive map criteria typically exhibit exponential enhancements in detection capability \cite{liu2022fundamental,collins2015random}. Positive map criteria state that if $\mathcal{P}_A\otimes\mathcal{I}_B(\rho)$ is not a positive matrix, then the state $\rho$ is entangled. Here, $\mathcal{P}$ represents the positive map being used, and $A$ and $B$ label the two subsystems of $\rho$.  However, the main challenges in achieving positive map detection protocols lie in the non-physical nature of positive maps as well as the difficulty of performing spectral decomposition for the exponentially large matrix $\mathcal{P}_A\otimes\mathcal{I}_B(\rho)$. These two obstacles can be overcome using the algorithms of HME and quantum phase estimation.

\begin{figure*}[t]
\centering
\includegraphics[width=0.99\textwidth]{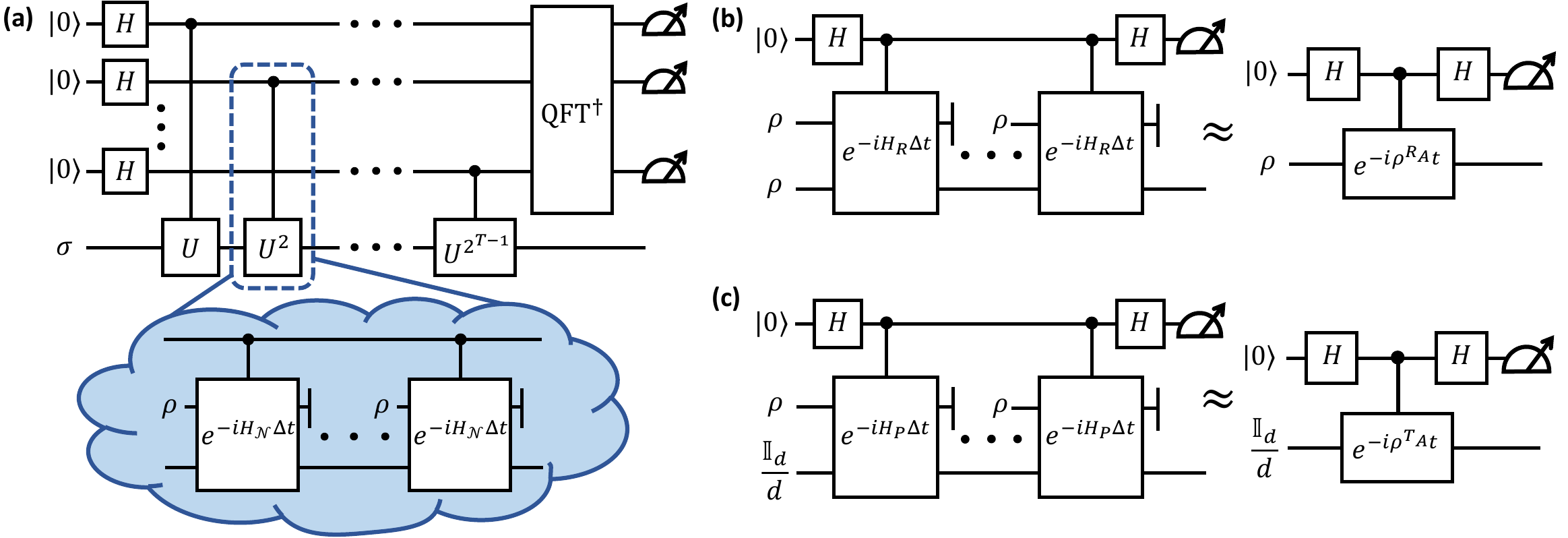}
\caption{(a) represents the circuit of the HME-based entanglement detection protocol, which is constructed by combining the HME and quantum phase estimation algorithms. The quantum phase estimation consists of the input state $\sigma$, $T$ ancilla qubits, controlled unitary operations, inverse quantum Fourier transformation, and final measurements on the ancilla qubits. The controlled unitary evolutions are approximately realized by HME, where $H_{\mathcal{N}}$ is the Hamiltonian for implementing $\mathcal{N}=\mathcal{P}_A\otimes\mathcal{I}_B$. (b) is a special case of (a) where only one ancilla qubit remains, the inverse quantum Fourier transformation is reduced to the Hadamard gate, and $H_R$ is the Hamiltonian for implementing the partial reduction map. The evolved state and sequentially inputted states are set to the target state of the entanglement detection task. (c) is the circuit for estimating entanglement negativity. Compared with (b), we change the evolved state to the maximally mixed state and the Hamiltonian to $H_P$.}
\label{fig:ent_det}
\end{figure*}

To facilitate our subsequent discussions, we now give a brief introduction of the functions and properties of the quantum phase estimation algorithm, which is employed to analyze the eigenvalue, $e^{i\varphi}$, and eigenstate of a given unitary $U$ \cite{kitaev1995quantum,nielsen2010quantum}. As depicted in Fig.~\ref{fig:ent_det}(a), the number of ancilla qubits is determined by the desired estimation accuracy of target eigenvalues. While in some special cases when the binary representation of $\varphi/2\pi$ has a finite length, a finite number of ancilla qubits can estimate $\varphi$ accurately. If a sufficient number of ancilla qubits is employed and the input state $\sigma$ is one of the eigenstates of $U$, measuring the ancilla qubits will yield the corresponding eigenvalue of $\sigma$. While, if $\sigma$ is not an eigenstate of $U$, the measurement of the ancilla qubits will randomly output one of the eigenvalues of $U$ and $\sigma$ will collapse into the corresponding eigenstate, assuming no degeneracy. The probability of obtaining some eigenvalue is proportional to the fidelity between $\sigma$ and the corresponding eigenstate.

By utilizing HME to implement controlled unitary operations, $e^{-i\mathcal{P}_A\otimes\mathcal{I}_B(\rho)t}$, in quantum phase estimation, we propose an HME-based entanglement detection protocol, as illustrated in Fig.~\ref{fig:ent_det}(a). Exploiting the properties of quantum phase estimation algorithm, the measurements of the ancilla qubits directly provide information about the spectrum of $U=e^{-i\mathcal{P}_A\otimes\mathcal{I}_B(\rho)t}$, which is equivalent to the spectrum of $\mathcal{P}_A\otimes\mathcal{I}_B(\rho)$. If $\mathcal{P}_A\otimes\mathcal{I}_B(\rho)$ is not positive, we can select $\sigma$ to be a state that exhibits significant overlap with the negative subspace of $\mathcal{P}_A\otimes\mathcal{I}_B(\rho)$. Consequently, there is a significant probability of detecting the negative eigenvalues of $\mathcal{P}_A\otimes\mathcal{I}_B(\rho)$ and thus the entanglement of $\rho$. In the following, we will demonstrate that certain positive maps facilitate the selection of an appropriate $\sigma$ due to their properties.

Although we have many tools to detect entanglement, it has been proven that when lacking the prior knowledge of the target states, an exponential amount of resources is still required for entanglement detection \cite{liu2022fundamental}, even in the case of pure states \cite{2023complexity}. The following presents a typical example \cite{2023complexity}.

\begin{fact}\label{fact:ent_det}
Consider a bipartite system with subsystem dimensions $d_A=d_B=\sqrt{d}$. Let $\rho$ be a pure state, which can be either a global Haar random state $\psi_{AB}$ or a tensor product of local Haar random pure states $\psi_A\otimes\psi_B$, each with equal probabilities. If we are limited to single-copy operations and measurements on $\rho$, determining whether $\rho$ is entangled requires $\Theta(d^{1/4})$ experiments, with a success probability of at least $2/3$.
\end{fact}
Here we employ the reduction map to demonstrate the advantage of HME in entanglement detection. The reduction map $R$ is defined as $\rho^R=\Tr(\rho)\mathbb{I}_d-\rho$, where $d=d_A\times d_B$ represents the dimension of the Hilbert space of $\rho$. The reduction criterion involves testing whether
\begin{equation}
\rho^{R_A}=\mathbb{I}_{d_A}\otimes\rho_B-\rho
\end{equation}
is a positive matrix or not, where $\rho_B=\Tr_A(\rho)$ denotes the reduced density matrix for subsystem $B$. While the reduction criterion may not be the strongest positive map criterion \cite{Horodecki1999reduction}, we find that it is well-suited for the HME-based entanglement detection protocol. This is because $\rho^{R_A}$ effectively captures the principal components of $\rho$. Through the implementation of the partial reduction map using HME, we find that the entanglement detection task implied in Fact~\ref{fact:ent_det} can be solved with a constant sample complexity, showing an exponential advantage compared with all single-copy entanglement detection protocols.

\begin{theorem}\label{theorem:ED_HME+Reduction}
Let $H_R=\mathbb{I}_A\otimes S_B-S_{AB}$, where $\mathbb{I}_A$, $S_B$, and $S_{AB}$ stand for the identity operator acting on two copies of system $A$, SWAP operator acting on two copies of system $B$, and SWAP operator acting on two copies of the joint system $AB$, respectively. Choosing $t=\pi$, the circuit depicted in Fig.~\ref{fig:ent_det}(b) requires at most $\mathcal{O}(1)$ copies of $\rho$ to accomplish the task described in Fact~\ref{fact:ent_det} with a success probability of at least $2/3$.
\end{theorem}

We leave the formal proof in Appendix~\ref{app:ED_HME+Reduction}, and discuss the intuition here. In the HME-based entanglement detection protocol depicted in Fig.~\ref{fig:ent_det}(a), the sample complexity is determined by three main factors: the number of copies of $\rho$ required to realize a single controlled unitary using HME, the number of repetitions of the controlled unitary operations, and the number of repetitions of the entire circuit to obtain the desired measurement outcomes. The constant sample complexity given in Theorem~\ref{theorem:ED_HME+Reduction} can be explained based on the following reasons. First, since $\norm{H_R}_\infty=2$, which remains constant as the system size grows, the cost of implementing the controlled-$e^{-i\rho^{R_A} t}$ evolution does not scale with system dimension, as indicated by Corollary~\ref{corollary:cost(controlled_version)}.
Second, a single controlled unitary evolution is sufficient. When $\rho=\psi_A\otimes\psi_B$, the lowest eigenvalue of $\rho^{R_A}=\mathbb{I}_{d_A}\otimes\psi_B-\psi_A\otimes\psi_B$ is zero. On the other hand, when $\rho=\psi_{AB}$ is a global random pure state, then the lowest eigenvalue of $\rho^{R_A}=\mathbb{I}_{d_A}\otimes\rho_B-\psi_{AB}\sim\frac{1}{d_B}\mathbb{I}_d-\psi_{AB}$ is approximately $-1$, especially for large $d$. Therefore, by setting $t=\pi$, the eigenvalues of $e^{-i(\psi_{AB})^{R_A}t}$ and $e^{-i(\psi_{A}\otimes\psi_{B})^{R_A}t}$ corresponding to the lowest eigenvalues of $(\psi_{AB})^{R_A}$ and $(\psi_{A}\otimes\psi_{B})^{R_A}$ are approximately $e^{i\pi}$ and $e^{i0}$, respectively. As the binary representations of $\pi/2\pi$ and $0/2\pi$ are $0.1$ and $0$, a single ancilla qubit, which is equivalent to a single controlled unitary operation in quantum phase estimation, is sufficient for telling these two different eigenvalues. Finally, one round of the experiment is sufficient. This is because $\rho$ itself has a significant overlap with the eigenspace corresponding to the lowest eigenvalue of $\rho^{R_A}$ in both cases. Therefore, by setting the evolved state $\sigma$ to be $\rho$, as depicted in Fig.~\ref{fig:ent_det}(b), one round of experiment can detect the entanglement with a significantly high probability.

\subsection{Negativity Estimation}\label{subsec:neg}

Among all positive map criteria, the positive partial transposition criterion \cite{peres1996ppt} is of fundamental importance due to its strong detection capability \cite{aubrun2012partial}, its connection to entanglement distillation \cite{Horodecki1998distillation}, and its concise mathematical formulation. The positive partial transposition criterion states that if the target state $\rho$ is separable, then the matrix $\rho^{T_A}$ has no negative eigenvalues, where $T_A$ denotes the partial transposition map on a subsystem $A$. The entanglement negativity serves as an entanglement measure by quantifying the violation of the positive partial transposition criterion,
\begin{equation}
N(\rho)=\frac{\norm{\rho^{T_A}}_1-1}{2},
\end{equation}
which equals to the sum of the absolute values of all negative eigenvalues of $\rho^{T_A}$. As a widely-used quantifier for mix-state entanglement, negativity plays a vital role in both quantum information science \cite{vidal2002negativity} and theoretical physics \cite{calabrese2012cft,lu2020negativity}.

However, due to the non-physical nature of the transposition map and the highly nonlinear behavior of negativity, there is still no efficient protocol for unbiasedly estimating negativity. Some approaches, such as directly measuring the spectral values of $\rho^{T_A}$ \cite{keyl2001spectrum,horodecki2002direct}, require highly complex joint operations that are extremely challenging for practical quantum devices. Tomography-based protocols estimate negativity by reconstructing the density matrix, which requires a large sample complexity and significant classical computational resources \cite{kueng2017low,chen2022tight}. Other attempts primarily utilize the moments of the partially transposed density matrix to infer properties of negativity \cite{gray2018machine,yu2021optimal}, rather than providing an unbiased estimation. Therefore, these methods cannot be directly employed for quantifying entanglement.

HME can also alleviate the difficulties in measuring negativity, from a similar perspective as in entanglement detection. First, as illustrated in Fig.~\ref{fig:proof}(c), HME can realize the evolution of $e^{-i\rho^{T_A}t}$ by setting the evolution Hamiltonian to $H_P=\Phi^+_A\otimes S_B$. A more complete calculation in Appendix~\ref{app:proof_of_err:PT} reveals that the complexity of achieving the evolution of $e^{-i\rho^{T_A}t}$ is $\mathcal{O}(\epsilon^{-1}d_At^2)$, which is lower than predicted by Theorem~\ref{theorem:cost} and scales only with the dimension of subsystem $A$. Second, as $e^{-i\rho t}$ is a nonlinear function of $\rho$, density matrix exponentiation enables the measurement of certain nonlinear quantities $f(\rho)=f(\{\lambda_i\}_i)$, which is a function of the eigenvalues of $\rho$, such as quantum entropy \cite{pichler2016cold,wang2023entropy}. HME further extends the measurement capability to a wider range of functions, $f[\mathcal{N}(\rho)]$, as the exponentiation is applied to $\mathcal{N}(\rho)$ rather than $\rho$.

We propose a negativity estimation protocol inspired by the Hadamard-test-like circuit in Ref.~\cite{knill1998power}, as depicted in Fig.~\ref{fig:ent_det}(c), which is originally designed for estimating the trace of a unitary. Following a similar approach, by setting the evolved state as the maximally mixed state and performing the controlled-$e^{-i\rho^{T_A}t}$ operation on the ancilla and evolved states using HME, we can estimate quantities like $\Tr[\cos(\rho^{T_A}t)]$ or $\Tr[\sin(\rho^{T_A}t)]$ depending on the measurement basis of the ancilla qubit. For instance, measuring the expectation value of the Pauli-$X$ operator on the ancilla qubit would yield
\begin{equation}
\begin{aligned}
&\Tr\left[(X\otimes\mathbb{I}_d) \mathrm{C}\text{-}e^{-i\rho^{T_A}t}\left(\ketbra{+}{+}\otimes\frac{\mathbb{I}_d}{d}\right)\mathrm{C}\text{-}e^{i\rho^{T_A}t}\right]\\
=&\frac{1}{2d}\Tr\left[e^{-i\rho^{T_A}t}+e^{i\rho^{T_A}t}\right]=\frac{1}{d}\Tr\left[\cos(\rho^{T_A}t)\right],
\end{aligned}
\end{equation}
where $\mathrm{C}\text{-}e^{-\rho^{T_A}t}=\ketbra{0}{0}\otimes\mathbb{I}_d+\ketbra{1}{1}\otimes e^{-\rho^{T_A}t}$.
The measurement of $\Tr\left[\cos(\rho^{T_A}t)\right]$ with different values of $t$ enables us to estimate the negativity, which can be expressed as a convergent Fourier expansion
\begin{equation}\label{eq:talor}
\norm{\rho^{T_A}}_1=\frac{\pi}{2}d-\sum_{l=1}^\infty\frac{4}{\pi(2l-1)^2}\Tr[\cos((2l-1)\rho^{T_A})].
\end{equation}

In particular, this HME-based negativity estimation algorithm repeats the following procedures for a total of $M$ times:
\begin{enumerate}
\item Randomly sample an integer $l$ with a probability distribution according to the coefficients of Eq.~\eqref{eq:talor}.
\item Run the circuit of Fig.~\ref{fig:ent_det}(c) for one shot with evolution time $t=(2l-1)$ and $K(l)$ copies of inputting states $\rho$ to produce an estimation of $\Tr\left[\cos\left((2l-1)\rho^{T_A}\right)\right]$ using the measurement result of the ancilla qubit.
\end{enumerate}
By taking the median of means of these sequentially generated estimations, we can get the final estimation of $N(\rho)$. The details of this algorithm can be found in Algorithm~\ref{algo:nega}. Thus, the sampling times, $M$, together with the number of states in a single-shot experiment, $K(l)$, constitute the total sample complexity, which is proved to be:

\begin{theorem}\label{theorem:cost_Ent_Neg}
Using Algorithm~\ref{algo:nega}, one needs an expected number of $\widetilde{\mathcal{O}}\left(\log(\delta^{-1})\epsilon^{-3}d^2d_A\norm{\rho^{T_A}}_1\right)$ copies of $\rho$ to ensure $\abs{\hat{N}(\rho)-N(\rho)}\le\epsilon$ with a probability of at least $1-\delta$. Here, the $\widetilde{\mathcal{O}}$ notation suppresses logarithmic expressions for $d$ and $\epsilon$. Besides, the sampling times scales as $M=\mathcal{O}\left(\log(\delta^{-1})\epsilon^{-2}d\norm{\rho^{T_A}}_1\right)$, and the copies of $\rho$ needed in a single run of circuit in Fig.~\ref{fig:ent_det}(c) when setting $t=(2l-1)$ scales as $K(l)=\widetilde{\mathcal{O}}(\epsilon^{-1}dd_Al)$.
\end{theorem}

We leave the technical proof of Theorem~\ref{theorem:cost_Ent_Neg} in Appendix~\ref{app:nega_cost}. Due to the property of entanglement negativity, $\norm{\rho^{T_A}}_1=\norm{\rho^{T_B}}_1$, we can always assume $d_A\le d_B$. As $\norm{\rho^{T_A}}_1\le d_A$, the worst case sample complexity is at most $\widetilde{\mathcal{O}}\left(\epsilon^{-3}d^3\right)$. As mentioned before, a straightforward negativity estimation protocol is to first generate an estimation of $\rho$ using quantum state tomography and then compute the negativity based on it. Quantum state tomography uses multiple copies of $\rho$ to learn a classical description $\hat{\rho}$ such that $\norm{\hat{\rho}-\rho}_1\le\epsilon$. The sample complexity of quantum state tomography using single-copy measurements is proved to be $\Theta(\epsilon^{-2}d^3)$ \cite{kueng2017low,chen2022tight}. It seems that this single-copy strategy has a lower sample complexity than the HME-based protocol.

However, accurately learning a quantum state in trace distance is not equivalent to accurately estimating negativity. One can show that 
\begin{equation}
\abs{N(\rho_1)-N(\rho_2)}\le\frac{\sqrt{d}}{2}\norm{\rho_1-\rho_2}_1
\end{equation}
holds for all states $\rho_1$ and $\rho_2$, where the equal sign can be asymptotically satisfied. Therefore, if one wants to estimate the negativity of any state to $\epsilon$ accuracy, an accuracy of $\Theta(\epsilon/\sqrt{d})$ in trace distance for state tomography is required. Substituting this into the complexity of quantum state tomography, the complexity of this tomography-based negativity estimation protocol scales as $\Theta(\epsilon^{-2}d^4)$, which is exponentially higher than that of the HME-based protocol. 

Note that, the exponential sample complexity of estimating negativity is unavoidable, even with highly joint operations. Based on the Holevo-Helstrom theorem, we prove the following proposition in Appendix~\ref{app:nega_LB}:
\begin{proposition}\label{prop:nega_LB}
If there exists a protocol that can estimate $N(\rho)$ to $\epsilon$ accuracy with a probability of at
least $2/3$ for any state $\rho$, the worst case sample complexity is lower bounded by $\Omega(\epsilon^{-2}d)$, where $d$ is the dimension of the Hilbert space for $\rho$.
\end{proposition}

Furthermore, calculating the negativity from the classical description of $\rho$ requires to conduct spectral decomposition of an exponentially large matrix, leading to requirements for exponentially large computational time and classical memory. While in the HME-based protocol, the number of classical calculations, $\mathcal{O}(M)$, is much lower than that of the spectral decomposition. Besides, the classical memory requirement of Algorithm~\ref{algo:nega} mainly comes from the subroutine $\operatorname{\mathbf{MedianOfMeans}}$, which scales logarithmically with $d$.

\section{Quantum Noiseless State Recovery}\label{sec:QNSR}

\subsection{Background and Protocol}
With the rapid advancement of quantum hardware and the increasing number of controllable qubits, noises in quantum circuits act as a significant obstacle to implementing quantum algorithms \cite{stilck2021limitations,chen2022complexity}. To enhance the capability of processing quantum information, it is crucial to not only optimize quantum hardware but also develop algorithms that are robust against various types of noises. Researchers have dedicated significant efforts in this direction, resulting in various protocols that can be categorized into two main types: quantum error correction and mitigation.

Proposed thirty years ago, quantum error correction is widely recognized as the cornerstone of fault-tolerant universal quantum computing \cite{shor1995scheme,steane1996error,preskill1998fault}. By encoding the target state into a larger Hilbert space, quantum error correction protects the target state against a large variety of errors \cite{terhal2015qec}. From the perspective of channel reversibility, the noise influence acting on the code space is reversible, although the physical noises themselves are normally irreversible. In principle, when the physical noise rate is below a certain threshold \cite{aharonov1997FaultTolerant}, quantum error correction can suppress the logical error rate to an arbitrarily small value by consuming a large number of ancilla qubits \cite{fowler2012surface}.

To address the challenges of realizing quantum error correction and mitigating errors with limited capabilities for current quantum devices, the concept of quantum error mitigation was proposed \cite{endo2021hybrid,cai2022quantum}. In contrast to quantum error correction, which aims to protect the noiseless state, quantum error mitigation focuses on a simpler task of recovering the noiseless expectation value. Specifically, let the noiseless pure state be $\psi$ and the real state be $\mathcal{E}(\psi)$, where $\mathcal{E}$ stands for some noise channel, quantum error mitigation aims to recover the value of $\Tr(O\psi)$ through performing operations on $\mathcal{E}(\psi)$ for some observable $O$.

Following such a routine, a number of quantum error mitigation protocols have been developed. Some protocols are tailored to handle specific types of noises, such as virtual distillation \cite{William2021virtual,koczor2021exponential} for incoherent errors and subspace expansion for coherent errors \cite{McClean2017hybrid}. At the same time, there are some protocols designed to handle general forms of noises, such as probabilistic error cancellation \cite{temme2017mitigation, endo2018practical}, which is one of the leading approaches in quantum error mitigation. Similar to quantum error correction, the probabilistic error cancellation method also tries to reverse the influence of noises, while in a statistical manner. As the inverse map of the noise channel, $\mathcal{E}^{-1}$, is non-physical except for unitary noises, one needs to decompose $\mathcal{E}^{-1}$ into a sum of experimentally feasible physical channels with positive or negative coefficients. Then, to obtain the value of $\Tr(O\psi)=\Tr[O\mathcal{E}^{-1}\circ\mathcal{E}(\psi)]$, one randomly applies the decomposed channels to the noisy states $\mathcal{E}(\psi)$ according to some probability distribution and measures the resulting states. Due to the statistical nature of this approach, quantum error mitigation typically requires an exponential number of experiments to accurately estimate the noiseless expectation value \cite{Takagi2021optimal,Takagi2022fundamental}.

Building upon HME, we propose a new protocol to handle general noises, which is named \emph{quantum noiseless state recovery}. Similar to entanglement detection, the quantum noiseless state recovery protocol relies on the capability of HME to implement Hermitian-preserving maps. We consider a similar setting as quantum error mitigation, where the ideal noiseless state $\psi$ is generated by an ideal noiseless circuit, $\psi=\mathcal{U}(\ketbra{0}{0})$, while the existence of noise will change it to $\mathcal{E}(\psi)=\mathcal{E}\circ\mathcal{U}(\ketbra{0}{0})$. Note that $\mathcal{E}^{-1}$ is always Hermitian-preserving when it exists \cite{Jiang2021implement}, thus the implementation of $\mathcal{E}^{-1}$ becomes possible with HME. By sequentially inputting multiple copies of $\mathcal{E}(\psi)$ into the circuit depicted in Fig.~\ref{fig:overview}(c) and choosing an appropriate Hamiltonian, we can approximately realize the evolution of
\begin{equation}
e^{-i\mathcal{E}^{-1}\circ\mathcal{E}(\psi)t}=e^{-i\psi t}
\end{equation}
with arbitrary accuracy.

It is evident that $e^{-i\psi t}$ has only one nontrivial eigenvalue, $e^{-it}$, with the corresponding eigenstate $\ket{\psi}$. 
As discussed in Sec.~\ref{subsec:ent_det}, when setting $t=\pi$, the quantum phase estimation circuit only needs one ancilla qubit to prepare $\psi$ from the evolution of $e^{-i\psi \pi}$. Therefore, a simple circuit shown in Fig.~\ref{fig:QNSR} can be used to prepare the noiseless state and thus serves as the circuit for quantum noiseless state recovery. The correctness of our protocol is ensured by the following proposition:
\begin{proposition}\label{prop:QEM_IdealPerformance}
When measuring the ancilla qubit in the ideal circuit depicted on the right side of Fig.~\ref{fig:QNSR}, the outcome $\ket{1}$ occurs with a probability of $\bra{\psi}\sigma\ket{\psi}$. If the outcome $\ket{1}$ is observed, $\sigma$ will evolve into the desired noiseless pure state $\ket{\psi}$.
\end{proposition}

\begin{figure}[t]
\centering
\includegraphics[width=0.48\textwidth]{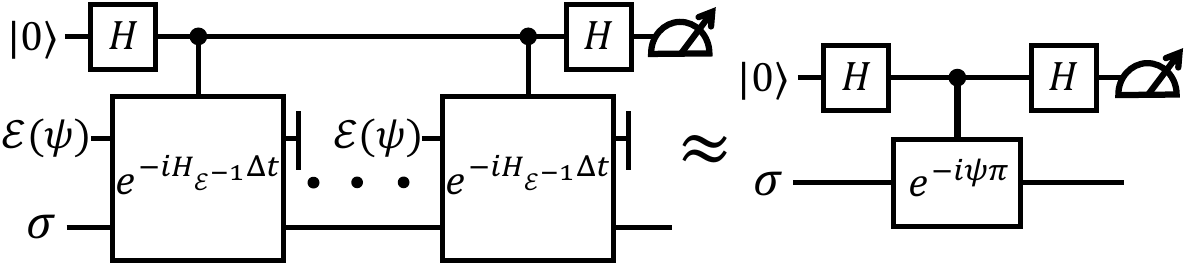}
\caption{The circuit of our error mitigation protocol. The Hamiltonian is set to be $H_{\mathcal{E}^{-1}}=\Lambda_{\mathcal{E}^{-1}}^{T_1}$, the partial transposition of the Choi matrix of the inverse map $\mathcal{E}^{-1}$. The evolved state $\sigma$ serves as a guiding state which needs to have a high fidelity with the noiseless state $\psi$.}
\label{fig:QNSR}
\end{figure}

In addition to the logic of quantum phase estimation, we can prove this proposition by directly analyzing the circuit. Observe that $e^{-i\psi\pi}=\mathbb{I}-2\psi$, when the controlled-$e^{-i\psi\pi}$ gate is applied to the ancilla qubit and the evolved state $\sigma$, the whole state will evolve into
\begin{equation}
\Sigma=\frac{1}{2}
\begin{bmatrix}
\sigma & \sigma-2\sigma\psi\\
\sigma-2\psi\sigma & \sigma-2\sigma\psi-2\psi\sigma+4\psi\sigma\psi
\end{bmatrix},
\end{equation}
whose row and column indices labeling the matrix blocks correspond to the indices of the ancilla qubit. When the measurement outcome of the ancilla qubit is $\ket{1}$, which corresponds to measuring $\ket{-}\otimes\mathbb{I}_d$ on $\Sigma$, the resulting state will become
\begin{equation}\label{eq:qnsr_ideal_output_state}
\begin{aligned}
&\Tr_c\left(\frac{(\ketbra{-}{-}\otimes\mathbb{I}_d)\Sigma(\ketbra{-}{-}\otimes\mathbb{I}_d)}{\Tr\big[(\ketbra{-}{-}\otimes\mathbb{I}_d)\Sigma\big]}\right)\\
=&\frac{1}{\bra{\psi}\sigma\ket{\psi}}\Tr_c\left(\frac{1}{2}
\begin{bmatrix}
\psi\sigma\psi & -\psi\sigma\psi \\
-\psi\sigma\psi & \psi\sigma\psi
\end{bmatrix}\right)=\psi,
\end{aligned}
\end{equation}
where $\Tr_c$ denotes tracing out the control qubit, and $\bra{\psi}\sigma\ket{\psi}$ is the corresponding probability.

The sample complexity of quantum noiseless state recovery is mainly influenced by two factors. First, the realization of the controlled-$e^{-i\psi \pi}$ evolution requires multiple input states $\mathcal{E}(\psi)$, the number of which is upper bounded by $\mathcal{O}(\epsilon^{-1}\norm{H_{\mathcal{E}^{-1}}}_\infty^2 \pi^2)=\mathcal{O}(\epsilon^{-1}\norm{H_{\mathcal{E}^{-1}}}_\infty^2)$, as predicted in Corollary~\ref{corollary:cost(controlled_version)}. Yet, after this evolution, one performs measurements and post-selection to prepare the target state. This post-selection amplifies a block in the whole state $\Sigma$ by $\bra{\psi}\sigma\ket{\psi}^{-1}$ times. Intuitively, to keep the deviation of the prepared state from the ideal state smaller than $\epsilon$, one needs to realize the controlled-$e^{-i\psi \pi}$ evolution with an accuracy of $\bra{\psi}\sigma\ket{\psi}\epsilon$, which increases the sample complexity by a factor of $\bra{\psi}\sigma\ket{\psi}^{-1}$. Second, only when the measurement outcome of the ancilla qubit is $\ket{1}$, the resultant state will be $\psi$. Thus, on average, a number of $\bra{\psi}\sigma\ket{\psi}^{-1}$ experiments are needed to obtain one successful experiment. Combining these two factors, we arrive at the following theorem:
\begin{theorem}\label{theorem:HME-QEM}
Let $\mathcal{E}$ be an invertible noise map. Denote $H_{\mathcal{E}^{-1}}:=\Lambda_{\mathcal{E}^{-1}}^{T_1}$ as the Hamiltonian corresponding to the inverse noise channel, and let $F:=\bra{\psi}\sigma\ket{\psi}$ be the fidelity between $\ketbra{\psi}{\psi}$ and $\sigma$. Then we need at most $\mathcal{O}\left(\log(\delta^{-1})\epsilon^{-1}F^{-2}\norm{H_{\mathcal{E}^{-1}}}_{\infty}^2\right)$ copies of $\mathcal{E}(\ketbra{\psi}{\psi})$ to approximately produce $\ketbra{\psi}{\psi}$ with trace distance smaller than $\epsilon$ and a success probability at least $1-\delta$.
\end{theorem}

We leave the technical proof of this theorem in Appendix~\ref{app:HME_QNSR_thm}. As $F$ is a key quantity determining the sample complexity, in a practical situation where the noise rate is relatively small, one can choose the noisy state $\mathcal{E}(\psi)$ as the evolved state $\sigma$ to reduce the sample complexity.

\subsection{Comparison}

We should acknowledge that, in comparison to the well-developed quantum error correction and mitigation protocols, quantum noiseless state recovery still faces significant challenges. The implementation of quantum noiseless state recovery relies on rigorous assumptions, such as having precise knowledge of the noise channel to determine the Hamiltonian and an efficient method for implementing controlled Hamiltonian evolution. To apply quantum noiseless state recovery in practical scenarios, many efforts should be made to enhance its practicality. However, quantum noiseless state recovery exhibits fundamental differences from quantum error correction and mitigation, which may provide inspiration for alternative error management protocols. Therefore, a deeper understanding of the characteristics of quantum noiseless state recovery is meaningful.

Quantum noiseless state recovery combines the starting point of quantum error mitigation and the target of quantum error correction. When considering the state of processing, quantum noiseless state recovery is closer to quantum error mitigation, as it deals with states that have already been affected by noises. In contrast, quantum error correction begins with noiseless states and protects them from the influence of noises. Regarding the target, quantum noiseless state recovery aligns more closely with quantum error correction since both aim to recover the noiseless state rather than solely the noiseless expectation values. For a clearer comparison, we present a summary of the distinctive features of these methods in Table~\ref{tab:QEC&QEM&HME}.

\begin{table}[htbp]
\centering
\resizebox{0.48\textwidth}{!}{
\begin{tabular}{c|c|c|c}
\hline
\hline
& QEC & QEM & \makecell{QNSR} \\
\hline
\hline
Main Technique & \makecell{Encoding and\\Decoding} & \makecell{Statistical\\Methods} & \makecell{HME and Quantum\\Phase Estimation} \\
\hline 
State of Processing & Noiseless & Noisy & Noisy\\
\hline
Target & \makecell{State\\ Recovery} & \makecell{Expectation\\ Value Recovery} & \makecell{State\\ Recovery} \\
\hline 
Qubit Overhead & High & Low & Moderate \\
\hline
Sample Complexity & Polynomial & Exponential & Exponential\\
\hline
General Noise & No & Yes & Yes\\
\hline
\hline
\end{tabular}
}
\caption{The comparison of QEC (quantum error correction), QEM (quantum error mitigation), and QNSR (quantum noiseless state recovery).}
\label{tab:QEC&QEM&HME}
\end{table}

In terms of qubit overhead, quantum noiseless state recovery also falls between quantum error correction and mitigation. For the circuit of quantum noiseless state recovery, as shown in Fig.~\ref{fig:QNSR}, a system with $(2n+1)$ qubits can be utilized to recover a $n$-qubit noiseless state. In comparison, typical error correction codes such as the surface code or concatenated code \cite{fowler2012surface,aharonov1997FaultTolerant} may require a much larger number of ancilla qubits to protect a single logical qubit. Whereas, quantum error mitigation only requires a small number of ancilla qubits, but is thus limited to recovering the noiseless expectation values. It is worth noting that the ability to recover quantum states, rather than solely retrieving expectation values, offers several advantages \cite{cai2022quantum}. For instance, in applications such as quantum storage, the primary goal is to obtain noise-free quantum states. Additionally, in practical scenarios such as the Shor algorithm \cite{shor1994factoring} and quantum key distribution \cite{bennett1984quantum,ekert1991quantum}, noiseless states are required to obtain single-shot measurement results instead of just expectation values.

Both quantum error mitigation and quantum noiseless state recovery suffer from higher sample complexities compared to quantum error correction. The presence of single-qubit independent noise poses a significant challenge for quantum devices by greatly limiting their ability to demonstrate quantum advantages and solve practical problems \cite{stilck2021limitations}. Studies have revealed that for independent single-qubit noises, the sample complexity of quantum error mitigation has an exponential lower bound \cite{Takagi2022fundamental,quek2022exponentially}. Unfortunately, quantum noiseless state recovery faces the same difficulty in this regard. Based on Theorem~\ref{theorem:LowerBound}, we prove in Appendix~\ref{app:hardness_of_indnependnet_noise} that:

\begin{corollary}\label{corollary:hardness_of_indnependnet_noise}
Let $\mathcal{E}^{\otimes n}$ be an $n$-qubit noise channel, which is the $n$-fold tensor product of single-qubit invertible and non-unitary noise channel $\mathcal{E}$. Then, the sample complexity of realizing controlled-$e^{-i(\mathcal{E}^{-1})^{\otimes n}(\rho)t}$ evolution to an accuracy of $\epsilon$ in diamond distance has an exponential lower bound, $K \ge \epsilon^{-1}2^{\Omega(n)}t^2$, even with the joint operations depicted in Fig.~\ref{fig:LowerBoundModel}(a).
\end{corollary}

It is important to emphasize that the higher sample complexity of quantum noiseless state recovery compared to quantum error correction does not solely stem from the fact that quantum noiseless state recovery uses fewer ancilla qubits. While both methods aim to obtain the noiseless state, quantum noiseless state recovery faces a more challenging task: recovering the noiseless quantum state from noisy states instead of protecting noiseless states from noises. As local noises will exponentially reduce the distinguishability of quantum states, recovery from noisy states will encounter exponential sample complexities, as suggested by the spirit of the Holevo-Helstrom theorem.

In practice, quantum noiseless state recovery excels in handling certain types of noises that pose challenges for quantum error correction. For instance, leakage error is a significant challenge for some quantum platforms \cite{Motzoi2009leakage}, which occurs when the target system possesses a hidden space, and the quantum state may leak out of the computational space. Detecting and operating on the hidden space is typically challenging, making it difficult for quantum error correction to handle leakage errors effectively. In contrast, quantum noiseless state recovery can naturally handle such non-trace-preserving errors, as its inverse map is also Hermitian-preserving.

\section{Other Applications}\label{sec:other_app}
Although not as evident as in entanglement detection and noiseless state recovery, Hermitian-preserving maps actually emerge in a wide range of quantum information tasks. In this section, we will explore several other examples, such as expectation value measurements, quantum imaginary time evolution, and linear combinations of unitaries, where HME can be effectively applied.

\subsection{Expectation Value Measurements}

Except for protocols such as randomized measurements \cite{Elben2023toolbox, Brydges2019Probing, huang2020predicting} and shadow tomography \cite{Aaronson2018shadow}, which either require exponential repetition times or joint operations, we typically employ two methods to efficiently estimate the expectation value of a single observable, $\Tr(O\rho)$. The first method is based on the fundamental principles of quantum mechanics. Suppose $O$ can be diagonalized as $O=V_O\Lambda_OV_O^\dagger$, where $V_O$ and $\Lambda_O$ are unitary and diagonal matrices, respectively. In this case, we can evolve the state $\rho$ using $V_O^\dagger$, measure it in the computational basis multiple times, and use these measurement results to estimate $\Tr(O\rho)=\Tr(\Lambda_OV_O^\dagger\rho V_O)$. The circuit for this method is shown in Fig.~\ref{fig:obs_circuit}(a). The second method is based on embedding the observable into a unitary $U_O$. A common embedding technique is based on the textbook result $O=\frac{\norm{O}_{\infty}}{2}(U_O+U_O^\dagger)$. With this embedding, we can perform the Hadamard test circuit shown in Fig.~\ref{fig:obs_circuit}(b) using a controlled-$U_O$ operation to acquire the value of $\Tr[(U_O+U_O^\dagger)\rho]$. This method is particularly efficient when $O$ is not only Hermitian but also unitary in certain scenarios.

In certain practical scenarios, $O$ may correspond to a sparse or intrinsic Hamiltonian of physical systems. In such cases, implementing $V_O$ and $U_O$ may still require deep quantum circuits \cite{Ueda2020eth} as they do not exploit the structure information of $O$. While, the implementation of $e^{-iOt}$ and controlled-$e^{-iOt}$ evolutions can be quite efficient, with methods such as Hamiltonian simulation protocols \cite{suzuki1990fractal, suzuki1991fractal, childs2012LCU} or Hilbert space enlargement techniques \cite{Zhou2011control, Zhou2013eigenvalues}.

\begin{figure}
\centering
\includegraphics[width=0.48\textwidth]{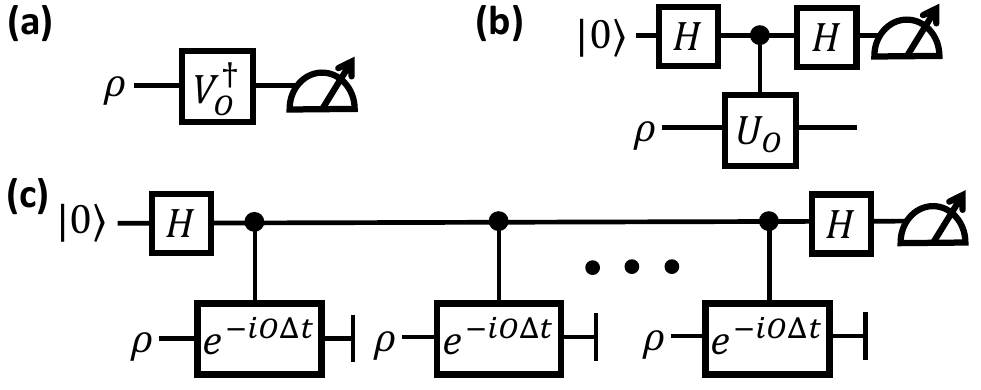}
\caption{The circuits of three direct expectation value measurement methods, where all the measurements are conducted in the $Z$ basis. (a) $V_O$ is the unitary that diagonalizes $O$. (b) $U_O$ is obtained by decomposing $O$ as $O=\frac{\norm{O}_{\infty}}{2}(U_O+U_O^\dagger)$, and $H$ represents the Hadamard gate. (c) The circuit utilizes HME to perform the observable measurement, utilizing controlled-$e^{-iO\Delta t}$ operations.}
\label{fig:obs_circuit}
\end{figure}

An intriguing observation is that $\Tr(O\rho)$ is also a Hermitian-preserving map that acts on $\rho$ since $\Tr(O\rho)$ is a real number. This suggests the possibility of conducting observable measurements through HME. To facilitate this application, we can define a new map $\mathcal{N}_O(\rho) = \Tr(O\rho)\ketbra{1}{1}$, where the output space is a single-qubit Hilbert space and $\ket{1}$ is one of the basis states in this space. According to Theorem~\ref{theorem:maintheorem}, the Hamiltonian for implementing $\mathcal{N}_O$ is $H_O=O\otimes\ketbra{1}{1}$, and the evolution of $e^{-iH_Ot}$ is actually the controlled-$e^{-iOt}$ operation. With this Hamiltonian, HME enables the encoding of expectation values into relative phases of reference states. By setting the evolved state as $\ket{+}=\frac{1}{\sqrt{2}}(\ket{0}+\ket{1})$, as shown in Fig.~\ref{fig:obs_circuit}(c), the evolved state will approximately become
\begin{equation}
e^{-i\Tr(O\rho)t\ketbra{1}{1}}\ket{+}=\frac{1}{\sqrt{2}}\left(\ket{0}+e^{-i\Tr(O\rho)t}\ket{1}\right).
\end{equation}
Subsequently, one can measure the evolved state in the Pauli basis to extract the value of $\tr(O\rho)$. By defining a new map as $\mathcal{N}(\rho)=\sum_i\Tr(O_i\rho)\ketbra{i}{i}$ and setting the evolved state to be a higher-dimensional coherent state, HME can encode multiple expectation values simultaneously. This HME-based phase encoding operation may also exhibit some advantages for other applications such as gradient estimation \cite{Andras2019gradient, Huggins2022Optimal}.

\subsection{Completely Positive Maps without Trace Preservation}

Until now, most of the non-physical maps we have discussed, including positive maps, inverse maps of noise channels, and phase encoding maps, violate the complete positivity requirements for quantum channels. However, there also exist some completely positive but not trace-preserving maps that play important roles in quantum information processing tasks. One important example is the map $\mathcal{N}(\rho)=P\rho P^\dagger$, where $P$ can be an arbitrary rectangular matrix. In the task of linear combination of unitaries \cite{childs2012LCU}, the matrix $P$ is chosen as the sum of several unitaries, allowing for producing an arbitrary pure state or realizing a desired Hamiltonian evolution. Additionally, in quantum imaginary time evolution \cite{lu2021finite, zhang2021thermal, motta2020determining}, $P$ is set as $e^{-\beta H}$, where $\beta$ represents the inverse temperature and $H$ is a Hamiltonian. By increasing $\beta$, one can prepare a pure state that approximates the ground state of $H$ to arbitrary precision. 

It is worth noting that linear combination of unitaries and quantum imaginary time evolution share a common goal of preparing pure states, which is similar to quantum noiseless state recovery. To achieve this goal, one can adopt a similar circuit as shown in Fig.~\ref{fig:QNSR}. In particular, by carefully selecting a Hamiltonian, HME facilitates the realization of the controlled-$e^{-iP\ketbra{0}{0}P^\dagger t}$ evolution, where $\ketbra{0}{0}$ represents some initial pure state. If one has sufficient prior knowledge of $P$, the pure state corresponding to $P\ketbra{0}{0}P^\dagger$ can also be generated using the quantum phase estimation circuit with a small number of ancilla qubits.

\section{Discussion and outlook}\label{sec:conclusion}

The Hermitian-preserving map exponentiation algorithm relies on quantum memory, joint operations, and the ability to reset states. Based on these resources, the advantages of HME manifest in three aspects. First, quantum memory is known to provide exponential speedups in certain tasks such as state discrimination and property testing \cite{chen2022memory}. Consequently, HME exhibits exponential speedups compared to single-copy protocols, as demonstrated in tasks of entanglement detection and quantification. Second, by employing ancilla qubits and joint operations, HME enables the transformation of incoherent operations into coherent ones, thereby realizing tasks that are infeasible for incoherent operations. An important example is quantum error mitigation, where only single-copy operations are allowed due to hardware limitations, one can only recover the noiseless expectation values. While using HME, it becomes possible to recover the noiseless state from noisy states, thereby enabling tasks such as quantum storage. Additionally, HME offers a means to utilize quantum resources to replace classical resources, especially in state benchmarking protocols. Leveraging state-resetting operations, HME extracts information from multiple copies of states, reducing the classical resources for estimating nonlinear quantities. A direct example is the negativity estimation where HME exponentially reduces the requirements for classical computation and memory compared to tomography-based methods. Similarly, in energy estimation, HME utilizes the ability to perform Hamiltonian evolution as a substitute for the computational cost of the spectral decomposition of the Hamiltonian operator.

As a building block for quantum algorithms, HME holds the potential for improvement, application, and extension to various areas. One possible avenue is to combine HME with algorithms such as amplitude amplification, which could enhance the performance of tasks discussed in this work, including noiseless state recovery. Furthermore, while density matrix exponentiation was initially proposed for processing classical data, such as covariance matrices \cite{Lloyd2014qpca}, it has been shown to lack exponential advantages compared to classical strategies without state preparation assumptions \cite{tang2021only}. However, HME might restore the possibility of achieving quantum advantages by allowing more general operations on classical data. Moreover, by replacing the Hamiltonian evolution in the HME circuit with non-Hermitian evolution, HME can even circumvent the Hermitian preservation requirement and showcases potential advantages for a wider range of tasks involving Hermitian-violating maps. Finally, instead of adjusting the HME algorithm itself, another direction worth exploring is to search for alternative methods to implement more general non-physical maps, such as nonlinear maps.

Some open problems remain. As demonstrated in Appendix~\ref{app:freedom}, there is freedom in choosing the Hamiltonian for HME. Could this freedom lead to an improved upper bound on the sample complexity for certain non-physical maps compared to what is predicted in Theorem~\ref{theorem:cost}? As discussed in Sec.~\ref{subsec:optimality}, is it possible to obtain a better lower bound on the sample complexity for implementing non-physical maps by taking into account hardware limitations, such as the size of the quantum memory? In Sec.~\ref{subsec:ent_det}, we demonstrate the advantages of HME through a specific entanglement detection task. An open problem that arises is whether we can combine HME with other positive maps to solve more general entanglement detection tasks with provable advantages. In Sec.~\ref{subsec:neg}, we compare our HME-based negativity estimation protocol with the tomography-based one and observe an exponential speedup. This leads to the natural question of whether the HME-based protocol possesses quantum advantages compared to all single-copy protocols. Furthermore, can we design more efficient single-copy negativity measurement protocols that outperform the tomography-based approach?

In Sec.~\ref{sec:QNSR}, we introduced the quantum noiseless state recovery protocol, which utilizes multiple noisy states to generate noiseless ones. Essentially, the HME algorithm provides the ability for undoing a quantum channel with arbitrarily low error. In contrast, conventional methods such as the Petz recovery map \cite{Petz1986,gilyen2022petz} can only realize a quantum channel that is close to the inverse map. It is worth exploring other potential applications of this undoing operation in addition to managing quantum errors. Besides, with the ability to recover noiseless states, the protocol holds potential for applications in quantum communication, particularly in the recovery of entangled states shared between distant clients. Therefore, an important future direction is to adapt the quantum noiseless state recovery protocol into a distributed version, allowing experimenters from different nodes of a quantum network to implement it effectively. In addition, the quantum noiseless state recovery protocol still requires prior knowledge of the noise channel. Can we eliminate this requirement by sacrificing other resources?

\begin{acknowledgments}
We appreciate insightful discussions with Zhenyu Cai, Jinzhao Sun, Zhaohui Wei, and Huixuan He. Fuchuan Wei and Zhengwei Liu are supported by BMSTC and ACZSP (Grant No. Z221100002722017). Zhengwei Liu is supported by NKPs (Grant No. 2020YFA0713000), Beijing Natural Science Foundation Key Program (Grant No. Z220002). Zhenhuan Liu, Guoding Liu, and Xiongfeng Ma are supported by the National Natural Science Foundation of China (Grant No. 12174216) and the Innovation Program for Quantum Science and Technology (Grant No.2021ZD0300804).  Zizhao Han and Dong-Ling Deng are supported by the National Natural Science Foundation of China (Grants No. 12075128 and T2225008), Tsinghua University Dushi Program, and Shanghai Qi Zhi Institute.
\end{acknowledgments}

\appendix

\section{Auxiliary Lemmas}\label{app:Auxiliary_Lemmas}

The essential property of the diamond distance is its subadditivity:

\begin{lemma}[Proposition 3.48 in \cite{watrous2018theory}]\label{lemma:subadditivity}
Suppose $\mathcal{Q}_0,\mathcal{Q}_0^{\prime}\in T(\mathcal{X},\mathcal{Y})$ and $\mathcal{Q}_1,\mathcal{Q}_1^{\prime}\in T(\mathcal{Y},\mathcal{R})$ are CPTP maps, then the following inequality holds:
\begin{equation}\label{eq:subadd}
\norm{\mathcal{Q}_1\circ\mathcal{Q}_0-\mathcal{Q}_1^{\prime}\circ\mathcal{Q}_0^{\prime}}_{\diamond}\le\norm{\mathcal{Q}_1-\mathcal{Q}_1^{\prime}}_{\diamond}+\norm{\mathcal{Q}_0-\mathcal{Q}_0^{\prime}}_{\diamond}.
\end{equation}
\end{lemma}

The following lemma provides an upper bound on the diamond distance between two channels corresponding to two Hamiltonian evolutions:

\begin{lemma}\label{lemma:DiamondDistance_HamiltonianEvolution}
Let $H_1$ and $H_2$ be two Hermitian operators, and $t>0$. Then, it holds that:
\begin{equation}\label{eq:DiamondDistance_HamiltonianEvolution}
\begin{aligned}
&\norm{\left[e^{-iH_1t}\right]-\left[e^{-iH_2t}\right]}_{\diamond}\\
\le&2t\norm{H_1-H_2}_{\infty}\exp\Big(t\max\left\{\norm{H_1}_{\infty},\norm{H_2}_{\infty}\right\}\Big).
\end{aligned}
\end{equation}
\end{lemma}
\begin{proof}
By Lemma 5 in the appendix of \cite{Caro2022generalization}, we have
\begin{equation}
\norm{\left[e^{-iH_1t}\right]-\left[e^{-iH_2t}\right]}_{\diamond}\le2\norm{e^{-iH_1t}-e^{-iH_2t}}_{\infty}.
\end{equation}
Moreover,
\begin{equation}
\begin{aligned}
&\norm{e^{-iH_1t}-e^{-iH_2t}}_{\infty}\\
=&\norm{\sum_{m=0}^{\infty}\frac{1}{m!}(-iH_1t)^m-\sum_{m=0}^{\infty}\frac{1}{m!}(-iH_2t)^m}_{\infty}\\
\le&\sum_{m=1}^{\infty}\frac{1}{m!}t^m\norm{H_1^m-H_2^m}_{\infty}\\
=&\sum_{m=1}^{\infty}\frac{1}{m!}t^m\norm{\sum_{l=0}^{m-1}H_1^{m-1-l}(H_1-H_2)H_2^l}_{\infty}\\
\le&\sum_{m=1}^{\infty}\frac{1}{m!}t^m\sum_{l=0}^{m-1}\norm{H_1}_{\infty}^{m-1-l}\norm{H_1-H_2}_{\infty}\norm{H_2}_{\infty}^l\\
\le&t\norm{H_1-H_2}_{\infty}\sum_{m=1}^{\infty}\frac{t^{m-1}}{(m-1)!}\max\left\{\norm{H_1}_{\infty},\norm{H_2}_{\infty}\right\}^{m-1}\\
=&t\norm{H_1-H_2}_{\infty}\exp\Big(t\max\left\{\norm{H_1}_{\infty},\norm{H_2}_{\infty}\right\}\Big).
\end{aligned}
\end{equation}
\end{proof}

When measuring an observable over two quantum states, the difference in the expectation values can be upper bounded by applying the \textit{Hölder inequality} \cite{watrous2018theory} to $\abs{\Tr\big(O(\sigma_1-\sigma_2)\big)}$:

\begin{lemma}\label{lemma:expectation_tracenorm}
Let $\sigma_1,\sigma_2\in D(\mathcal{H})$ be two quantum states, then the difference in the expectation values when measuring $O$ satisfies
\begin{equation}
\abs{\Tr(O\sigma_1)-\Tr(O\sigma_2)}\le\norm{O}_{\infty}\norm{\sigma_1-\sigma_2}_{1}.
\end{equation}
\end{lemma}

\section{Freedom of Hamiltonian in HME}\label{app:freedom}

As shown in Theorem~\ref{theorem:maintheorem}, $\Lambda_{\mathcal{N}}^{T_1}$ is a Hamiltonian that can be used to implement the Hermitian-preserving map $\mathcal{N}\in T(\mathcal{H},\mathcal{K})$. In this section we point out that, $\Lambda_{\mathcal{N}}^{T_1}$ is not the unique choice. Specifically, we prove that all the Hamiltonians that can be used to implement $\mathcal{N}$ form the set
\begin{equation}
\mathscr{H}_{\mathcal{N}}=\left\{\Lambda_{\mathcal{N}}^{T_1}+M\otimes \mathbb{I}_{\mathcal{K}}:M\in\mathcal{L}(\mathcal{H}),M^{\dagger}=M\right\}.
\end{equation}

Note that the equation $\Tr_1\left(e^{-iH\Delta t}(\rho\otimes\sigma)e^{iH\Delta t}\right)=e^{-i\mathcal{N}(\rho)\Delta t}\sigma e^{i\mathcal{N}(\rho)\Delta t}+\mathcal{O}(\Delta t^2)$ is equivalent to
\begin{equation}\label{eq:first_order_alignment}
\begin{aligned}
\Tr_1\big([H,\rho\otimes\sigma]\big)=[\mathcal{N}(\rho),\sigma].
\end{aligned}
\end{equation}
Thus, any $H$ satisfying Eq.~\eqref{eq:first_order_alignment} for all $\rho\in D(\mathcal{H})$ and $\sigma\in D(\mathcal{K})$ is a valid choice for implementing $\mathcal{N}$.

It can be easily verified that any Hamiltonian in $\mathscr{H}_{\mathcal{N}}$ satisfies Eq.~\eqref{eq:first_order_alignment} for all $\rho$ and $\sigma$, as $\Tr_1\big([M\otimes\mathbb{I}_{\mathcal{K}},\rho\otimes\sigma]\big)=0$ always holds.

On the other hand, if $H$ satisfies Eq.~\eqref{eq:first_order_alignment} for all $\rho$ and $\sigma$, we have
\begin{equation}\label{eq:tr_1_commute_rho_sigma}
\Tr_1\left(\left[H-\Lambda_{\mathcal{N}}^{T_1},\rho\otimes\sigma\right]\right)=0,~\forall\rho\in D(\mathcal{H}),~\sigma\in D(\mathcal{K}).
\end{equation}
Since any matrix $A\in L({\mathcal{H}\otimes\mathcal{K}})$ can be expressed as a complex linear combination of matrices of the form $\rho\otimes\sigma$, Eq.~\eqref{eq:tr_1_commute_rho_sigma} implies that
\begin{equation}
\Tr_1\left(\left[H-\Lambda_{\mathcal{N}}^{T_1},A\right]\right)=0,~\forall A\in L({\mathcal{H}\otimes\mathcal{K}}).
\end{equation}
Based on Lemma~\ref{lemma:tr_1_commute_implies_some_structure} presented below, we conclude that $H-\Lambda_{\mathcal{N}}^{T_1}=M\otimes \mathbb{I}_{\mathcal{K}}$ for some matrix $M\in L({\mathcal{H}})$. As $H$ should be a valid Hamiltonian, $M$ must be a Hermitian matrix, which implies $H\in\mathscr{H}_{\mathcal{N}}$.

\begin{lemma}\label{lemma:tr_1_commute_implies_some_structure}
Suppose a matrix $B\in L({\mathcal{H}\otimes\mathcal{K}})$ satisfies
\begin{equation}
\Tr_1\big([B,A]\big)=0,~\forall A\in L({\mathcal{H}\otimes\mathcal{K}}),
\end{equation}
where $\Tr_1$ denotes the partial trace over the first system, $\mathcal{H}$. Then $B=M\otimes \mathbb{I}_{\mathcal{K}}$ for some matrix $M\in L({\mathcal{H}})$.
\end{lemma}
\begin{proof}
For all $\ketbra{i}{j}\otimes\ketbra{k}{l}\in L({\mathcal{H}\otimes\mathcal{K}})$, we have
\begin{equation}
\Tr_1\left(\big[B,\ketbra{i}{j}\otimes\ketbra{k}{l}\big]\right)=0.
\end{equation}
Thus, for all $i$ and $j$, we have
\begin{equation}
\big[\left(\bra{j}\otimes \mathbb{I}_{\mathcal{K}}\right)B\left(\ket{i}\otimes \mathbb{I}_{\mathcal{K}}\right),\ketbra{k}{l}\big]=0,~\forall k,l,
\end{equation}
which implies that $\left(\bra{j}\otimes \mathbb{I}_{\mathcal{K}}\right)B\left(\ket{i}\otimes \mathbb{I}_{\mathcal{K}}\right)=M_{ji}\mathbb{I}_{\mathcal{K}}$ for some constant $M_{ji}\in\mathbb{C}$. 
By defining $M\in L(\mathcal{H})$ as a matrix with elements in the $j$-th row and $i$-th column given by $M_{ji}$, we have $B=M\otimes \mathbb{I}_{\mathcal{K}}$.
\end{proof}
An open problem is whether the freedom in choosing the Hamiltonian can help to reduce the sample complexity of HME. While for the Hermitian-preserving maps discussed in this paper, $\Lambda_{\mathcal{N}}^{T_1}$ is a suitable choice, it does not exclude the possibility that for certain Hermitian-preserving maps $\mathcal{N}$, there exist better choices within the set $\mathscr{H}_{\mathcal{N}}$ that could potentially reduce the sample complexity of realizing $e^{-i\mathcal{N}(\cdot)t}$.

\section{Sample Complexity Upper Bound of HME}\label{app:proof_of_err}

\subsection{Proof of Theorem~\ref{theorem:cost}}\label{app:Proof_of_Thm2}

As $\mathcal{Q}_t=\mathcal{Q}_{\Delta t}^{\circ K}$ and $\mathcal{U}_t=\mathcal{U}_{\Delta t}^{\circ K}$, we can use Lemma~\ref{lemma:subadditivity} iteratively to obtain
\begin{equation}\label{eq:linearly}
\norm{\mathcal{Q}_t-\mathcal{U}_t}_{\diamond}=\norm{\mathcal{Q}_{\Delta t}^{\circ K}-\mathcal{U}_{\Delta t}^{\circ K}}_{\diamond}\le K\norm{\mathcal{Q}_{\Delta t}-\mathcal{U}_{\Delta t}}_{\diamond}.
\end{equation}
Therefore, the error only accumulates additively, and it is sufficient to analyze the error of a single step: $\norm{\mathcal{Q}_{\Delta t}-\mathcal{U}_{\Delta t}}_{\diamond}$. Let $\mathcal{N}\in T(\mathcal{H},\mathcal{K})$ be a Hermitian-preserving map from $L(\mathcal{H})$ to $L(\mathcal{K})$. Take arbitrary reference system $\mathcal{R}$ and arbitary $\sigma_{\mathcal{K}\mathcal{R}}\in D(\mathcal{K}\mathcal{R})$, we have
\begin{equation}\label{eq:appro_terms}
\begin{aligned}
&(\mathcal{Q}_{\Delta t}\otimes\mathcal{I}_{\mathcal{R}})(\sigma_{\mathcal{K}\mathcal{R}})\\
=&\Tr_1\left[(e^{-i H\Delta t}\otimes \mathbb{I}_{\mathcal{R}})(\rho\otimes\sigma_{\mathcal{K}\mathcal{R}})(e^{i H\Delta t}\otimes \mathbb{I}_{\mathcal{R}})\right]\\
=&\Tr_1\left[e^{-i (H\otimes \mathbb{I}_{\mathcal{R}})\Delta t}(\rho\otimes\sigma_{\mathcal{K}\mathcal{R}})e^{i (H\otimes \mathbb{I}_{\mathcal{R}})\Delta t}\right]\\
=&\Tr_1\left[\sum_{m=0}^{\infty}\frac{1}{m!}\operatorname{ad}_{\left(-i (H\otimes \mathbb{I}_{\mathcal{R}})\Delta t\right)}^m(\rho\otimes\sigma_{\mathcal{K}\mathcal{R}})\right]\\
=&\sigma_{\mathcal{K}\mathcal{R}}-i\Delta t\left[\mathcal{N}(\rho)\otimes \mathbb{I}_{\mathcal{R}},\sigma_{\mathcal{K}\mathcal{R}}\right]\\
&+\Tr_1\left[\sum_{m=2}^{\infty}\frac{1}{m!}\operatorname{ad}_{\left(-i (H\otimes \mathbb{I}_{\mathcal{R}})\Delta t\right)}^m(\rho\otimes\sigma_{\mathcal{K}\mathcal{R}})\right].
\end{aligned}
\end{equation}
Here the notation $\operatorname{ad}$ is defined as
\begin{equation}
\begin{aligned}
\operatorname{ad}_{A}(B)&=[A,B],\\
\operatorname{ad}_{A}^m(B)&=[A,\operatorname{ad}_{A}^{m-1}(B)],
\end{aligned}
\end{equation}
for square matrices $A$ and $B$ of the same size. The trace norm of the third term in the last equation of Eq.~\eqref{eq:appro_terms} satisfies
\begin{equation}\label{eq:appro_remainder_thirdterm}
\begin{aligned}
&\norm{\Tr_1\left[\sum_{m=2}^{\infty}\frac{1}{m!}\operatorname{ad}_{\left(-i (H\otimes \mathbb{I}_{\mathcal{R}})\Delta t\right)}^m(\rho\otimes\sigma_{\mathcal{K}\mathcal{R}})\right]}_1\\
\le &\norm{\sum_{m=2}^{\infty}\frac{1}{m!}\operatorname{ad}_{\left(-i (H\otimes \mathbb{I}_{\mathcal{R}})\Delta t\right)}^m(\rho\otimes\sigma_{\mathcal{K}\mathcal{R}})}_1\\
\le &\sum_{m=2}^{\infty}\frac{1}{m!}\norm{\operatorname{ad}_{\left(-i (H\otimes \mathbb{I}_{\mathcal{R}})\Delta t\right)}^m(\rho\otimes\sigma_{\mathcal{K}\mathcal{R}})}_1\\
=&\sum_{m=2}^{\infty}\frac{1}{m!}(\Delta t)^m\norm{\operatorname{ad}_{\left(H\otimes \mathbb{I}_{\mathcal{R}}\right)}^m(\rho\otimes\sigma_{\mathcal{K}\mathcal{R}})}_1\\
\le&\sum_{m=2}^{\infty}\frac{1}{m!}(\Delta t)^m 2^m \norm{H\otimes \mathbb{I}_{\mathcal{R}}}_{\infty}^m\norm{\rho\otimes\sigma_{\mathcal{K}\mathcal{R}}}_1\\
=&\sum_{m=2}^{\infty}\frac{1}{m!}(2\norm{H}_{\infty}\Delta t)^m\\
=&e^{2\norm{H}_{\infty}\Delta t}-1-2\norm{H}_{\infty}\Delta t\\
\le&4\norm{H}_{\infty}^2\Delta t^2,\text{ when } \norm{H}_{\infty}\Delta t\in(0,0.8].
\end{aligned}
\end{equation}
The first inequality in Eq.~\eqref{eq:appro_remainder_thirdterm} is due to the non-increasing property of trace norm under partial trace. The third inequality in Eq.~\eqref{eq:appro_remainder_thirdterm} is obtained by iteratively applying the inequalities $\norm{AB}_1\le\norm{A}_{\infty}\norm{B}_1$ and $\norm{AB}_1\le\norm{A}_1\norm{B}_{\infty}$, which hold for all square matrices $A$ and $B$ of the same size. The second equality in Eq.~\eqref{eq:appro_remainder_thirdterm} follows from $\norm{\rho\otimes\sigma_{\mathcal{K}\mathcal{R}}}_1=1$ and $\norm{H\otimes \mathbb{I}_{\mathcal{R}}}_{\infty}=\norm{H}_{\infty}$.

On the other hand, we have
\begin{equation}\label{eq:exact_terms}
\begin{aligned}
&(\mathcal{U}_{\Delta t}\otimes\mathcal{I}_{\mathcal{R}})(\sigma_{\mathcal{K}\mathcal{R}})\\
=&(e^{-i \mathcal{N}(\rho)\Delta t}\otimes \mathbb{I}_{\mathcal{R}})(\sigma_{\mathcal{K}\mathcal{R}})(e^{i \mathcal{N}(\rho)\Delta t}\otimes \mathbb{I}_{\mathcal{R}})\\
=&e^{-i \left(\mathcal{N}(\rho)\otimes \mathbb{I}_{\mathcal{R}}\right)\Delta t}(\sigma_{\mathcal{K}\mathcal{R}})e^{i \left(\mathcal{N}(\rho)\otimes \mathbb{I}_{\mathcal{R}}\right)\Delta t}\\
=&\sum_{m=0}^{\infty}\frac{1}{m!}\operatorname{ad}_{\left(-i \left(\mathcal{N}(\rho)\otimes \mathbb{I}_{\mathcal{R}}\right)\Delta t\right)}^m(\sigma_{\mathcal{K}\mathcal{R}})\\
=&\sigma_{\mathcal{K}\mathcal{R}}-i\Delta t\left[\mathcal{N}(\rho)\otimes \mathbb{I}_{\mathcal{R}},\sigma_{\mathcal{K}\mathcal{R}}\right]\\
&+\sum_{m=2}^{\infty}\frac{1}{m!}\operatorname{ad}_{\left(-i \left(\mathcal{N}(\rho)\otimes \mathbb{I}_{\mathcal{R}}\right)\Delta t\right)}^m(\sigma_{\mathcal{K}\mathcal{R}}).
\end{aligned}
\end{equation}
The trace norm of the third term in the last equation of Eq.~\eqref{eq:exact_terms} satisfies
\begin{equation}\label{eq:exact_remainder}
\begin{aligned}
&\norm{\sum_{m=2}^{\infty}\frac{1}{m!}\operatorname{ad}_{\left(-i \left(\mathcal{N}(\rho)\otimes \mathbb{I}_{\mathcal{R}}\right)\Delta t\right)}^m(\sigma_{\mathcal{K}\mathcal{R}})}_1\\
\le&\sum_{m=2}^{\infty}\frac{1}{m!}(\Delta t)^m\norm{\operatorname{ad}_{\left(\mathcal{N}(\rho)\otimes \mathbb{I}_{\mathcal{R}}\right)}^m(\sigma_{\mathcal{K}\mathcal{R}})}_1\\
\le&\sum_{m=2}^{\infty}\frac{1}{m!}(\Delta t)^m 2^m \norm{\mathcal{N}(\rho)\otimes \mathbb{I}_{\mathcal{R}}}_{\infty}^m\norm{\sigma_{\mathcal{K}\mathcal{R}}}_1\\
=&\sum_{m=2}^{\infty}\frac{1}{m!}(2\norm{\mathcal{N}(\rho)}_{\infty}\Delta t)^m\\
\le&\sum_{m=2}^{\infty}\frac{1}{m!}(2\norm{H}_{\infty}\Delta t)^m\\
\le&4\norm{H}_{\infty}^2\Delta t^2,\text{ when } \norm{H}_{\infty}\Delta t\in(0,0.8],
\end{aligned}
\end{equation}
where the third inequality is given by $\norm{\mathcal{N}(\rho)}_{\infty}\le\norm{H}_{\infty}$. To gain a better understanding of this inequality, we provide a graphical proof in Fig.~\ref{fig:norm_inequality}. Combining Eq.~\eqref{eq:appro_terms}-\eqref{eq:exact_remainder}, we can conclude that when $\norm{H}_{\infty}\Delta t\in(0,0.8]$, it always holds true that
\begin{equation}
\begin{aligned}
&\norm{(\mathcal{Q}_{\Delta t}\otimes\mathcal{I}_{\mathcal{R}})(\sigma_{\mathcal{K}\mathcal{R}})-(\mathcal{U}_{\Delta t}\otimes\mathcal{I}_{\mathcal{R}})(\sigma_{\mathcal{K}\mathcal{R}})}_1\\
\le&\norm{\Tr_1\left[\sum_{m=2}^{\infty}\frac{1}{m!}\operatorname{ad}_{\left(-i (H\otimes \mathbb{I}_{\mathcal{R}})\Delta t\right)}^m(\rho\otimes\sigma_{\mathcal{K}\mathcal{R}})\right]}_1\\
&+\norm{\sum_{m=2}^{\infty}\frac{1}{m!}\operatorname{ad}_{\left(-i \left(\mathcal{N}(\rho)\otimes \mathbb{I}_{\mathcal{R}}\right)\Delta t\right)}^m(\sigma_{\mathcal{K}\mathcal{R}})}_1\\
\le& 8\norm{H}_{\infty}^2\Delta t^2,
\end{aligned}
\end{equation}
regardless of the dimension of the reference system $\mathcal{R}$.
By taking the supremum over $\mathcal{R}$ and $\sigma_{\mathcal{K}\mathcal{R}}$, we have $\norm{\mathcal{Q}_{\Delta t}-\mathcal{U}_{\Delta t}}_{\diamond}\le 8\norm{H}_{\infty}^2\Delta t^2$. Substituting this result into Eq.~\eqref{eq:linearly}, when $K\in[1.25\norm{H}_{\infty}t,\infty)$, we have
\begin{equation}
\begin{aligned}
\norm{\mathcal{Q}_t-\mathcal{U}_t}_{\diamond}
\le 8K\norm{H}_{\infty}^2\Delta t^2=8K^{-1}\norm{H}_{\infty}^2 t^2.
\end{aligned}
\end{equation}
So, in order to guarantee $\norm{\mathcal{Q}_t-\mathcal{U}_t}_{\diamond}\le\epsilon$, it suffices to set
\begin{equation}
\begin{aligned}
K=&\max\left\{\lceil8\epsilon^{-1}\norm{H}_{\infty}^2t^2\rceil,\lceil1.25\norm{H}_{\infty}t\rceil\right\}\\
\le&\mathcal{O}\left(\epsilon^{-1}\norm{H}_{\infty}^2t^2\right).
\end{aligned}
\end{equation}

\begin{figure}[t]
\centering
\includegraphics[width=0.48\textwidth]{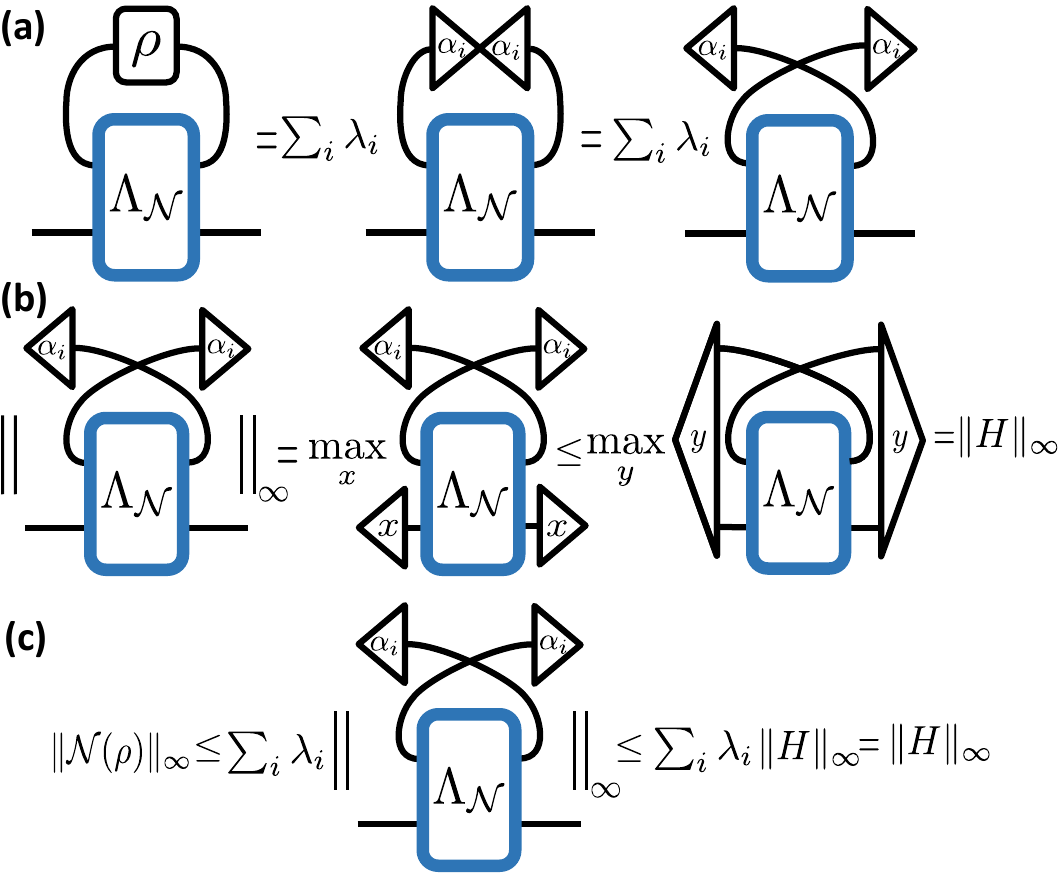}
\caption{Proof of $\norm{\mathcal{N}(\rho)}_{\infty}\le\norm{H}_{\infty}$ using tensor diagrams. The triangles stand for vectors and rounded rectangles stand for matrices. (a) We decompose $\rho=\sum_i\lambda_i\ketbra{\alpha_i}{\alpha_i}$ and represent $\mathcal{N}(\rho)$ with $\Lambda_{\mathcal{N}}^{T_1}$, $\lambda_i$, and $\ket{\alpha_i}$ by dragging two triangles across each other. (b) Using the fact that $\norm{M}_{\infty}=\max\{\bra{s}M\ket{s}:s\in\mathcal{X},\norm{s}=1\}$ for Hermitian matrix $M\in L(\mathcal{X})$, we derive $\norm{\mathcal{N}(\ketbra{\alpha_i}{\alpha_i})}_{\infty}\le\norm{H}_{\infty}$. (c) By applying the convexity of the operator norm, we prove that $\norm{\mathcal{N}(\rho)}_{\infty}\le\norm{H}_{\infty}$.}
\label{fig:norm_inequality}
\end{figure}

\subsection{Sample Complexity Upper Bound for implementing the Partial Transposition Map}\label{app:proof_of_err:PT}

\begin{figure*}[!t]
\centering
\includegraphics[width=0.89\linewidth]{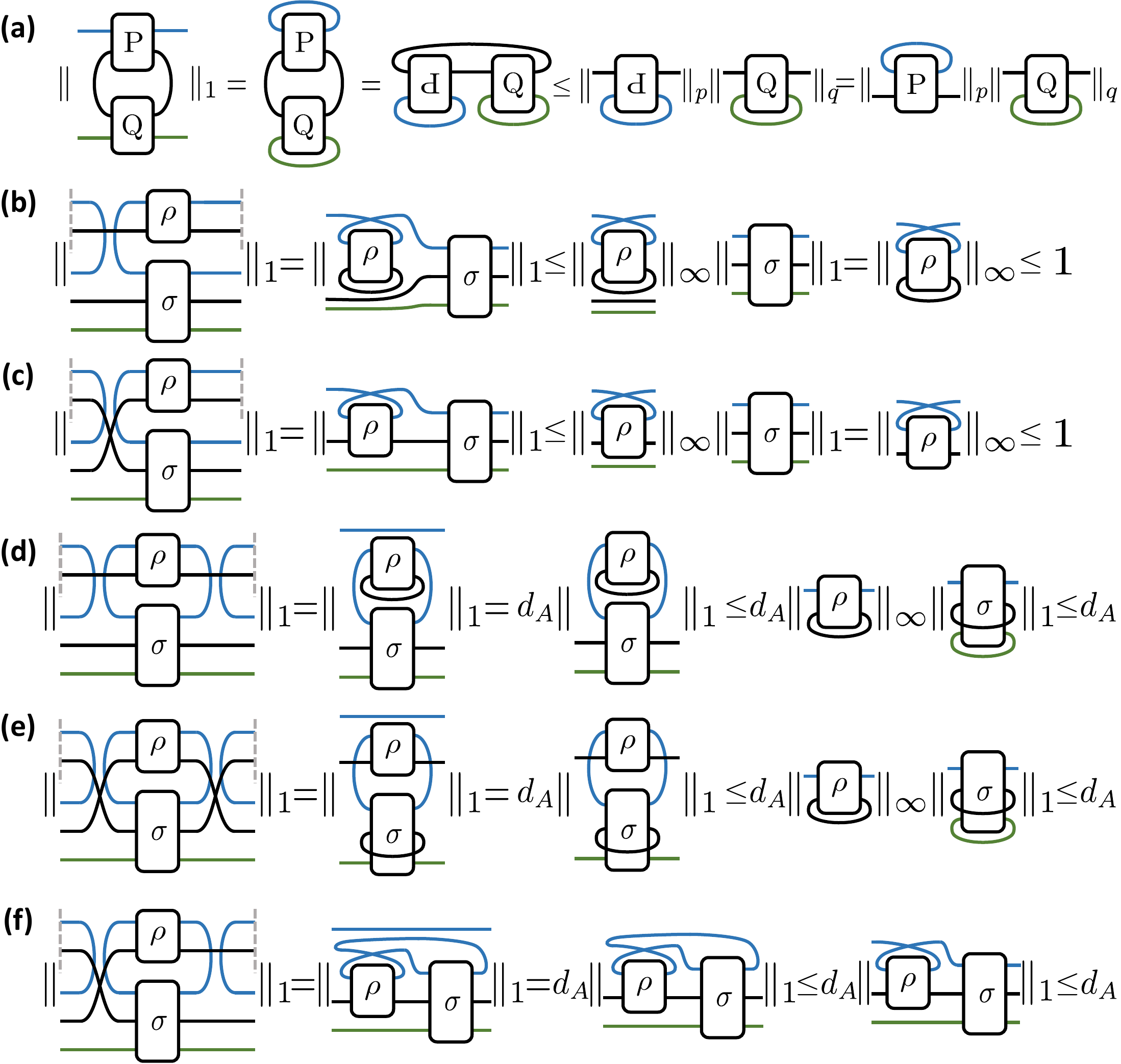}
\caption{(a) A tensor network version of Eq.~\eqref{eq:convolution_norm}, where the inverted $P$ box denotes the original $P$ box after transposition and exchanging the upper and lower legs. The gray dashed lines denote the trace function, which is represented by connecting legs in tensor network diagrams. (b)-(f) are diagrams for the proof of Lemma~\ref{lemma:PT_single_remainder_norm}, divided into five cases according to the values of $k_1$ and $k_2$. (b) $k_1$ even, $k_2=0$. The last inequality holds because $\Tr_2\left(\rho\right)^T$ is still a quantum state. (c) $k_1$ odd, $k_2=0$. The last inequality holds because $\norm{\rho^{T_A}}_{\infty}\le\sum_i\lambda_i\norm{\ketbra{\alpha_i}{\alpha_i}^{T_A}}_{\infty}\le\sum_i\lambda_i=1$, where $\rho=\sum_i\lambda_i\ketbra{\alpha_i}{\alpha_i}$ is the spectral decomposition, and the eigenvalues of $\ketbra{\alpha_i}{\alpha_i}^{T_A}$ are in the range $[-1/2,1]$ \cite{Johnston2018inverse}. (d) $k_1$ even, $k_2$ even. The first inequality is given by Lemma~\ref{lemma:convolution_norm}. (e) $k_1$ odd, $k_2$ odd. The first inequality is given by Lemma~\ref{lemma:convolution_norm}. (f) $k_1$ odd, $k_2$ even. The first inequality is given by the non-increasing property of trace norm under partial trace. The last equality follows the derivation in (c). Other cases, such as `$k_1=0$, $k_2$ odd', can be derived by taking the Hermitian conjugate in the above derivations.}
\label{fig:PTcost+ConvolutionNorm}
\end{figure*}

As stated in the main text, Theorem~\ref{theorem:cost} does not provide a tight upper bound on the sample complexity of HME for all non-physical maps. Here, we will use the partial transposition map to illustrate this.
\begin{proposition}\label{prop:PT_cost}
Setting the Hamiltonian to be $H_P=\Phi^+_A\otimes S_B$, the HME algorithm shown in Fig.~\ref{fig:overview} requires at most $\mathcal{O}(\epsilon^{-1}d_At^2)$ copies of sequentially inputting state $\rho$ to ensure that 
\begin{equation}
\norm{\mathcal{Q}_t-[e^{-i\rho^{T_A}t}]}_\diamond\le \epsilon
\end{equation}
holds for arbitrary $\rho$.
\end{proposition}

To prove Proposition~\ref{prop:PT_cost}, we need to make some preliminary preparations.

When $k\ge 1$, the expression for $H_P^k$ can be written as:
\begin{equation}\label{eq:H_P^k}
H_P^k=\begin{cases}d_A^{k-1}\Phi_A^+\otimes \mathbb{I}_B & \text{ if } k \text{ even,} \\ d_A^{k-1}\Phi_A^+\otimes S_B & \text{ if } k \text{ odd,}\end{cases}
\end{equation}
where $\mathbb{I}_B$ represents the identity operator acting on two
copies of the system $B$.

For two operators $P\in L({\mathcal{H}_1\otimes\mathcal{H}_2})$ and $Q\in L({\mathcal{H}_2\otimes\mathcal{H}_3})$, we define the \textit{convolution} of $P$ and $Q$ as 
\begin{equation}
P*Q:=(\mathbb{I}_{\mathcal{H}_1}\otimes\bra{\Phi^+}\otimes \mathbb{I}_{\mathcal{H}_3})(P\otimes Q)(\mathbb{I}_{\mathcal{H}_1}\otimes\ket{\Phi^+}\otimes \mathbb{I}_{\mathcal{H}_3}),
\end{equation}
where $\ket{\Phi^{+}}$ represents the unnormalized maximally entangled state of $\mathcal{H}_2\otimes\mathcal{H}_2$ in the computational basis. The following lemma provides an inequality concerning the norm of the convolution of two positive operators.

\begin{lemma}\label{lemma:convolution_norm}
For positive operators $P\in L({\mathcal{H}_1\otimes\mathcal{H}_2})$ and $Q\in L({\mathcal{H}_2\otimes\mathcal{H}_3})$, the trace norm of $P*Q$ satisfies
\begin{equation}
\begin{aligned}
\norm{P*Q}_1\le\norm{\Tr_{\mathcal{H}_1}(P)}_p\norm{\Tr_{\mathcal{H}_3}(Q)}_q,
\end{aligned}
\end{equation}
where $p,q\in\left[1,\infty\right]$ satisfy $\frac{1}{p}+\frac{1}{q}=1$.
\end{lemma}
\begin{proof}
$P,Q\ge 0$ implies $P*Q\ge 0$. We have
\begin{equation}\label{eq:convolution_norm}
\begin{aligned}
&\norm{P*Q}_1=\Tr_{\mathcal{H}_1,\mathcal{H}_3}\left(P*Q\right)=\Tr_{\mathcal{H}_1}\left(P\right)*\Tr_{\mathcal{H}_3}\left(Q\right)\\
=&\Tr\left(\Tr_{\mathcal{H}_1}\left(P\right)^T\Tr_{\mathcal{H}_3}\left(Q\right)\right)=\Tr\left(\overline{\Tr_{\mathcal{H}_1}\left(P\right)}^{\dagger}\Tr_{\mathcal{H}_3}\left(Q\right)\right)\\
\le&\norm{\overline{\Tr_{\mathcal{H}_1}(P)}}_p\norm{\Tr_{\mathcal{H}_3}(Q)}_q=\norm{\Tr_{\mathcal{H}_1}(P)}_p\norm{\Tr_{\mathcal{H}_3}(Q)}_q,
\end{aligned}
\end{equation}
where $\overline{M}$ denotes the complex conjugate of an operator $M$. The first equality holds due to the positivity of $P*Q$. The inequality is given by the \textit{Hölder inequality} \cite{watrous2018theory}. A tensor network representation of Eq.~\eqref{eq:convolution_norm} is shown in Fig.~\ref{fig:PTcost+ConvolutionNorm}(a).
\end{proof}

Taking advantage of the structure of $H_P$, we can provide a more refined estimation for the third term in the last equation of Eq.~\eqref{eq:appro_terms}.

\begin{lemma}\label{lemma:PT_single_remainder_norm}
For non-negative integers $k_1$ and $k_2$ such that $k_1+k_2\ge 1$, we have 
\begin{equation}
\begin{aligned}
&\norm{\Tr_1\left[\left(H_P\otimes \mathbb{I}_{\mathcal{R}}\right)^{k_1}(\rho\otimes\sigma_{\mathcal{K}\mathcal{R}})\left(H_P\otimes \mathbb{I}_{\mathcal{R}}\right)^{k_2})\right]}_1\le d_A^{k_1+k_2-1}
\end{aligned}
\end{equation}
for any quantum states $\rho$ and $\sigma_{\mathcal{K}\mathcal{R}}$.
\end{lemma}
\begin{proof}
When $k_1$ and $k_2$ are odd, we have
\begin{equation}
\begin{aligned}
&\left(H_P\otimes \mathbb{I}_{\mathcal{R}}\right)^{k_1}(\rho\otimes\sigma_{\mathcal{K}\mathcal{R}})\left(H_P\otimes \mathbb{I}_{\mathcal{R}}\right)^{k_2}\\
=&d_A^{k_1+k_2-2}(\Phi_A^+\otimes S_B\otimes \mathbb{I}_{\mathcal{R}})(\rho\otimes\sigma_{\mathcal{K}\mathcal{R}})(\Phi_A^+\otimes S_B\otimes \mathbb{I}_{\mathcal{R}}).
\end{aligned}
\end{equation}
Taking partial trace over the first system and estimating the trace norm, we have
\begin{equation}\label{eq:PT_cost_odd_odd}
\begin{aligned}
&\norm{\Tr_1\left[(\Phi_A^+\otimes S_B\otimes \mathbb{I}_{\mathcal{R}})(\rho\otimes\sigma_{\mathcal{K}\mathcal{R}})(\Phi_A^+\otimes S_B\otimes \mathbb{I}_{\mathcal{R}})\right]}_1\\
&=\norm{\mathbb{I}_{d_A}\otimes\big(S\rho S*\Tr_B(\sigma_{\mathcal{K}\mathcal{R}})\big)}_1\\
&=d_A\norm{S\rho S*\Tr_B(\sigma_{\mathcal{K}\mathcal{R}})}_1\\
&\le d_A\norm{\Tr_{B}(\rho)}_{\infty}\norm{\Tr_{B,\mathcal{R}}(\sigma)}_{1}\le d_A.
\end{aligned}
\end{equation}
In the second and third lines of Eq.~\eqref{eq:PT_cost_odd_odd}, the symbol $S$ denotes the SWAP operator between the two systems on which the density matrix $\rho$ acts. The tensor network derivation of this expression can be found in Fig.~\ref{fig:PTcost+ConvolutionNorm}(e). The proofs for other cases follow a similar approach and are provided in Fig.~\ref{fig:PTcost+ConvolutionNorm}.
\end{proof}

\begin{proof}[Proof of Proposition~\ref{prop:PT_cost}]
The trace norm of the third term in the last equation of Eq.~\eqref{eq:appro_terms} can be estimated as
\begin{equation}\label{eq:PT_appro_remainders}
\begin{aligned}
&\norm{\Tr_1\left[\sum_{m=2}^{\infty}\frac{1}{m!}\operatorname{ad}_{\left(-i (H_P\otimes \mathbb{I}_{\mathcal{R}})\Delta t\right)}^m(\rho\otimes\sigma_{\mathcal{K}\mathcal{R}})\right]}_1\\
\le &\sum_{m=2}^{\infty}\frac{1}{m!}(\Delta t)^m\norm{\Tr_1\left[\operatorname{ad}_{\left(H_P\otimes \mathbb{I}_{\mathcal{R}}\right)}^m(\rho\otimes\sigma_{\mathcal{K}\mathcal{R}})\right]}_1\\
\le&\sum_{m=2}^{\infty}\frac{1}{m!}(2\Delta t)^md_A^{m-1}\\
=&d_A^{-1}\sum_{m=2}^{\infty}\frac{1}{m!}(2d_A\Delta t)^m\\
=&\mathcal{O}(d_A\Delta t^2),\\
\end{aligned}
\end{equation}
where the second inequality is given by
\begin{equation}
\norm{\Tr_1\left[\left(H_P\otimes \mathbb{I}_{\mathcal{R}}\right)^k(\rho\otimes\sigma_{\mathcal{K}\mathcal{R}})\left(H_P\otimes \mathbb{I}_{\mathcal{R}}\right)^{m-k})\right]}_1\le d_A^{m-1},
\end{equation}
which can be derived from Lemma~\ref{lemma:PT_single_remainder_norm} and holds for $0\le k \le m$.

On the other hand, we have
\begin{equation}\label{eq:PT_exact_remainders}
\begin{aligned}
&\norm{\sum_{m=2}^{\infty}\frac{1}{m!}\operatorname{ad}_{\left(-i \left(\rho^{T_A}\otimes \mathbb{I}_{\mathcal{R}}\right)\Delta t\right)}^m(\sigma_{\mathcal{K}\mathcal{R}})}_1\\
\le&\mathcal{O}(\norm{\rho^{T_A}}_{\infty}^2\Delta t^2)= \mathcal{O}(\Delta t^2).
\end{aligned}
\end{equation}

Combining Eq.~\eqref{eq:appro_terms} \eqref{eq:exact_terms} \eqref{eq:PT_appro_remainders} \eqref{eq:PT_exact_remainders}, we can conclude that in order to guarantee $\norm{\mathcal{Q}_t-\left[e^{-i\rho^{T_A}t}\right]}_{\diamond}\le\epsilon$, the sample complexity $K$ we need is at most $\mathcal{O}(\epsilon^{-1}d_At^2)$.
\end{proof}

Note that the proof of Proposition~\ref{prop:PT_cost} can be easily generalized to the controlled version. Therefore, we need at most $\mathcal{O}(\epsilon^{-1}d_At^2)$ copies of $\rho$ to implement controlled-$e^{-i\rho^{T_A}t}$ up to error $\epsilon$ in diamond distance.

\section{Robustness of HME}
\subsection{Proof of Theorem~\ref{theorem:robustness}}\label{app:proof_of_robustness}

Likewise the proof of Theorem~\ref{theorem:cost} in Appendix~\ref{app:Proof_of_Thm2}, to derive the distance between the whole channel $\norm{\mathcal{Q}_t^\prime-\mathcal{Q}_t}$, we need to first derive the distance of a step $\norm{\mathcal{Q}_{k,\Delta t}^\prime-\mathcal{Q}_{\Delta t}}$, where 
\begin{equation}
\mathcal{Q}_{k,\Delta t}^\prime(\cdot):=\Tr_1\left[(e^{-i H^\prime_k\Delta t}\otimes \mathbb{I}_{\mathcal{R}})(\rho^{\prime}_k\otimes\cdot)(e^{i H^\prime_k\Delta t}\otimes \mathbb{I}_{\mathcal{R}})\right]
\end{equation}
with $H_k^\prime$ and $\rho_k^\prime$ varying for each step, and
\begin{equation}
\begin{aligned}
\mathcal{Q}_t^\prime=\mathcal{Q}_{K,\Delta t}^\prime\circ\mathcal{Q}_{K-1,\Delta t}^\prime\circ\cdots\circ\mathcal{Q}_{1,\Delta t}^\prime.
\end{aligned}
\end{equation}
Then, for arbitrary reference system $\mathcal{R}$ and any state $\sigma_{\mathcal{K}\mathcal{R}}\in D(\mathcal{K}\mathcal{R})$, we have
\begin{widetext}
\begin{equation}\label{eq:robustness_terms}
\begin{aligned}
&\Big\|\mathcal{Q}^{\prime}_{k,\Delta t}\otimes\mathcal{I}_{\mathcal{R}}(\sigma_{\mathcal{K}\mathcal{R}})-\mathcal{Q}_{\Delta t}\otimes\mathcal{I}_{\mathcal{R}}(\sigma_{\mathcal{K}\mathcal{R}})\Big\|_1\\
=&\norm{\Tr_1\big[(e^{-i H^\prime_k\Delta t}\otimes \mathbb{I}_{\mathcal{R}})(\rho^{\prime}_k\otimes\sigma_{\mathcal{K}\mathcal{R}})(e^{i H^\prime_k\Delta t}\otimes \mathbb{I}_{\mathcal{R}})\big]-\Tr_1\big[(e^{-i H\Delta t}\otimes \mathbb{I}_{\mathcal{R}})(\rho\otimes\sigma_{\mathcal{K}\mathcal{R}})(e^{i H\Delta t}\otimes \mathbb{I}_{\mathcal{R}})\big]}_1\\
\le&\norm{\Tr_1\big[(e^{-i H^\prime_k\Delta t}\otimes \mathbb{I}_{\mathcal{R}})(\rho^{\prime}_k\otimes\sigma_{\mathcal{K}\mathcal{R}})(e^{i H^\prime_k\Delta t}\otimes \mathbb{I}_{\mathcal{R}})\big]-\Tr_1\big[(e^{-i H\Delta t}\otimes \mathbb{I}_{\mathcal{R}})(\rho^{\prime}_k\otimes\sigma_{\mathcal{K}\mathcal{R}})(e^{iH\Delta t}\otimes \mathbb{I}_{\mathcal{R}})\big]}_1\\
&+\Big\|\Tr_1\big[(e^{-i H\Delta t}\otimes \mathbb{I}_{\mathcal{R}})(\rho^{\prime}_k\otimes\sigma_{\mathcal{K}\mathcal{R}})(e^{i H\Delta t}\otimes \mathbb{I}_{\mathcal{R}})\big]-\Tr_1\big[(e^{-i H\Delta t}\otimes \mathbb{I}_{\mathcal{R}})(\rho\otimes\sigma_{\mathcal{K}\mathcal{R}})(e^{i H\Delta t}\otimes \mathbb{I}_{\mathcal{R}})\big]\Big\|_1\\
\le&\norm{(e^{-i H^\prime_k\Delta t}\otimes \mathbb{I}_{\mathcal{R}})(\rho^{\prime}_k\otimes\sigma_{\mathcal{K}\mathcal{R}})(e^{i H^\prime_k\Delta t}\otimes \mathbb{I}_{\mathcal{R}})-(e^{-i H\Delta t}\otimes \mathbb{I}_{\mathcal{R}})(\rho^{\prime}_k\otimes\sigma_{\mathcal{K}\mathcal{R}})(e^{i H\Delta t}\otimes \mathbb{I}_{\mathcal{R}})}_1\\
&+\Big\|{\Tr_1\left[(e^{-i H\Delta t}\otimes \mathbb{I}_{\mathcal{R}})\left((\rho^{\prime}_k-\rho)\otimes\sigma_{\mathcal{K}\mathcal{R}}\right)(e^{i H\Delta t}\otimes \mathbb{I}_{\mathcal{R}})\right]}\Big\|_1\\
\le&\norm{[e^{-iH^\prime_k\Delta t}]-[e^{-iH\Delta t}]}_{\diamond}+\Big\|\Tr_1\left[(e^{-i H\Delta t}\otimes \mathbb{I}_{\mathcal{R}})\left((\rho^{\prime}_k-\rho)\otimes\sigma_{\mathcal{K}\mathcal{R}}\right)(e^{i H\Delta t}\otimes \mathbb{I}_{\mathcal{R}})\right]\Big\|_1,
\end{aligned} 
\end{equation}
\end{widetext}
where the last inequality is given by the definition of the diamond norm. Consider the first term of the last line, we have
\begin{equation}\label{eq:robustness_terms_1}
\begin{aligned}
&\norm{[e^{-iH^\prime_k\Delta t}]-[e^{-iH\Delta t}]}_{\diamond}\\
\le&2\Delta t\norm{H^\prime_k-H}_{\infty}\exp\Big(\Delta t\max\left\{\norm{H^\prime_k}_{\infty},\norm{H}_{\infty}\right\}\Big)\\
\le&2\sqrt{e}\Delta t\norm{H^\prime_k-H}_{\infty},
\end{aligned}
\end{equation}
where the first inequality is given by Lemma~\ref{lemma:DiamondDistance_HamiltonianEvolution}, and the second inequality holds when
\begin{equation}
\begin{aligned}
\Delta t\max\left\{\norm{H^\prime_k}_{\infty},\norm{H}_{\infty}\right\}\in(0,0.5].
\end{aligned}
\end{equation}
The second term of the last equation of Eq.~\eqref{eq:robustness_terms} satisfies
\begin{equation}\label{eq:robustness_terms_2}
\begin{aligned}
&\Big\|\Tr_1\left[(e^{-i H\Delta t}\otimes \mathbb{I}_{\mathcal{R}})\left((\rho^{\prime}_k-\rho)\otimes\sigma_{\mathcal{K}\mathcal{R}}\right)(e^{i H\Delta t}\otimes \mathbb{I}_{\mathcal{R}})\right]\Big\|_1\\
=&\norm{\sum_{m=0}^{\infty}\frac{1}{m!}\Tr_1\left[\operatorname{ad}_{\left(-i (H\otimes \mathbb{I}_{\mathcal{R}})\Delta t\right)}^m((\rho^{\prime}_k-\rho)\otimes\sigma_{\mathcal{K}\mathcal{R}})\right]}_1\\
\le&\norm{\Tr_1\left[(\rho^{\prime}_k-\rho)\otimes\sigma_{\mathcal{K}\mathcal{R}})\right]}_1\\
&+\sum_{m=1}^{\infty}\frac{1}{m!}\norm{\operatorname{ad}_{\left(-i (H\otimes \mathbb{I}_{\mathcal{R}})\Delta t\right)}^m((\rho^{\prime}_k-\rho)\otimes\sigma_{\mathcal{K}\mathcal{R}})}_1\\
\le&\sum_{m=1}^{\infty}\frac{1}{m!}(2\Delta t\norm{H}_{\infty})^m\norm{(\rho^{\prime}_k-\rho)\otimes\sigma_{\mathcal{K}\mathcal{R}}}_1\\
=&\norm{\rho^{\prime}_k-\rho}_1\left(e^{2\Delta t\norm{H}_{\infty}}-1\right)\\
\le&4\Delta t\norm{H}_{\infty}\norm{\rho^{\prime}_k-\rho}_1,
\end{aligned}
\end{equation}
where the last inequality holds when $\Delta t\norm{H}_{\infty}\in(0,0.5]$. 

Combining Eq.~\eqref{eq:robustness_terms} \eqref{eq:robustness_terms_1} \eqref{eq:robustness_terms_2}, we can conclude that
\begin{equation}
\begin{aligned}
&\norm{\mathcal{Q}^{\prime}_{k,\Delta t}\otimes\mathcal{I}_{\mathcal{R}}(\sigma_{\mathcal{K}\mathcal{R}})-\mathcal{Q}_{\Delta t}\otimes\mathcal{I}_{\mathcal{R}}(\sigma_{\mathcal{K}\mathcal{R}})}_1\\
\le&\Delta t\left(2\sqrt{e}\norm{H^\prime_k-H}_{\infty}+4\norm{H}_{\infty}\norm{\rho^{\prime}_k-\rho}_1\right),
\end{aligned}
\end{equation}
which holds for all choices of $\mathcal{R}$ and $\sigma_{\mathcal{K}\mathcal{R}}\in D(\mathcal{K}\mathcal{R})$. Thus we have
\begin{equation}
\begin{aligned}
&\norm{\mathcal{Q}^{\prime}_{k,\Delta t}-\mathcal{Q}_{\Delta t}}_{\diamond}\\
\le&\Delta t\left(2\sqrt{e}\norm{H^\prime_k-H}_{\infty}+4\norm{H}_{\infty}\norm{\rho^{\prime}_k-\rho}_1\right).
\end{aligned}
\end{equation}
By the subadditivity property of diamond distance (Lemma~\ref{lemma:subadditivity}), we have
\begin{equation}\label{eq:Q_distance}
\begin{aligned}
&\norm{\mathcal{Q}^{\prime}_{t}-\mathcal{Q}_{t}}_{\diamond}\le\sum_{k=1}^K\norm{\mathcal{Q}^{\prime}_{k,\Delta t}-\mathcal{Q}_{\Delta t}}_{\diamond}\\
\le&t\left(2\sqrt{e}\frac{1}{K}\sum_{k=1}^K\norm{H^{\prime}_k-H}_{\infty}+4\norm{H}_{\infty}\frac{1}{K}\sum_{k=1}^K\norm{\rho^{\prime}_k-\rho}_1\right)\\
\le&4t\left(D_H+\norm{H}_{\infty}D_S\right).
\end{aligned}
\end{equation}

According to the derivation, the requirement for the validity of Eq.~\eqref{eq:Q_distance} is that
\begin{equation}
\frac{t}{K}\max\Big\{\norm{H}_{\infty},\norm{H^\prime_1}_{\infty},\cdots,\norm{H^\prime_K}_{\infty}\Big\}\le 0.5,
\end{equation}
that is, $K\ge2t\norm{H}_{\infty}$ and $K\ge2t\norm{H^\prime_k}_{\infty}$ for $k=1,\cdots,K$.

\subsection{Robustness for implementing the Partial Transposition Map}\label{app:proof_of_robustness:PT}

\begin{figure*}[!t]
\centering
\includegraphics[width=0.95\linewidth]{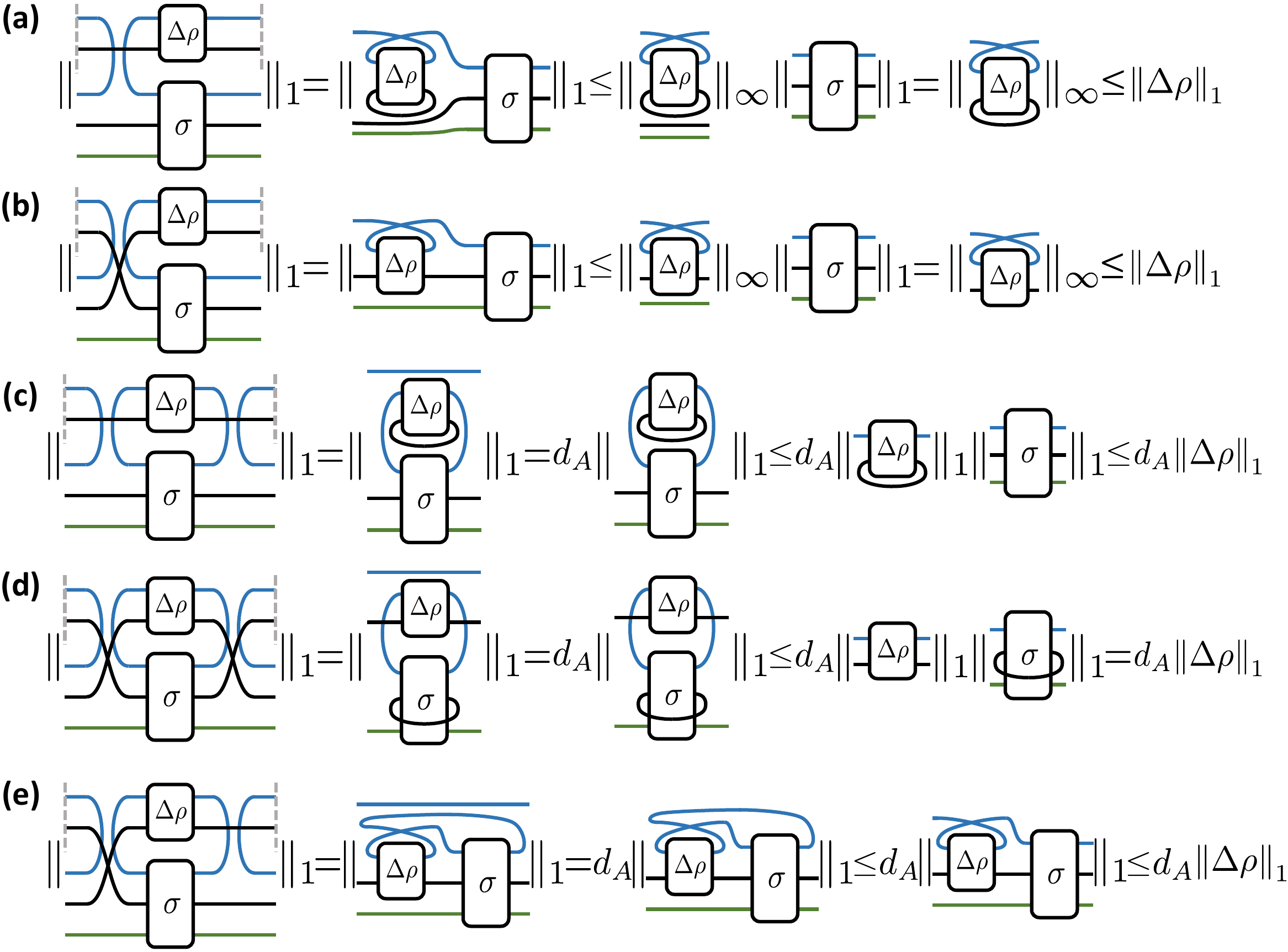}
\caption{Diagrams for the proof of Lemma~\ref{lemma:PT_robustness_single_remainder_norm}, divided into five cases according to the values of $k_1$ and $k_2$. We denote $\rho_1-\rho_2$ by $\Delta\rho$. Other cases, such as `$k_1=0$, $k_2$ odd' can be derived by taking the Hermitian conjugate of the derivations listed in the figure. (a) $k_1$ even, $k_2=0$. The last inequality is given by $\norm{\Tr_2(\Delta \rho)^{T}}_{\infty}=\norm{\Tr_2(\Delta \rho)}_{\infty}\le\norm{\Tr_2(\Delta \rho)}_{1}\le\norm{\Delta \rho}_{1}$. (b) $k_1$ odd, $k_2=0$. The last inequality holds because $\norm{(\Delta\rho)^{T_A}}_{\infty}\le\sum_i\abs{\lambda_i}\norm{\ketbra{\alpha_i}{\alpha_i}^{T_A}}_{\infty}\le\sum_i\abs{\lambda_i}=\norm{\Delta\rho}_1$, where $\Delta\rho=\sum_i\lambda_i\ketbra{\alpha_i}{\alpha_i}$ is the spectral decomposition of $\Delta\rho$. (c) $k_1$ even, $k_2$ even. The first inequality is given by Lemma~\ref{lemma:convolution_norm_111}. (d) $k_1$ odd, $k_2$ odd. The inequality is given by Lemma~\ref{lemma:convolution_norm_111}. (f) $k_1$ odd, $k_2$ even. The last inequality follows the derivation in (b).}
\label{fig:PTroubstness}
\end{figure*}

For certain specific Hermitian-preserving maps, we can obtain more accurate estimations of the robustness of HME. Again we use the partial transposition map as an example to illustrate this.

\begin{proposition}\label{proposition:PT_robuseness}
Setting the Hamiltonian to be $H_P=\Phi^+_A\otimes S_B$, following the noise assumptions in Theorem~\ref{theorem:robustness}, we have
\begin{equation}
\norm{\mathcal{Q}_t^{\prime}-\mathcal{Q}_t}_{\diamond}\le4t\left(D_H+D_S\right).
\end{equation}
\end{proposition}

To prove this proposition, we first need to present two lemmas, similar to what we have done in Appendix~\ref{app:proof_of_err:PT}.

\begin{lemma}\label{lemma:convolution_norm_111}
For $P\in L({\mathcal{H}_1\otimes\mathcal{H}_2})$ and $Q\in L({\mathcal{H}_2\otimes\mathcal{H}_3})$, we have
\begin{equation}
\norm{P*Q}_1\le\norm{P}_1\norm{Q}_1.
\end{equation}
\end{lemma}
\begin{proof}
Let $P=\sum_i\lambda_i\ketbra{\alpha_i}{\alpha_i}$ and $Q=\sum_i\mu_i\ketbra{\beta_i}{\beta_i}$ be the spectral decomposition of $P$ and $Q$ respectively. We have
\begin{equation}
\begin{aligned}
\norm{P*Q}_1\le&\sum_{i,j}\abs{\lambda_i}\cdot\abs{\mu_j}\cdot\norm{\ketbra{\alpha_i}{\alpha_i}*\ketbra{\beta_j}{\beta_j}}_1\\
\le&\sum_{i,j}\abs{\lambda _i}\cdot\abs{\mu_j}=\norm{P}_1\norm{Q}_1,
\end{aligned}
\end{equation}
where the second inequality can be derived by Lemma~\ref{lemma:convolution_norm}.
\end{proof}

\begin{lemma}\label{lemma:PT_robustness_single_remainder_norm}
For non-negative integers $k_1$ and $k_2$ such that $k_1+k_2\ge 1$, we have
\begin{equation}
\begin{aligned}
&\norm{\Tr_1\left[\left(H_P\otimes \mathbb{I}_{\mathcal{R}}\right)^{k_1}((\rho_1-\rho_2)\otimes\sigma_{\mathcal{K}\mathcal{R}})\left(H_P\otimes \mathbb{I}_{\mathcal{R}}\right)^{k_2})\right]}_1\\
&\le d_A^{k_1+k_2-1}\norm{\rho_1-\rho_2}_1
\end{aligned}
\end{equation}
for any quantum states $\rho_1$, $\rho_2$ and $\sigma_{\mathcal{K}\mathcal{R}}$.
\end{lemma}
\begin{proof}
Denote $\rho_1-\rho_2$ by $\Delta\rho$ for brevity. By Eq.~\eqref{eq:H_P^k}, when $k_1$ and $k_2$ are odd, we have
\begin{equation}
\begin{aligned}
&\left(H_P\otimes \mathbb{I}_{\mathcal{R}}\right)^{k_1}(\Delta\rho\otimes\sigma_{\mathcal{K}\mathcal{R}})\left(H_P\otimes \mathbb{I}_{\mathcal{R}}\right)^{k_2}\\
=&d_A^{k_1+k_2-2}(\Phi_A^+\otimes S_B\otimes \mathbb{I}_{\mathcal{R}})(\Delta\rho\otimes\sigma_{\mathcal{K}\mathcal{R}})(\Phi_A^+\otimes S_B\otimes \mathbb{I}_{\mathcal{R}}).
\end{aligned}
\end{equation}
Taking partial trace on the first system and estimating the trace norm, we have
\begin{equation}\label{eq:PT_robustness_odd_odd}
\begin{aligned}
&\norm{\Tr_1\left[(\Phi_A^+\otimes S_B\otimes \mathbb{I}_{\mathcal{R}})(\Delta\rho\otimes\sigma_{\mathcal{K}\mathcal{R}})(\Phi_A^+\otimes S_B\otimes \mathbb{I}_{\mathcal{R}})\right]}_1\\
&=\norm{\mathbb{I}_{d_A}\otimes[S\Delta\rho S*\Tr_B(\sigma_{\mathcal{K}\mathcal{R}})]}_1\\
&=d_A\norm{[S\Delta\rho S*\Tr_B(\sigma_{\mathcal{K}\mathcal{R}})]}_1\\
&\le d_A\norm{S\Delta\rho S}_{1}\norm{\Tr_{B}(\sigma)}_{1}\le d_A\norm{\Delta\rho}_1,
\end{aligned}
\end{equation}
where the symbol $S$ denotes the SWAP operator between the two systems on which the density matrix $\rho$ acts. The tensor network derivation of this expression can be found in Fig.~\ref{fig:PTroubstness}(d). The proofs of the other cases follow a similar approach and are presented in Fig.~\ref{fig:PTroubstness}.
\end{proof}

\begin{proof}[Proof of Proposition~\ref{proposition:PT_robuseness}]
The second term of the last
equation in Eq.~\eqref{eq:robustness_terms} satisfies
\begin{equation}\label{eq:robustness_terms_2_PT}
\begin{aligned}
&\Big\|\Tr_1\left[(e^{-i H_P\Delta t}\otimes \mathbb{I}_{\mathcal{R}})\left((\rho^{\prime}_k-\rho)\otimes\sigma_{\mathcal{K}\mathcal{R}}\right)(e^{i H_P\Delta t}\otimes \mathbb{I}_{\mathcal{R}})\right]\Big\|_1\\
&\le\sum_{m=1}^{\infty}\frac{1}{m!}\norm{\Tr_1\left[\operatorname{ad}_{\left(-i (H_P\otimes \mathbb{I}_{\mathcal{R}})\Delta t\right)}^m((\rho^{\prime}_k-\rho)\otimes\sigma_{\mathcal{K}\mathcal{R}})\right]}_1\\
&\le\sum_{m=1}^{\infty}\frac{1}{m!}(\Delta t)^m\norm{\Tr_1\left[\operatorname{ad}_{(H_P\otimes \mathbb{I}_{\mathcal{R}})}^m((\rho^{\prime}_k-\rho)\otimes\sigma_{\mathcal{K}\mathcal{R}})\right]}_1\\
&\le\sum_{m=1}^{\infty}\frac{1}{m!}(2\Delta t)^md_A^{m-1}\norm{\rho^{\prime}_k-\rho}_1\\
&=\frac{1}{d_A}\norm{\rho^{\prime}_k-\rho}_1\left(e^{2d_A\Delta t}-1\right)\\
&\le4\Delta t\norm{\rho^{\prime}_k-\rho}_1,
\end{aligned}
\end{equation}
where the third inequality is given by Lemma~\ref{lemma:PT_robustness_single_remainder_norm}, and the last inequality holds when $d_A\Delta t\in(0,0.5]$.

Combining Eq.~\eqref{eq:robustness_terms}, \eqref{eq:robustness_terms_1} and \eqref{eq:robustness_terms_2_PT}, we obtain
\begin{equation}
\begin{aligned}
&\norm{\mathcal{Q}^{\prime}_{k,\Delta t}\otimes\mathcal{I}_{\mathcal{R}}(\sigma_{\mathcal{K}\mathcal{R}})-\mathcal{Q}_{\Delta t}\otimes\mathcal{I}_{\mathcal{R}}(\sigma_{\mathcal{K}\mathcal{R}})}_1\\
\le&\Delta t\left(2\sqrt{e}\norm{H^\prime_k-H}_{\infty}+4\norm{\rho^{\prime}_k-\rho}_1\right),
\end{aligned}
\end{equation}
which holds for arbitrary $\mathcal{R}$ and $\sigma_{\mathcal{K}\mathcal{R}}\in D(\mathcal{K}\mathcal{R})$. Thus by the definition and the subadditivity property of diamond distance, we have
\begin{equation}
\begin{aligned}
&\norm{\mathcal{Q}^{\prime}_{t}-\mathcal{Q}_{t}}_{\diamond}\le\sum_{k=1}^K\norm{\mathcal{Q}^{\prime}_{k,\Delta t}-\mathcal{Q}_{\Delta t}}_{\diamond}\\
\le&t\left(2\sqrt{e}\frac{1}{K}\sum_{k=1}^K\norm{H^{\prime}_k-H}_{\infty}+4\frac{1}{K}\sum_{k=1}^K\norm{\rho^{\prime}_k-\rho}_1\right)\\
\le&4t\left(D_H+D_S\right).
\end{aligned}
\end{equation}

\end{proof}

\section{Analysis of Optimality}

\subsection{Proof of Theorem~\ref{theorem:LowerBound}}\label{app:proof_of_LB}

We will utilize the following two lemmas in our proof.

\begin{lemma}[Holevo-Helstrom theorem \cite{Holevo1973decision,Helstrom1969detection}]\label{lemma:Holevo-Helstrom}
Given two states, $\rho_0,\rho_1\in D(\mathcal{H})$, for any two-outcomes POVM $\{M_0,M_1\}$, we have
\begin{equation}\label{eq:Holevo-Helstrom}
\frac{1}{2}\Tr(M_0\rho_0)+\frac{1}{2}\Tr(M_1\rho_1)\le\frac{1}{2}+\frac{1}{4}\norm{\rho_0-\rho_1}_1.
\end{equation}
Furthermore, equality can be achieved by an appropriate choice of ${M_0,M_1}$.
\end{lemma}

\begin{lemma}[Channel analogue of Holevo-Helstrom theorem \cite{watrous2018theory}]\label{lemma:Holevo-Helstrom_channel}
Let $\mathcal{Q}_0,\mathcal{Q}_1\in T(\mathcal{X},\mathcal{Y})$ be two CPTP maps. For any ancillary system $\mathcal{R}$, any two-outcomes POVM $\{M_0,M_1\}$ on $\mathcal{Y}\otimes\mathcal{R}$, and any input state $\sigma\in D(\mathcal{X}\otimes\mathcal{R})$, we have
\begin{equation}\label{eq:Holevo-Helstrom_channel}
\begin{aligned}
\frac{1}{2}\Tr\big(M_0\left(\mathcal{Q}_0 \otimes \mathcal{I}_{\mathcal{R}}\right)(\sigma)\big)+\frac{1}{2}\Tr\big(M_1\left(\mathcal{Q}_1 \otimes \mathcal{I}_{\mathcal{R}}\right)(\sigma)\big) \\
\le \frac{1}{2}+\frac{1}{4}\norm{\mathcal{Q}_0- \mathcal{Q}_1}_{\diamond}.
\end{aligned}
\end{equation}
Furthermore, if $\dim(\mathcal{R})\ge\dim(\mathcal{X})$, the equality can be achieved by an appropriate choice of $\sigma$ and $\{M_0,M_1\}$.
\end{lemma}

\begin{figure}[t]
\centering
\includegraphics[width=0.48\textwidth]{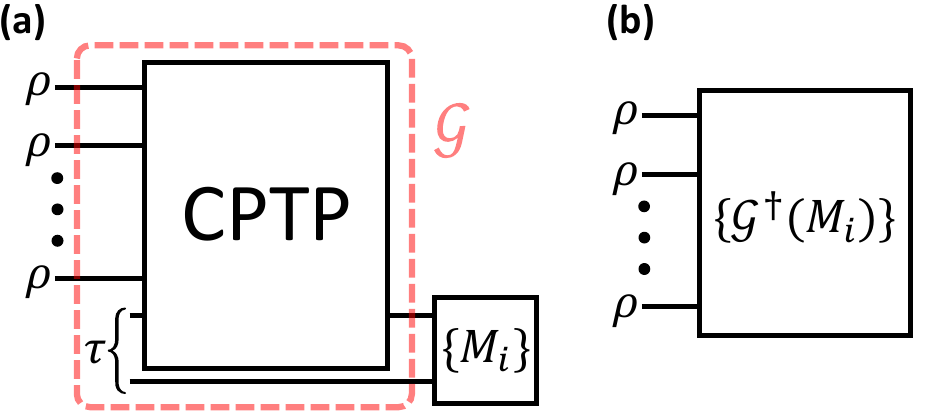}
\caption{(a) Multiple copies of $\rho$ together with a CPTP map induce a channel $\mathcal{Q}^\rho$ on the last input state. Here we demonstrate the optimal method to distinguish two different induced channels $\mathcal{Q}^{\rho_\pm}$, which is composed of first acting $\mathcal{Q}^\rho$ on a subsystem of a larger processed state $\tau$ and measuring it using a two-outcome POVM measurement. (b) The channel discrimination task shown in (a) is equivalent to a state discrimination task, as the CPTP map, the processed state $\tau$, and the POVM measurement $\{M_i\}$ can be considered as a large POVM $\{\mathcal{G}^\dagger(M_i)\}$ acting on $\rho_\pm^{\otimes K}$. Note that $\mathcal{G}^\dagger$ is a completely positive unital map, thus $\{\mathcal{G}^\dagger(M_i)\}$ is a valid POVM.}
\label{fig:Dis_channel_to_state}
\end{figure}

The logic for proving Theorem~\ref{theorem:LowerBound} has been presented in the main context, Sec.~\ref{subsec:optimality}. Hereafter, we will utilize certain quantities and functions that have been defined in Theorem~\ref{theorem:LowerBound}. First, let us choose two states given by
\begin{equation}
\rho_{\pm}:=\abs{A}\pm\lambda A=(1\pm\lambda)A^++(1\mp\lambda)A^-,
\end{equation}
where $A$ is an element of the feasible region $\mathscr{F}$ and $0<\lambda<1$. Note that $\norm{A}_1=1$ and $\Tr(A)=0$ imply $\Tr(\rho_\pm)=1$, and $0<\lambda<1$ implies $\rho_\pm\ge0$, together ensuring that $\rho_\pm$ are valid quantum states.
Denoting $A=\sum_{x}a_x\ketbra{\alpha_x}{\alpha_x}$ as the spectral decomposition of $A$, we have
\begin{equation}
\begin{aligned}
\rho_\pm=&\sum_{x:a_x>0}(1\pm\lambda)a_x\ketbra{\alpha_x}{\alpha_x}\\
&+\sum_{x:a_x<0}(1\mp\lambda)(-a_x)\ketbra{\alpha_x}{\alpha_x}.
\end{aligned}
\end{equation}
Thus, the fidelity between $\rho_+$ and $\rho_-$ is given by
\begin{equation}
\begin{aligned}
F\left(\rho_+,\rho_-\right)=&\left(\sum_x\sqrt{(1+\lambda)(1-\lambda)a_x^2}\right)^2\\
=&\left(1-\lambda^2\right)\left(\sum_x\abs{a_x}\right)^2=1-\lambda^2.
\end{aligned}
\end{equation}
Denote $R=R\left[\mathcal{N}(A)\right]$ for brevity. Take $t_{\lambda}=\frac{\pi}{2\lambda R}$. We can compute the diamond norm distance as
\begin{equation}\label{eq:diamond_distance=2}
\begin{aligned}
&\norm{\left[e^{-i\mathcal{N}(\rho_+)t_{\lambda}}\right]-\left[e^{-i\mathcal{N}(\rho_-)t_{\lambda}}\right]}_{\diamond}\\
=&\max_{\sigma\in D(\mathcal{K}\mathcal{R}),\mathcal{R}}\norm{\left(\left[e^{-i\mathcal{N}(\rho_+)t_{\lambda}}\right]-\left[e^{-i\mathcal{N}(\rho_-)t_{\lambda}}\right]\right)\otimes\mathcal{I}_{\mathcal{R}}(\sigma)}_1\\
=&\underset{\ket{u}}{\max}~\Big\|\left[e^{-i\mathcal{N}(\rho_+)t_{\lambda}}\right]\ketbra{u}{u}-\left[e^{-i\mathcal{N}(\rho_-)t_{\lambda}}\right]\ketbra{u}{u}\Big\|_1\\
=&\underset{\ket{u}}{\max}~2\left(1-\left|\bra{u}e^{i\mathcal{N}(\rho_+)t_{\lambda}}e^{-i\mathcal{N}(\rho_-)t_{\lambda}}\ket{u}\right|^2\right)^{\frac{1}{2}}\\
=&2\left(1-\underset{\ket{u}}{\min}\left|\bra{u}e^{i\pi\frac{\mathcal{N}(A)}{R}}\ket{u}\right|^2\right)^{\frac{1}{2}}=2.
\end{aligned}
\end{equation}
The second equality of Eq.~\eqref{eq:diamond_distance=2} holds because $\left[e^{-i\mathcal{N}(\rho_+)t_{\lambda}}\right]$ and $\left[e^{-i\mathcal{N}(\rho_-)t_{\lambda}}\right]$ are 
unitary channels. The fourth equality of Eq.~\eqref{eq:diamond_distance=2} holds because
\begin{equation}
\left[\mathcal{N}(\rho_+),\mathcal{N}(\rho_-)\right]=4\lambda\left[\mathcal{N}(A^+),\mathcal{N}(A^-)\right]=0,
\end{equation}
where the second equality follows from the fact that $A\in\mathscr{F}$. By setting $\ket{u}=\frac{1}{\sqrt{2}}\left(\ket{L}+\ket{S}\right)$, where $\ket{L}$ and $\ket{S}$ are the eigenstates corresponding to the largest and smallest eigenvalues of $\mathcal{N}(A)$, we can demonstrate that $\left|\bra{u}e^{i\pi\frac{\mathcal{N}(A)}{R}}\ket{u}\right|^2=0$. Hence, the last equality of Eq.~\eqref{eq:diamond_distance=2} holds.

Now for the non-physical map $\mathcal{N}$, desired accuracy $1/3$, and evolution time $t_{\lambda}$, suppose there exists a CPTP map in the general protocol shown in Fig.~\ref{fig:LowerBoundModel}(a) with $K$ copies of the inputting state $\rho$, such that for all $\rho\in D(\mathcal{H})$, the diamond distance between the induced channel $\mathcal{Q}^{\rho}$ and the target unitary channel $\left[e^{-i\mathcal{N}(\rho)t_{\lambda}}\right]$ satisfies $\norm{Q^{\rho}-\left[e^{-i\mathcal{N}(\rho)t_{\lambda}}\right]}_{\diamond}\le1/3$. By selecting $\rho_\pm$ as the inputting states for the general protocol, the diamond distance between the channels induced by them satisfies
\begin{equation}
\begin{aligned}
\norm{\mathcal{Q}^{\rho_+}-\mathcal{Q}^{\rho_-}}_{\diamond}\ge&\norm{\left[e^{-i\mathcal{N}(\rho_+)t_{\lambda}}\right]-\left[e^{-i\mathcal{N}(\rho_-)t_{\lambda}}\right]}_{\diamond}\\
&-\norm{Q^{\rho_+}-\left[e^{-i\mathcal{N}(\rho_+)t_{\lambda}}\right]}_{\diamond}\\
&-\norm{Q^{\rho_-}-\left[e^{-i\mathcal{N}(\rho_-)t_{\lambda}}\right]}_{\diamond}\ge\frac{4}{3}.
\end{aligned}
\end{equation}
According to Lemma~\ref{lemma:Holevo-Helstrom_channel}, if we take an ancillary system $\mathcal{R}$ such that $\dim(\mathcal{R})=\dim(\mathcal{K})$, then there exist a state $\tau\in D(\mathcal{K}\otimes\mathcal{R})$ and a POVM $\{M_0,M_1\}$ on $L(\mathcal{K}\otimes\mathcal{R})$ such that
\begin{equation}\label{eq:LB-SuccessProbLB}
\begin{aligned}
\frac{1}{2}\Tr\big(M_0\left(\mathcal{Q}^{\rho_+} \otimes \mathcal{I}_{\mathcal{R}}\right)(\tau)\big)+\frac{1}{2}\Tr\big(M_1\left(\mathcal{Q}^{\rho_-} \otimes \mathcal{I}_{\mathcal{R}}\right)(\tau)\big) \\
=\frac{1}{2}+\frac{1}{4}\norm{\mathcal{Q}^{\rho_+}- \mathcal{Q}^{\rho_-}}_{\diamond}\ge\frac{5}{6}.
\end{aligned}
\end{equation}
Therefore, it is possible to distinguish between $\mathcal{Q}^{\rho_+}$ and $\mathcal{Q}^{\rho_-}$, thus also $\rho_{\pm}$, with a success probability of at least $5/6$ using just a single copy of $\mathcal{Q}^{\rho_+}$ or $\mathcal{Q}^{\rho_-}$.

On the other hand, as shown in Fig.~\ref{fig:Dis_channel_to_state}(b), the CPTP channel, the processed state $\tau$, and the POVM operators can be unified into a global POVM measurement acting on the multiple input states $\rho$, labeled as $\{\mathcal{G}^\dagger(M_i)\}_{i=0,1}$, where $\mathcal{G}$ represents a new CPTP channel. Thus
\begin{equation}\label{eq:LowerBound_F_change_POVM}
\Tr\big(M_i\left(\mathcal{Q}^{\rho} \otimes \mathcal{I}_{\mathcal{R}}\right)(\tau)\big)=\Tr\left(\mathcal{G}^{\dagger}\left(M_i\right)\rho^{\otimes K}\right)
\end{equation}
holds for $i=0,1$ and all $\rho\in D(\mathcal{H})$. Combining Lemma~\ref{lemma:Holevo-Helstrom}, Eq.~\eqref{eq:LB-SuccessProbLB}, and Eq.~\eqref{eq:LowerBound_F_change_POVM}, we obtain
\begin{equation}\label{eq:LB-SuccessProbUB}
\begin{aligned}
\frac{5}{6}&\le\frac{1}{2}\Tr\left(\mathcal{G}^{\dagger}\left(M_0\right)\rho_+^{\otimes K}\right)+\frac{1}{2}\Tr\left(\mathcal{G}^{\dagger}\left(M_1\right)\rho_-^{\otimes K}\right)\\
&\le\frac{1}{2}+\frac{1}{4}\norm{\rho_+^{\otimes K}-\rho_-^{\otimes K}}_1.
\end{aligned}
\end{equation}
This implies that
\begin{equation}
\begin{aligned}
\frac{4}{3}\le\norm{\rho_+^{\otimes K}-\rho_-^{\otimes K}}_1\le2\sqrt{1-F\left(\rho_+^{\otimes K},\rho_-^{\otimes K}\right)}\\
=2\sqrt{1-F\left(\rho_+,\rho_-\right)^K}=2\sqrt{1-\left(1-\lambda^2\right)^K},
\end{aligned}
\end{equation}
where the second inequality follows from the relation between trace distance and fidelity, and the first equality is based on the property of fidelity \cite{nielsen2010quantum}. We thus have
\begin{equation}\label{eq:log_lower}
\begin{aligned}
K&\ge\log(1.8)\left(\log\frac{1}{1-\lambda^2}\right)^{-1}\\
&\ge\frac{\log(1.8)}{2}\frac{1}{\lambda^2}=\frac{2\log(1.8)}{\pi^2}R^2t_{\lambda}^2,
\end{aligned}
\end{equation}
where the second inequality holds for all $\lambda\in(0,0.8]$. Denote the constant as $C=2\log(1.8)/\pi^2$, by adjusting the value of $\lambda$ within the range of $(0,0.8]$, Eq.~\eqref{eq:log_lower} implies that for all $t\in[\frac{5\pi}{8R},\infty)$, we have
\begin{equation}\label{Eq:LB-1/3}
f_{\mathcal{N}}(1/3,t)\ge CR^2t^2.
\end{equation}

To derive the lower bound for $f_{\mathcal{N}}(\epsilon,t)$, we need to explore the properties of this function. Suppose a global circuit shown in Fig.~\ref{fig:LowerBoundModel}(a) can realize $[e^{-i\mathcal{N}(\rho)t}]$ with $\epsilon$ accuracy, meaning that $\norm{\mathcal{Q}^\rho-[e^{-i\mathcal{N}(\rho)t}]}_\diamond\le\epsilon$. According to the subadditivity of diamond distance, if we concatenate this global circuit for $m$ times, we can realize the evolution of $[e^{-i\mathcal{N}(\rho)mt}]$ with accuracy $m\epsilon$. This implies an important property of $f_{\mathcal{N}}(\epsilon,t)$:
\begin{equation}
f_{\mathcal{N}}(m\epsilon,mt)\le mf_{\mathcal{N}}(\epsilon,t).
\end{equation}

For all $\epsilon\in(0,\frac{1}{6}]$ and $t\in[\frac{15\pi\epsilon}{4R},\infty)$, we take $m=\lceil\frac{1}{6\epsilon}\rceil$, ensuring $m\epsilon\le\frac{1}{3}$ and $mt\ge\frac{5\pi}{8R}$. Thus,
\begin{equation}
\begin{aligned}
f_{\mathcal{N}}(\epsilon,t)\ge& \frac{1}{m}f_{\mathcal{N}}\left(m\epsilon,mt\right)\ge \frac{1}{m}f_{\mathcal{N}}\left(1/3,mt\right)\\
\ge&\frac{1}{m}CR^2(mt)^2=CmR^2t^2\ge\frac{C}{6}\epsilon^{-1}R^2t^2,
\end{aligned}
\end{equation}
where the third inequality is based on Eq.~\eqref{Eq:LB-1/3}. The above equation holds for all $A\in \mathscr{F}$. By taking the maximum value over the feasible region, we define
\begin{equation}
R_*:=\underset{A\in \mathscr{F}}{\max}R\left[\mathcal{N}(A)\right],
\end{equation}
and obtain
\begin{equation}
\begin{aligned}
f_{\mathcal{N}}(\epsilon,t)\ge\frac{C}{6}\epsilon^{-1}R_*^2t^2\ge\Omega\left(\epsilon^{-1}R_*^2t^2\right),
\end{aligned}
\end{equation}
which holds for all $\epsilon\in(0,\frac{1}{6}]$ and all $t\in[\frac{15\pi\epsilon}{4R_*},\infty)$. Thus, we have completed the proof of Theorem~\ref{theorem:LowerBound}.

It can be observed that Theorem~\ref{theorem:LowerBound} holds significance only if the feasible region $\mathscr{F}$ exists. The following proposition presents a sufficient condition for the non-emptiness of $\mathscr{F}$.

\begin{proposition}\label{proposition:non-empty_condition}
For a Hermitian-preserving map $\mathcal{N}\in T(\mathcal{H},\mathcal{K})$, suppose there exists a non-zero Hermitian matrix $P\in L(\mathcal{H})$ such that $\mathcal{N}(P)$ commutes with all matrices in the set $\left\{\mathcal{N}(M):M\in L(\mathcal{H}),M^{\dagger}=M\right\}$, then the feasible region $\mathscr{F}$ defined in Theorem~\ref{theorem:LowerBound} is non-empty.
\end{proposition}
\begin{proof}
If $\rank P=1$, we can select a positive matrix $Q\in L(\mathcal{H})$ whose non-zero eigenstates are orthogonal to the non-zero eigenstate of $P$, one can check that 
\begin{equation}
A=\frac{\abs{P}}{2\norm{P}_1}-\frac{Q}{2\norm{Q}_1}
\end{equation}
is an element of $\mathscr{F}$.

For the case where $\rank P>1$ and $P\ge0$ ($P\le0$), we can decompose $P=P_1+P_2$ using two non-zero matrices $P_1$ and $P_2$, where $P_1,P_2\ge0$ ($P_1,P_2\le0$) and have orthogonal support. One can check that
\begin{equation}
A=\frac{\abs{P_1}}{2\norm{P_1}_1}-\frac{\abs{P_2}}{2\norm{P_2}_1}
\end{equation}
is an element of $\mathscr{F}$.

If both $P^+$ and $P^-$, the positive and negative parts of $P$, are non-zero, then one can check that
\begin{equation}
A=\frac{P^+}{2\norm{P^+}_1}-\frac{P^-}{2\norm{P^-}_1}
\end{equation}
is an element of $\mathscr{F}$.
\end{proof}

Consider the following two easily satisfied conditions:
\begin{enumerate}
\item $\mathcal{N}$ is not injective,
\item $\dim(\mathcal{H})\ge\dim(\mathcal{K})$.
\end{enumerate}
If one of these conditions is satisfied, then the condition required by Proposition~\ref{proposition:non-empty_condition} is also satisfied, implying that $\mathscr{F}$ is non-empty. When $\mathcal{N}$ satisfies condition $1$, we can choose a non-zero Hermitian matrix $P\in L(\mathcal{H})$ such that $\mathcal{N}(P)=0$. As a result, $\mathcal{N}(P)$ commutes with all matrices in $L(\mathcal{K})$. On the other hand, when $\mathcal{N}$ satisfies condition $2$, we only need to consider the case when $\dim(\mathcal{H})=\dim(\mathcal{K})$ and $\mathcal{N}$ is injective, because otherwise $\mathcal{N}$ is not injective and satisfies condition $1$. In this case $\mathcal{N}$ is also surjective, so we can choose a non-zero Hermitian matrix $P\in L(\mathcal{H})$ such that $\mathcal{N}(P)=\mathbb{I}_{\mathcal{K}}$, and thus $\mathcal{N}(P)$ commutes with all matrices in $L(\mathcal{K})$.

Note that when $\mathscr{F}$ is empty, Theorem~\ref{theorem:LowerBound} will yield the trivial lower bound $f_{\mathcal{N}}(\epsilon,t)\ge 0$. However, for Hermitian-preserving maps discussed in this paper, Theorem~\ref{theorem:LowerBound} can always provide meaningful lower bounds. The most challenging condition to fulfill required by the feasible region $\mathscr{F}$ is $\left[\mathcal{N}(A^+),\mathcal{N}(A^-)\right]=0$. This condition is intended to facilitate the computation of $\bra{u}e^{i\mathcal{N}(\rho_+)t_{\lambda}}e^{-i\mathcal{N}(\rho_-)t_{\lambda}}\ket{u}$ in Eq.~\eqref{eq:diamond_distance=2}. Alternatively, more powerful mathematical tools, such as the Baker-Campbell-Hausdorff formula, can be employed to compute this quantity. Utilizing such tools has the potential to eliminate the constraint $\left[\mathcal{N}(A^+),\mathcal{N}(A^-)\right]=0$ from the feasible region $\mathscr{F}$ and derive improved lower bounds for $f_{\mathcal{N}}(\epsilon,t)$.

\subsection{Example: Inverting Local Amplitude Damping Noise Channel}\label{app:LowerBound_AD}

We now calculate a lower bound on the sample complexity for implementing the inverse channel of local amplitude damping noise. Our analysis demonstrates that even with highly global protocols, mitigating the effects of local amplitude damping noise requires an exponential sample complexity. The lower bound presented in this subsection and the upper bound provided by Theorem~\ref{theorem:cost} are asymptotically consistent for all variables: the desired accuracy $\epsilon$, the evolution time $t$, the number of qubits $n$ and the noise intensity $\gamma$.

Let $\mathcal{E}_{\gamma}$ denote the single-qubit amplitude damping noise with a damping rate $0<\gamma<1$. It can be expressed as $\mathcal{E}_{\gamma}(\sigma)=E_0\sigma E_0^{\dagger}+E_1\sigma E_1^{\dagger}$, where
\begin{equation}
\begin{aligned}
E_0&=\ketbra{0}{0}+\sqrt{1-\gamma}\ketbra{1}{1},\\
E_1&=\sqrt{\gamma}\ketbra{0}{1}.
\end{aligned}
\end{equation}
The inverse of $\mathcal{E}_{\gamma}$ acts as
\begin{equation}
\begin{aligned}
\mathcal{E}_{\gamma}^{-1}(\sigma)=&\left(\begin{array}{cc}
1&0\\
0&\frac{1}{\sqrt{1-\gamma}}
\end{array}\right)\sigma\left(\begin{array}{cc}
1&0\\
0&\frac{1}{\sqrt{1-\gamma}}
\end{array}\right)\\
&+\frac{-\gamma}{1-\gamma}\left(\begin{array}{cc}
0&1\\
0&0
\end{array}\right)\sigma\left(\begin{array}{cc}
0&0\\
1&0
\end{array}\right),
\end{aligned}
\end{equation}
Denote $\mathcal{N}_{\gamma,n}:=\left(\mathcal{E}_{\gamma}^{\otimes n}\right)^{-1}=\left(\mathcal{E}_{\gamma}^{-1}\right)^{\otimes n}$ as the inverse of the $n$-qubit local amplitude damping channel we aim to implement.

To compute the lower bound for implementing $\mathcal{N}_{\gamma,n}$, we need to find an element in its feasible region. We observe that
\begin{equation}
\begin{aligned}
\mathcal{N}_{\gamma,n}(\ketbra{0}{0}^{\otimes n})&=\ketbra{0}{0}^{\otimes n},\\
\mathcal{N}_{\gamma,n}(\ketbra{1}{1}^{\otimes n})&=\left(\frac{-\gamma}{1-\gamma}\ketbra{0}{0}+\frac{1}{1-\gamma}\ketbra{1}{1}\right)^{\otimes n},
\end{aligned}
\end{equation}
which implies that $\left[\mathcal{N}_{\gamma,n}(\ketbra{0}{0}^{\otimes n}),\mathcal{N}_{\gamma,n}(\ketbra{1}{1}^{\otimes n})\right]=0$. One can show that $A_n:=\frac{1}{2}\ketbra{1}{1}^{\otimes n}-\frac{1}{2}\ketbra{0}{0}^{\otimes n}$ is an element of the feasible region $\mathscr{F}$. Acting $\mathcal{N}_{\gamma,n}$ on $A_n$, we obtain
\begin{equation}
\begin{aligned}
&\mathcal{N}_{\gamma,n}\left(A_n\right)\\
=&\frac{1}{2}\sum_{s\in\{0,1\}^n}\left(\frac{-\gamma}{1-\gamma}\right)^{n-\abs{s}}\left(\frac{1}{1-\gamma}\right)^{\abs{s}}\ketbra{s}{s}-\frac{1}{2}\ketbra{0}{0}^{\otimes n}\\
=&\frac{1}{2}\left(\frac{1}{1-\gamma}\right)^n\ketbra{1}{1}^{\otimes n}+\left[\frac{1}{2}\left(\frac{-\gamma}{1-\gamma}\right)^n-\frac{1}{2}\right]\ketbra{0}{0}^{\otimes n}\\
&+\cdots,
\end{aligned}
\end{equation}
where $\abs{s}$ denote the number of $1$'s in the string $s$. Since $0<\gamma<1$, we have
\begin{equation}
\begin{aligned}
R_*\ge& R\left[\mathcal{N}_{\gamma,n}\left(A_n\right)\right]\\
\ge&\frac{1}{2}\left(\frac{1}{1-\gamma}\right)^n-\left[\frac{1}{2}\left(\frac{-\gamma}{1-\gamma}\right)^n-\frac{1}{2}\right]\\
\ge&\Omega\left(\left(\frac{1}{1-\gamma}\right)^n\right),
\end{aligned}
\end{equation}
which implies
\begin{equation}
f_{\mathcal{N}_{\gamma,n}}(\epsilon,t)\ge\Omega\left(\epsilon^{-1}\left(\frac{1}{1-\gamma}\right)^{2n}t^2\right).
\end{equation}

Through some straightforward calculations, we find that
\begin{equation}
\norm{\Lambda_{\mathcal{N}_{\gamma,n}}^{T_1}}_{\infty}=\left(\frac{1}{1-\gamma}\right)^n.
\end{equation}
Therefore, according to Theorem~\ref{theorem:cost}, HME is a protocol with sample complexity $\mathcal{O}\left(\epsilon^{-1}\left(\frac{1}{1-\gamma}\right)^{2n}t^2\right)$. This result demonstrates that HME is an asymptotically optimal protocol for realizing $e^{-i\mathcal{N}_{\gamma,n}(\rho)t}$.

\section{Entanglement Detection and
Quantification}

\subsection{Proof of Theorem~\ref{theorem:ED_HME+Reduction}}\label{app:ED_HME+Reduction}

Let $\mathscr{P}_0=\{\psi_A\otimes\psi_B\}$ and $\mathscr{P}_1=\{\psi_{AB}\}$ be the sets of product pure states and global pure states, respectively. An unknown state $\rho$ is selected through the following procedure: first, one randomly chooses $\mathscr{P}_0$ or $\mathscr{P}_1$ with equal probabilities, and then randomly selects a state from the chosen set according to the Haar measure. Our task is to determine which set the unknown state belongs to, using multiple copies of this state.

To accomplish this state discrimination task, we utilize $K+1$ copies of $\rho$ to implement the circuit depicted in Fig.~\ref{fig:ent_det}(b) for $K$ sequential steps. Let $\mathcal{Q}^{\rho}$ be the channel constructed by HME to approximate $\left[\mathrm{C}\text{-}e^{-i\rho^{R_A}\pi}\right]$. By using $K=\mathcal{O}(1)$ copies of $\rho$, we ensure that $\norm{\mathcal{Q}^{\rho}-\left[\mathrm{C}\text{-}e^{-i\rho^{R_A}\pi}\right]}_{\diamond}\le1/12$. Upon measuring the ancilla qubit, if the outcome is $\ket{i}$ $(i=0,1)$, we infer that $\rho$ belongs to the set $\mathscr{P}_i$.

Let us first consider the case when $\rho=\psi_{AB}\in \mathscr{P}_1$. Define
\begin{equation}
\begin{aligned}
&\Sigma:=\left[\mathrm{C}\text{-}e^{-i(-\psi_{AB})\pi}\right]\left(\ketbra{+}{+}\otimes\psi_{AB}\right),\\
&\Sigma^{\prime}:=\mathcal{Q}^{\psi_{AB}}(\ketbra{+}{+}\otimes\psi_{AB}).
\end{aligned}
\end{equation}
As $e^{i\psi_{AB}\pi}=\mathbb{I}-2\psi_{AB}$, $\Sigma$ can be written in a block matrix form as
\begin{equation}
\begin{aligned}
\Sigma=&\left[\mathrm{C}\text{-}(\mathbb{I}-2\psi_{AB})\right](\ketbra{+}{+}\otimes\psi_{AB})\\
=&\frac{1}{2}
\begin{bmatrix}
\psi_{AB} & -\psi_{AB} \\
-\psi_{AB} & \psi_{AB}
\end{bmatrix}.
\end{aligned}
\end{equation}
Upon measuring the ancilla qubit of $\Sigma$ in the Pauli-$X$ basis, we have
\begin{equation}\label{eq:reduction_app_-_exact}
\begin{aligned}
\Tr\big[(\ketbra{-}{-}\otimes\mathbb{I}_d) \Sigma\big]=\Tr(\psi_{AB})=1,
\end{aligned}
\end{equation}
This indicates that we will certainly obtain the measurement result $\ket{1}$.

By Lemma~\ref{lemma:DiamondDistance_HamiltonianEvolution}, we have
\begin{equation}
\begin{aligned}
&\norm{\left[\mathrm{C}\text{-}e^{-i(\psi_{AB})^{R_A}\pi}\right]-\left[\mathrm{C}\text{-}e^{i\psi_{AB}\pi}\right]}_{\diamond}\\
=&\norm{\left[e^{-i\left(\ketbra{1}{1}\otimes(\psi_{AB})^{R_A}\right)\pi}\right]-\left[e^{i\left(\ketbra{1}{1}\otimes\psi_{AB}\right)\pi}\right]}_{\diamond}\\
\le&2\pi\norm{(\psi_{AB})^{R_A}+\psi_{AB}}_{\infty}\\
&\exp\Big(\pi\max\left\{\norm{(\psi_{AB})^{R_A}}_{\infty},\norm{\psi_{AB}}_{\infty}\right\}\Big)\\
\le&2\pi e^{\pi}\norm{\mathbb{I}_{A}\otimes\Tr_{A}(\psi_{AB})-\psi_{AB}+\psi_{AB}}_{\infty}\\
=&2\pi e^{\pi}\norm{\Tr_{A}(\psi_{AB})}_{\infty},
\end{aligned}
\end{equation}
where the second inequality follows from
\begin{equation}
\begin{aligned}
&\norm{(\psi_{AB})^{R_A}}_{\infty}=\norm{\mathbb{I}_{A}\otimes\Tr_{A}(\psi_{AB})-\psi_{AB}}_{\infty}\\
\le&\max\big\{\norm{\mathbb{I}_{A}\otimes\Tr_{A}(\psi_{AB})}_{\infty},\norm{\psi_{AB}}_{\infty}\big\}=1.
\end{aligned}
\end{equation}
The average purity of the reduced density matrix $\Tr_A(\psi_{AB})$ is known to satisfy \cite{Lubkin1993average}
\begin{equation}
\underset{\psi_{AB}\sim\text{Haar}}{\mathbb{E}}\Tr\left[\Tr_A(\psi_{AB})^2\right]=\frac{d_A+d_B}{d_Ad_B+1}=\frac{2\sqrt{d}}{d+1},
\end{equation}
where the last equality holds because we take $d_A=d_B=\sqrt{d}$ for simplicity. Note that
\begin{equation}
\begin{aligned}
\underset{\psi_{AB}\sim\text{Haar}}{\mathbb{E}}\norm{\Tr_A(\psi_{AB})}_{\infty}\le&\underset{\psi_{AB}\sim\text{Haar}}{\mathbb{E}}\sqrt{\Tr\left[\Tr_A(\psi_{AB})^2\right]}\\
\le&\sqrt{\underset{\psi_{AB}\sim\text{Haar}}{\mathbb{E}}\Tr\left[\Tr_A(\psi_{AB})^2\right]}, 
\end{aligned}
\end{equation}
where the second inequality follows from Jensen's inequality. Thus, we have
\begin{equation}
\underset{\psi_{AB}\sim\text{Haar}}{\mathbb{E}}\norm{\Tr_A(\psi_{AB})}_{\infty}\rightarrow0,\text{ as }d\rightarrow\infty.
\end{equation}
When $d$ is sufficiently large, by Markov's inequality, we can guarantee $\norm{\Tr_A(\psi_{AB})}_{\infty}\le(2\pi e^{\pi})^{-1}/12$ with a probability of at least $4/5$. Thus, with a probability of at least $4/5$, we have
\begin{equation}\label{eq:reduction_app_}
\begin{aligned}
\norm{\left[\mathrm{C}\text{-}e^{-i(\psi_{AB})^{R_A}\pi}\right]-\left[\mathrm{C}\text{-}e^{i\psi_{AB}\pi}\right]}_{\diamond}\le1/12.
\end{aligned}
\end{equation}
Combined with
\begin{equation}
\norm{\mathcal{Q}^{\psi_{AB}}-\left[\mathrm{C}\text{-}e^{-i(\psi_{AB})^{R_A}\pi}\right]}_{\diamond}\le1/12,
\end{equation}
we obtain
\begin{equation}\label{eq:reduction_app_-_ChannelDistance_1}
\norm{\mathcal{Q}^{\psi_{AB}}-\left[\mathrm{C}\text{-}e^{i\psi_{AB}\pi}\right]}_{\diamond}\le1/6.
\end{equation}
According to the definition of the diamond distance, we have $\norm{\Sigma^{\prime}-\Sigma}_1\le1/6$. By Lemma~\ref{lemma:expectation_tracenorm}, we know
\begin{equation}\label{eq:reduction_app_-_appro}
\begin{aligned}
\Big|\Tr\big[(\ketbra{-}{-}\otimes\mathbb{I}_d)\Sigma^{\prime}\big]-\Tr\big[(\ketbra{-}{-}\otimes\mathbb{I}_d)\Sigma\big]\Big|\le1/6.
\end{aligned}
\end{equation}
Combined with Eq.~\eqref{eq:reduction_app_-_exact}, we have $\Tr\big[(\ketbra{-}{-}\otimes\mathbb{I}_d)\Sigma^{\prime}\big]\ge5/6$. Now we can conclude that
\begin{equation}
\begin{aligned}
&\Pr\big[\text{obtain~}\ket{1}\big|\rho\in \mathscr{P}_1\big]\\
\ge&\Pr\big[\text{Eq.}~\eqref{eq:reduction_app_-_ChannelDistance_1}\text{~holds~}\big|\rho\in \mathscr{P}_1\big]\\
&\Pr\big[\text{obtain~}\ket{1}\big|\text{Eq.}~\eqref{eq:reduction_app_-_ChannelDistance_1}\text{~holds},\rho\in \mathscr{P}_1\big]\\
\ge&\frac{4}{5}\cdot\frac{5}{6}=\frac{2}{3}.
\end{aligned}
\end{equation}

On the other hand, when $\rho=\psi_A\otimes\psi_B$, we denote 
\begin{equation}
\begin{aligned}
&\Gamma:=\left[\mathrm{C}\text{-}e^{-i(\psi_A\otimes\psi_B)^{R_A}\pi}\right](\ketbra{+}{+}\otimes\psi_A\otimes\psi_B),\\
&\Gamma^{\prime}:=\mathcal{Q}^{\psi_A\otimes\psi_B}(\ketbra{+}{+}\otimes\psi_A\otimes\psi_B).
\end{aligned}
\end{equation}
It can be observed that
\begin{equation}
\begin{aligned}
e^{-i(\psi_A\otimes\psi_B)^{R_A}\pi}=e^{-i\left((\mathbb{I}_A-\psi_A)\otimes\psi_B\right)\pi}=\mathbb{I}_d-2\Pi,
\end{aligned}
\end{equation}
where $\Pi=(\mathbb{I}_A-\psi_A)\otimes\psi_B$ is a projector that satisfies $\Pi(\psi_A\otimes\psi_B)=(\psi_A\otimes\psi_B)\Pi=0$. This leads to the expression of $\Gamma$ as follows:
\begin{equation}
\Gamma=\frac{1}{2}
\begin{bmatrix}
\psi_A\otimes\psi_B & \psi_A\otimes\psi_B \\
\psi_A\otimes\psi_B & \psi_A\otimes\psi_B
\end{bmatrix},
\end{equation}
which satisfies the property
\begin{equation}\label{eq:reduction_app_+_exact}
\Tr\big[(\ketbra{+}{+}\otimes\mathbb{I}_d)\Gamma\big]=\Tr(\psi_A\otimes\psi_B)=1.
\end{equation}
Since we have
\begin{equation}
\norm{\mathcal{Q}^{\psi_A\otimes\psi_B}-\left[\mathrm{C}\text{-}e^{-i(\psi_A\otimes\psi_B)^{R_A}\pi}\right]}_{\diamond}\le1/12,
\end{equation}
it follows that
\begin{equation}\label{eq:reduction_app_+_appro}
\Big|\Tr\big[(\ketbra{+}{+}\otimes\mathbb{I}_d)\Gamma^{\prime}\big]-\Tr\big[(\ketbra{+}{+}\otimes\mathbb{I}_d)\Gamma\big]\Big|\le1/12.
\end{equation}
Combined with Eq.~\eqref{eq:reduction_app_+_exact}, we can deduce that
\begin{equation}
\Pr\big[\text{obtain~}\ket{0}\big|\rho\in \mathscr{P}_0\big]\ge\frac{11}{12}.
\end{equation}

Combining the above derivations, the probability of successfully distinguishing the sets satisfies
\begin{equation}
\begin{aligned}
\Pr\big[\text{success}\big]=&\Pr\big[\rho\in \mathscr{P}_0\big]\Pr\big[\text{obtain~}\ket{0}\big|\rho\in \mathscr{P}_0\big]\\
&+\Pr\big[\rho\in \mathscr{P}_1\big]\Pr\big[\text{obtain~}\ket{1}\big|\rho\in \mathscr{P}_1\big]\\
\ge&\frac{1}{2}\cdot\frac{11}{12}+\frac{1}{2}\cdot\frac{2}{3}>\frac{2}{3}.
\end{aligned}
\end{equation}

\subsection{HME-based Negativity Estimation Algorithm}\label{app:nega_algo}

\begin{algorithm}[H]
\caption{Negativity Estimation}\label{algo:nega}
\begin{algorithmic}[1]
\Require
The measurement accuracy $\epsilon$, the failure probability $\delta$, the system dimension $d=d_Ad_B$.
\Ensure
An estimation $\hat{N}(\rho)$ of $N(\rho)$.
\For{$i= 1 \text{\textbf{ to }} M$} 
\State Randomly sample an integer $l$ according to the probability distribution $\{p(l)=\frac{8}{\pi^2(2l-1)^2}\}_{l=1}^\infty$. 
\If{$l> L=\lceil\frac{3d}{2\pi\epsilon}+\frac{1}{2}\rceil$}
\State Set $\hat{\mathbb{X}}_i=0$.
\Else
\State Run the quantum circuit shown in Fig.~\ref{fig:ent_det}(c) to realize the $\mathrm{C}\text{-}e^{-i\rho^{T_A}t}$ operation for $t=2l-1$ up to an error of $\frac{\pi(2l-1)}{3\left(2+\log(2L-1)\right)}\frac{\epsilon}{d}$ using $K(l)$ copies of $\rho$.
\If{The measurement result of ancilla qubit is $\ket{0}$}
\State Set $\hat{\mathbb{X}}_i=-\frac{\pi}{2}d$.
\Else
\State Set $\hat{\mathbb{X}}_i=+\frac{\pi}{2}d$.
\EndIf
\EndIf
\EndFor
\State Set $\hat{\mathbb{X}}=\operatorname{\mathbf{MedianOfMeans}}(\hat{\mathbb{X}}_1,\cdots,\hat{\mathbb{X}}_M)$.
\State Output $\hat{N}(\rho)=\frac{1}{2}\left(\frac{\pi}{2}d+\hat{\mathbb{X}}-1\right)$.
\end{algorithmic}
\end{algorithm}

\subsection{Proof of Theorem~\ref{theorem:cost_Ent_Neg}}\label{app:nega_cost}

The Fourier series of $\abs{x}$ on $[-\pi,\pi]$ is given by
\begin{equation}
\abs{x}=\frac{\pi}{2}+\sum_{n=1}^{\infty}a_{2n-1}\cos((2n-1)x),
\end{equation}
where $a_{2n-1}=-\frac{4}{\pi(2n-1)^2}$. Then the trace norm of $\rho^{T_A}$ is given by
\begin{equation}\label{eq:Neg_series}
\begin{aligned}
&\norm{\rho^{T_A}}_1=\Tr(\abs{\rho^{T_A}})\\
=&\Tr\left[\frac{\pi}{2}\mathbb{I}_d+\sum_{n=1}^{\infty}a_{2n-1}\cos((2n-1)\rho^{T_A})\right]\\
=&\frac{\pi}{2}d+\sum_{n=1}^{\infty}a_{2n-1}\Tr\left[\cos((2n-1)\rho^{T_A})\right].
\end{aligned}
\end{equation}

Let $\tau^{(2n-1)}$ be the state of the ancilla qubit before the Pauli-$X$ basis measurement when the evolution time $t$ is set to $t=2n-1$ in the right circuit shown in Fig.~\ref{fig:ent_det}(c). Correspondingly, let $\widetilde{\tau}^{(2n-1)}$ be the state of the ancilla qubit before the Pauli-$X$ basis measurement in the left circuit shown in Fig.~\ref{fig:ent_det}(c). Note that
\begin{equation}
\Tr(\tau^{(2n-1)}X)=\frac{1}{d}\Tr\left[\cos((2n-1)\rho^{T_A})\right].
\end{equation}

Now, we truncate the series in Eq.~\eqref{eq:Neg_series} and consider only the first $L$
terms to avoid divergence. We estimate the absolute value of the sum of the discarded terms as follows:
\begin{equation}
\begin{aligned}
&\abs{\sum_{n=L+1}^{\infty}a_{2n-1}\Tr\left[\cos((2n-1)\rho^{T_A})\right]}\\
\le&\sum_{n=L+1}^{\infty}\abs{a_{2n-1}}\Big|\Tr\left[\cos((2n-1)\rho^{T_A})\right]\Big|\\
\le& \frac{4d}{\pi}\sum_{n=L+1}^{\infty}\frac{1}{(2n-1)^2}\\
<& \frac{4d}{\pi}\int_{L}^{\infty}\frac{1}{(2x-1)^2}\,dx\\
=&\frac{2d}{(2L-1)\pi}.
\end{aligned}
\end{equation}

To ensure that $\frac{2d}{(2L-1)\pi}\le\epsilon_1$, we choose $L=\lceil\frac{1}{\pi}d\epsilon_1^{-1}+\frac{1}{2}\rceil$. Here, $\epsilon_1$ represents the truncation error threshold, which we will specify later.

Our goal is to estimate the sum
\begin{equation}\label{eq:negativity_first_N_terms}
\begin{aligned}
&\sum_{n=1}^{L}a_{2n-1}\Tr\left[\cos((2n-1)\rho^{T_A})\right]
\end{aligned}
\end{equation}
up to an error of $\epsilon_2$. To achieve this, we introduce a new random variable $\mathbb{X}$, with the following probabilities for its values:
\begin{itemize}
\item $-\frac{\pi}{2}d$ with probability
\begin{equation}
\sum_{n=1}^L\frac{8}{\pi^2(2n-1)^2}\bra{+}\tau^{(2n-1)}\ket{+},
\end{equation}
\item $\frac{\pi}{2}d$ with probability
\begin{equation}
\sum_{n=1}^L\frac{8}{\pi^2(2n-1)^2}\bra{-}\tau^{(2n-1)}\ket{-},
\end{equation}
\item $0$ with probability
\begin{equation}
1-\sum_{n=L+1}^{\infty}\frac{8}{\pi^2(2n-1)^2}.
\end{equation}
\end{itemize}
The random variable $\mathbb{X}$ serves as a good estimator for Eq.~\eqref{eq:negativity_first_N_terms} because its expected value is
\begin{equation}
\begin{aligned}
\mathbb{E}\left[\mathbb{X}\right]=&-\frac{\pi}{2}d\left(\sum_{n=1}^L\frac{8}{\pi^2(2n-1)^2}\bra{+}\tau^{(2n-1)}\ket{+}\right.\\
&\left.-\sum_{n=1}^L\frac{8}{\pi^2(2n-1)^2}\bra{-}\tau^{(2n-1)}\ket{-}\right)\\
=&-\frac{\pi}{2}d\sum_{n=1}^L\frac{8}{\pi^2(2n-1)^2}\Tr\left(\tau^{(2n-1)}X\right)\\
=&\sum_{n=1}^{L}a_{2n-1}\Tr\left[\cos((2n-1)\rho^{T_A})\right].
\end{aligned}
\end{equation}

Correspondingly, we introduce a random variable $\widetilde{\mathbb{X}}$, with the following probabilities for its values:
\begin{itemize}
\item $-\frac{\pi}{2}d$ with probability
\begin{equation}
\sum_{n=1}^L\frac{8}{\pi^2(2n-1)^2}\bra{+}\widetilde{\tau}^{(2n-1)}\ket{+},
\end{equation}
\item $\frac{\pi}{2}d$ with probability
\begin{equation}
\sum_{n=1}^L\frac{8}{\pi^2(2n-1)^2}\bra{-}\widetilde{\tau}^{(2n-1)}\ket{-},
\end{equation}
\item $0$ with probability
\begin{equation}
1-\sum_{n=L+1}^{\infty}\frac{8}{\pi^2(2n-1)^2}.
\end{equation}
\end{itemize}
The bias between the random variables $\mathbb{X}$ and $\widetilde{\mathbb{X}}$ is given by
\begin{equation}
\begin{aligned}
&\abs{\mathbb{E}[\widetilde{\mathbb{X}}]-\mathbb{E}[\mathbb{X}]}\\
\le&\sum_{n=1}^L\frac{4d}{\pi(2n-1)^2}\abs{\Tr\left(\widetilde{\tau}^{(2n-1)}X\right)-\Tr\left(\tau^{(2n-1)}X\right)}.
\end{aligned}
\end{equation}

To sample the random variable $\widetilde{\mathbb{X}}$, we first sample a random integer $l$ according to the probability distribution
\begin{equation}
\left\{p(l)=\frac{8}{\pi^2(2l-1)^2}\right\}_{l=1}^\infty.
\end{equation}
If $l\ge L+1$, we set $\tilde{\mathbb{X}}=0$. Otherwise, we use HME to generate the state $\widetilde{\tau}^{(2l-1)}$, and then perform a measurement in the Pauli-$X$ basis. Set $\widetilde{\mathbb{X}}=-\frac{\pi}{2}d$ if the measurement outcome is $\ket{+}$, set $\widetilde{\mathbb{X}}=+\frac{\pi}{2}d$ if the measurement outcome is $\ket{-}$.

For $n=1,\cdots,L$, we set
\begin{equation}
e_{2n-1}:=\frac{\pi(2n-1)}{2\left(2+\log(2L-1)\right)}\frac{\epsilon_2}{d}.
\end{equation}
We use HME to approximate the controlled-$e^{-i(2n-1)\rho^{T_A}}$ operation with a precision of $e_{2n-1}$. This precision ensures that $\norm{\widetilde{\tau}^{(2n-1)}-\tau^{(2n-1)}}_1\le e_{2n-1}$. By applying Lemma~\ref{lemma:expectation_tracenorm}, we bound the absolute difference between $\Tr(\widetilde{\tau}^{(2n-1)}X)$ and $\Tr(\tau^{(2n-1)}X)$ as:
\begin{equation}
\abs{\Tr(\widetilde{\tau}^{(2n-1)}X)-\Tr(\tau^{(2n-1)}X)}\le e_{2n-1}.
\end{equation}
Therefore, the bias between $\widetilde{\mathbb{X}}$ and $\mathbb{X}$ is bounded by
\begin{equation}
\begin{aligned}
&\abs{\mathbb{E}[\widetilde{\mathbb{X}}]-\mathbb{E}[\mathbb{X}]}\le\sum_{n=1}^L\frac{4d}{\pi(2n-1)^2}e_{2n-1}\\
=&\frac{2\epsilon_2}{2+\log(2L-1)}\sum_{n=1}^L\frac{1}{2n-1}\\
<&\frac{2\epsilon_2}{2+\log(2L-1)}\left(1+\int_{1}^{L}\frac{1}{2x-1}dx\right)=\epsilon_2.
\end{aligned}
\end{equation}

We sample $\widetilde{\mathbb{X}}$ for $M$ times and obtain independent and identically distributed random variables $\widetilde{\mathbb{X}}_1,\cdots,\widetilde{\mathbb{X}}_M$. Now, let's estimate the required value of $M$ to ensure the desired accuracy and success probability. To do so, we need to estimate the variance of $\widetilde{\mathbb{X}}$, denoted as $\operatorname{Var}[\widetilde{\mathbb{X}}]$. Note that
\begin{equation}
\begin{aligned}
&\abs{\mathbb{E}[\widetilde{\mathbb{X}}]-\left(-\frac{\pi}{2}d\right)}\\
\le&\abs{\mathbb{E}[\widetilde{\mathbb{X}}]-\mathbb{E}[\mathbb{X}]}+\abs{\mathbb{E}[\mathbb{X}]-\left(-\frac{\pi}{2}d\right)}\\
\le&\epsilon_2+\abs{\norm{\rho^{T_A}}_1-\sum_{n=L+1}^{\infty}a_{2n-1}\Tr\left[\cos((2n-1)\rho^{T_A})\right]}\\
<&\epsilon_2+\norm{\rho^{T_A}}_1+\epsilon_1.
\end{aligned}
\end{equation}
Thus, we have
\begin{equation}
\begin{aligned}
&\operatorname{Var}[\widetilde{\mathbb{X}}]=\mathbb{E}[\widetilde{\mathbb{X}}^2]-\mathbb{E}[\widetilde{\mathbb{X}}]^2\\
\le&\abs{\mathbb{E}[\widetilde{\mathbb{X}}^2]-\frac{\pi^2d^2}{4}}+\abs{\mathbb{E}[\widetilde{\mathbb{X}}]^2-\frac{\pi^2d^2}{4}}\\
=&\frac{\pi^2d^2}{4}\sum_{n=L+1}^{\infty}\frac{8}{\pi^2(2n-1)^2}+\abs{\mathbb{E}[\widetilde{\mathbb{X}}]-\frac{\pi d}{2}}\abs{\mathbb{E}[\widetilde{\mathbb{X}}]+\frac{\pi d}{2}}\\
<&\frac{\pi}{2}\epsilon_1d+\pi d\left(\norm{\rho^{T_A}}_1+\epsilon_1+\epsilon_2\right)=\mathcal{O}\left(d\norm{\rho^{T_A}}_1\right),
\end{aligned}
\end{equation}
where the last inequality is given by 
\begin{equation}
\sum_{n=L+1}^{\infty}\frac{1}{(2n-1)^2}<\frac{\pi\epsilon_1}{4d},
\end{equation}
which is guaranteed by our previous choice of $L$.
Since $\operatorname{Var}[\widetilde{\mathbb{X}}]$ is upper bounded by $\mathcal{O}\left(d\norm{\rho^{T_A}}_1\right)$, we can use the following lemma to bound the number of samples required.

\begin{lemma}\label{lemma:median_of_means}
Suppose we aim to estimate the expectation value $\mathbb{E}[R]$ of a random variable $R$ with finite variance, using independently drawn samples $R_1, \cdots, R_M$ from its distribution. Then by employing the median of means estimation method and taking $M=\mathcal{O}\left(\log \left(\delta^{-1}\right)\operatorname{Var}[R]\epsilon^{-2}\right)$, we can guarantee
\begin{equation}
\Big|\operatorname{\mathbf{MedianOfMeans}}(R_1,\cdots,R_M)-\mathbb{E}[R]\Big|\le\epsilon
\end{equation}
with a probability of at least $1-\delta$.
\end{lemma}

According to Lemma~\ref{lemma:median_of_means}, to estimate $\mathbb{E}[\widetilde{\mathbb{X}}]$ with accuracy $\epsilon_3$ and a success probability of at least $1-\delta$, it suffices to sample $\mathbb{X}$ for $M=\mathcal{O}\left(\log(\delta^{-1})d\norm{\rho^{T_A}}_1\epsilon_3^{-2}\right)$ times.

By Proposition~\ref{prop:PT_cost}, the expected number of $\rho$ required for a single sampling of $\widetilde{\mathbb{X}}$ is
\begin{equation}
\begin{aligned}
&\underset{l\sim p(\cdot)}{\mathbb{E}}\left[\mathcal{O}\left(e_{2l-1}^{-1}d_A(2l-1)^2\right)\right]\\
=&\mathcal{O}\left(\sum_{l=1}^{L}\frac{8}{\pi^2(2l-1)^2}\left(\frac{2\left(2+\log(2L-1)\right)}{\pi(2l-1)}\frac{d}{\epsilon_2}\right)d_A(2l-1)^2\right)\\
=&\mathcal{O}\left(\frac{16}{\pi^3}\frac{dd_A}{\epsilon_2}\sum_{l=1}^{L}\frac{\left(2+\log(2L-1)\right)}{2l-1}\right)\\
=&\mathcal{O}\left(\frac{8}{\pi^3}\frac{dd_A}{\epsilon_2}\left(2+\log(2L-1)\right)^2\right)\\
=&\mathcal{O}\left(\epsilon_2^{-1}dd_A\log(L)^2\right).
\end{aligned}
\end{equation}

We take $\epsilon_1=\epsilon_2=\epsilon_3=\frac{2\epsilon}{3}$. Now, with a success probability of at least $1-\delta$, we have
\begin{equation}
\begin{aligned}
&\abs{\left(\frac{\pi}{2}d+\operatorname{\mathbf{MedianOfMeans}}(\hat{\mathbb{X}}_1,\cdots,\hat{\mathbb{X}}_M)\right)-\norm{\rho^{T_A}}_1}\\
\le&\abs{\operatorname{\mathbf{MedianOfMeans}}(\hat{\mathbb{X}}_1,\cdots,\hat{\mathbb{X}}_M)-\mathbb{E}[\widetilde{\mathbb{X}}]}\\
&+\abs{\mathbb{E}[\widetilde{\mathbb{X}}]-\mathbb{E}[\mathbb{X}]}+\abs{\left(\frac{\pi}{2}d+\mathbb{E}[\mathbb{X}]\right)-\norm{\rho^{T_A}}_1}\\
<&\epsilon_3+\epsilon_2+\epsilon_1=2\epsilon.
\end{aligned}
\end{equation}
Thus, we can estimate $\norm{\rho^{T_A}}_1$ to an accuracy of $2\epsilon$, which allows us to determine $N(\rho)=\frac{\norm{\rho^{T_A}}_1-1}{2}$ to an accuracy of $\epsilon$.

Putting everything together, the expectation value of the total number of copies of $\rho$ we need is
\begin{equation}
\begin{aligned}
&M\underset{l\sim p(\cdot)}{\mathbb{E}}\left[\mathcal{O}\left(e_{2l-1}^{-1}d_A(2l-1)^2\right)\right]\\
=&\mathcal{O}\left(\log(\delta^{-1})d\norm{\rho^{T_A}}_1\epsilon_3^{-2}\right)\mathcal{O}\left(\epsilon_2^{-1}dd_A\log(L)^2\right)\\
=&\mathcal{O}\left(\log(\delta^{-1})\epsilon^{-3}d^2d_A\norm{\rho^{T_A}}_1\log(\epsilon^{-1}d)^2\right)\\
=&\widetilde{\mathcal{O}}\left(\log(\delta^{-1})\epsilon^{-3}d^2d_A\norm{\rho^{T_A}}_1\right),
\end{aligned}
\end{equation}
here the notation $\widetilde{\mathcal{O}}$, suppresses logarithmic factors of $d$ and $\epsilon$. Thus, we have completed the proof of Theorem~\ref{theorem:cost_Ent_Neg}.

Now, we discuss several properties of entanglement negativity, which facilitate the discussions in Sec.~\ref{subsec:ent_det}. We first show that $\norm{\rho^{T_A}}_1\le d_A$. Let $\ket{\psi}$ be a $d=d_A\times d_B$ dimensional pure state with Schmidt coefficients $s_1\ge\cdots\ge s_{d_A}\ge 0$. According to Lemma 1 in \cite{Johnston2018inverse}, we have
\begin{equation}
\begin{aligned}
\norm{\ketbra{\psi}{\psi}^{T_A}}_1=(s_1+\cdots+s_{d_A})^2.
\end{aligned}
\end{equation}
Since $s_1^2+\cdots+s_{d_A}^2=1$, it follows that $\norm{\ketbra{\psi}{\psi}^{T_A}}_1\le d_A$. Let $\rho=\sum_{i=1}^d\lambda_i\ketbra{\psi_i}{\psi_i}$ be the spectral decomposition of $\rho$, then we have
\begin{equation}
\begin{aligned}
\norm{\rho^{T_A}}_1&=\norm{\sum_{i=1}^d\lambda_i\ketbra{\psi_i}{\psi_i}^{T_A}}_1\\
&\le\sum_{i=1}^d\lambda_i\norm{\ketbra{\psi_i}{\psi_i}^{T_A}}_1\le d_A.
\end{aligned}
\end{equation}

We then show that for all $d$-diamonsional states $\rho_1$ and $\rho_2$, the following inequality holds:
\begin{equation}
\abs{N(\rho_1)-N(\rho_2)}\le\frac{\sqrt{d}}{2}\norm{\rho_1-\rho_2}_1.
\end{equation}
Let $\rho_1-\rho_2=\sum_{i=1}^d\mu_i\ketbra{\phi_i}{\phi_i}$ be the spectral decomposition of $\rho_1-\rho_2$, then
\begin{equation}
\begin{aligned}
&\abs{N(\rho_1)-N(\rho_2)}=\frac{1}{2}\abs{\norm{\rho_1^{T_A}}_1-\norm{\rho_2^{T_A}}_1}\\
\le&\frac{1}{2}\norm{\rho_1^{T_A}-\rho_2^{T_A}}_1=\frac{1}{2}\norm{(\rho_1-\rho_2)^{T_A}}_1\\
\le&\frac{1}{2}\sum_{i=1}^d\abs{\mu_i}\norm{\ketbra{\phi_i}{\phi_i}^{T_A}}_1\le\frac{d_A}{2}\norm{\rho_1-\rho_2}_1\\
\le&\frac{\sqrt{d}}{2}\norm{\rho_1-\rho_2}_1.
\end{aligned}
\end{equation}

\subsection{Proof of Proposition~\ref{prop:nega_LB}}\label{app:nega_LB}

Consider the case when $d_A=d_B=\sqrt{d}$. We define $\rho(x)$ as follows:
\begin{equation}
\begin{aligned}
\rho(x):=x \frac{1}{\sqrt{d}}\Phi^++(1-x)\frac{1}{d}\mathbb{I}_d,
\end{aligned}
\end{equation}
where $\Phi^+$ is the unnormalized maximally entangled state.
This state provides an example that the difference in negativity will be exponentially larger than the difference in trace distance. For $x>\frac{1}{1+\sqrt{d}}$, we have
\begin{equation}
N[\rho(x)]=x\frac{\sqrt{d}}{2}-\frac{1}{2}+\frac{1-x}{2\sqrt{d}}.
\end{equation} 
Thus, when $x_1>x_2>\frac{1}{1+\sqrt{d}}$, the difference in negativity will be
\begin{equation}
N[\rho(x_1)]-N[\rho(x_2)]=\left(\frac{\sqrt{d}}{2}-\frac{1}{2\sqrt{d}}\right)(x_1-x_2).
\end{equation}
While at the same time, 
\begin{equation}
\norm{\rho(x_1)-\rho(x_2)}_1=2\left(1-\frac{1}{d}\right)(x_1-x_2).
\end{equation}

Now we use $\rho(x)$ to prove Proposition~\ref{prop:nega_LB}. For $0\le x_1,x_2\le1$, the fidelity between $\rho(x_1)$ and $\rho(x_2)$ can be lower bounded as
\begin{equation}\label{eq:nega_LB_fidelity}
\begin{aligned}
&F\big(\rho(x_1),\rho(x_2)\big)\\
=&\left[\sqrt{\left(x_1+(1-x_1)\frac{1}{d}\right)\left(x_2+(1-x_2)\frac{1}{d}\right)}\right.\\
&+\left.(d-1)\sqrt{(1-x_1)(1-x_2)\frac{1}{d^2}}\right]^2\\
\ge&\left[\sqrt{x_1x_2}+\sqrt{(1-x_1)(1-x_2)}\right]^2,
\end{aligned}
\end{equation}
where the inequality holds for all $0\le x_1,x_2\le1$ and all dimension $d$. Taking $\lambda=6\epsilon/\sqrt{d}$, 
\begin{equation}
\begin{aligned}
N[\rho(1/2+\lambda)]-N[\rho(1/2)]=\lambda\left(\frac{\sqrt{d}}{2}-\frac{1}{2\sqrt{d}}\right)>2.5\epsilon
\end{aligned}
\end{equation}
holds when $d>6$. Thus, if we can estimate $N(\rho)$ to an accuracy of $\epsilon$ with a success probability of at least $2/3$ using $K$ copies of the unknown state $\rho$, then we can distinguish between $\rho(1/2+\lambda)$ and $\rho(1/2)$ with a probability of at least $2/3$ using $K$ copies of $\rho$. By Eq.~\eqref{eq:nega_LB_fidelity}, we have
\begin{equation}
\begin{aligned}
F\big(\rho(1/2+\lambda),\rho(1/2)\big)&\ge\left[\sqrt{\frac{1}{4}+\frac{\lambda}{2}}+\sqrt{\frac{1}{4}-\frac{\lambda}{2}}\right]^2\\
&\ge1-2\lambda^2,
\end{aligned}
\end{equation}
where the second inequality holds for all $\lambda\in(0,1/2]$.

By Lemma~\ref{lemma:Holevo-Helstrom}, in order to distinguish between $\rho(1/2+\lambda)$ and $\rho(1/2)$ with a probability of at least $2/3$ using $K$ copies of the unknown state, $K$ must be large enough to make sure
\begin{equation}\label{eq:nega_LB_HH1}
\begin{aligned}
\norm{\rho(1/2+\lambda)^{\otimes K}-\rho(1/2)^{\otimes K}}_1\ge\frac{2}{3}.
\end{aligned}
\end{equation}
Note that
\begin{equation}\label{eq:nega_LB_HH2}
\begin{aligned}
&\norm{\rho(1/2+\lambda)^{\otimes K}-\rho(1/2)^{\otimes K}}_1\\
\le&2\sqrt{1-F\big(\rho(1/2+\lambda)^{\otimes K},\rho(1/2)^{\otimes K}\big)}\\
=&2\sqrt{1-F\big(\rho(1/2+\lambda),\rho(1/2)\big)^K}\\
\le&2\sqrt{1-(1-2\lambda^2)^K}.
\end{aligned}
\end{equation}

Combine Eq.~\eqref{eq:nega_LB_HH1} and \eqref{eq:nega_LB_HH2}, $K$ must satisfies
\begin{equation}
\begin{aligned}
K&\ge\log\left(\frac{9}{8}\right)\log\left(\frac{1}{1-2\lambda^2}\right)^{-1}\\
&\ge\log\left(\frac{9}{8}\right)\frac{1}{4\lambda^2}\\
&\ge\Omega\left(\lambda^{-2}\right)=\Omega\left(\epsilon^{-2}d\right),
\end{aligned}
\end{equation}
where the second inequality holds when $\lambda\in(0,1/2]$.

\section{Quantum Noiseless State
Recovery}

\subsection{Proof of Theorem~\ref{theorem:HME-QEM}}\label{app:HME_QNSR_thm}

We use $\mathcal{Q}$ to represent the channel constructed by HME, which
approximates the channel $\left[\mathrm{C}\text{-}e^{-i\psi\pi}\right]$. We select the number of sequential operations $K = \mathcal{O}(\tilde{\epsilon}^{-1}\norm{H_{\mathcal{E}^{-1}}}_{\infty}^2)$ to guarantee $\norm{\mathcal{Q}-\left[\mathrm{C}\text{-}e^{-i\psi\pi}\right]}_{\diamond}\le\tilde{\epsilon}$, where $\tilde{\epsilon}$ represents an error threshold that will be further defined in our analysis. Let
\begin{equation}
\begin{aligned}
&\Sigma:=\left[\mathrm{C}\text{-}e^{-i\psi\pi}\right]\left(\ketbra{+}{+}\otimes\sigma\right),\\
&\Sigma^{\prime}:=\mathcal{Q}(\ketbra{+}{+}\otimes\sigma),
\end{aligned}
\end{equation}
by definition of diamond distance, we have $\norm{\Sigma^{\prime}-\Sigma}_1<\tilde{\epsilon}$. After measuring the ancilla qubit of $\Sigma^{\prime}$ in the Pauli-$X$
basis and obtaining the outcome $\ket{-}$, the state evolves into
\begin{equation}
\Tr_c\left(\frac{(\ketbra{-}{-}\otimes\mathbb{I}_d)\Sigma^{\prime}(\ketbra{-}{-}\otimes\mathbb{I}_d)}{\Tr\big[(\ketbra{-}{-}\otimes\mathbb{I}_d)\Sigma^{\prime}\big]}\right),
\end{equation}
which closely approximates $\psi$. Note that
\begin{equation}
\begin{aligned}
F^{\prime}:=\Tr\big[(\ketbra{-}{-}\otimes\mathbb{I}_d)\Sigma^{\prime}\big]
\end{aligned}
\end{equation}
is the probability of obtaining the measurement result $\ket{-}$. By Lemma~\ref{lemma:expectation_tracenorm}, we can observe that
\begin{equation}
\begin{aligned}
&\big|F^{\prime}-F\big|\\
=&\Big|\Tr\big[(\ketbra{-}{-}\otimes\mathbb{I}_d)\Sigma^{\prime}\big]-\Tr\big[(\ketbra{-}{-}\otimes\mathbb{I}_d)\Sigma\big]\Big|\\
\le&\norm{\Sigma^{\prime}-\Sigma}_1<\tilde{\epsilon},
\end{aligned}
\end{equation}
where
\begin{equation}
\begin{aligned}
F:=\Tr\big[(\ketbra{-}{-}\otimes\mathbb{I}_d)\Sigma\big]=\bra{\psi}\sigma\ket{\psi}
\end{aligned}
\end{equation}
is the fidelity between the initial state $\sigma$ and the target noiseless state $\psi$. We have
\begin{widetext}
\begin{equation}
\begin{aligned}
&\norm{\Tr_c\left(\frac{(\ketbra{-}{-}\otimes\mathbb{I}_d)\Sigma^{\prime}(\ketbra{-}{-}\otimes\mathbb{I}_d)}{\Tr\big[(\ketbra{-}{-}\otimes\mathbb{I}_d)\Sigma^{\prime}\big]}\right)-\ketbra{\psi}{\psi}}_1\\
=&\norm{\Tr_c\left(\frac{(\ketbra{-}{-}\otimes\mathbb{I}_d)\Sigma^{\prime}(\ketbra{-}{-}\otimes\mathbb{I}_d)}{\Tr\big[(\ketbra{-}{-}\otimes\mathbb{I}_d)\Sigma^{\prime}\big]}\right)-\Tr_c\left(\frac{(\ketbra{-}{-}\otimes\mathbb{I}_d)\Sigma(\ketbra{-}{-}\otimes\mathbb{I}_d)}{\Tr\big[(\ketbra{-}{-}\otimes\mathbb{I}_d)\Sigma\big]}\right)}_1\\
\le&\norm{\frac{1}{F^{\prime}}\Tr_c\Big((\ketbra{-}{-}\otimes\mathbb{I}_d)\Sigma^{\prime}(\ketbra{-}{-}\otimes\mathbb{I}_d)\Big)-\frac{1}{F^{\prime}}\Tr_c\Big((\ketbra{-}{-}\otimes\mathbb{I}_d)\Sigma(\ketbra{-}{-}\otimes\mathbb{I}_d)\Big)}_1\\
&+\norm{\frac{1}{F^{\prime}}\Tr_c\Big((\ketbra{-}{-}\otimes\mathbb{I}_d)\Sigma(\ketbra{-}{-}\otimes\mathbb{I}_d)\Big)-\frac{1}{F}\Tr_c\Big((\ketbra{-}{-}\otimes\mathbb{I}_d)\Sigma(\ketbra{-}{-}\otimes\mathbb{I}_d)\Big)}_1\\
\le&\frac{1}{F^{\prime}}\Big\|(\ketbra{-}{-}\otimes\mathbb{I}_d)(\Sigma^{\prime}-\Sigma)(\ketbra{-}{-}\otimes\mathbb{I}_d)\Big\|_1+\left|\frac{1}{F^{\prime}}-\frac{1}{F}\right|\norm{\Tr_c\Big((\ketbra{-}{-}\otimes\mathbb{I}_d)\Sigma(\ketbra{-}{-}\otimes\mathbb{I}_d)\Big)}_1\\
\le&\frac{\tilde{\epsilon}}{F^{\prime}}+\frac{\tilde{\epsilon}}{F^{\prime}F}F=\frac{2\tilde{\epsilon}}{F^{\prime}}\le\frac{2\tilde{\epsilon}}{F-\tilde{\epsilon}}.
\end{aligned}
\end{equation}
\end{widetext}
In order to guarantee
\begin{equation}
\begin{aligned}
\norm{\Tr_c\left(\frac{(\ketbra{-}{-}\otimes\mathbb{I}_d)\Sigma^{\prime}(\ketbra{-}{-}\otimes\mathbb{I}_d)}{\Tr\big[(\ketbra{-}{-}\otimes\mathbb{I}_d)\Sigma^{\prime}\big]}\right)-\ketbra{\psi}{\psi}}_1<\epsilon
\end{aligned}
\end{equation}
for some given error threshold $\epsilon$, we can take $\tilde{\epsilon}=F\epsilon/3$. Thus, $\mathcal{O}(\epsilon^{-1}F^{-1}\norm{H_{\mathcal{E}^{-1}}}_{\infty}^2)$ copies of $\rho$ are sufficient for us to get $\psi$ up to an error of $\epsilon$ when the post-selection succeeds.

We repeatedly run the circuit shown in Fig.~\ref{fig:QNSR} for $m$ times to increase the success probability to $1-(1-F^{\prime})^m$. To ensure $1-(1-F^{\prime})^m\ge 1-\delta$, it suffices to take
\begin{equation}
\begin{aligned}
m=&\left\lceil\log(\delta^{-1})\log(\frac{1}{1-F^{\prime}})^{-1}\right\rceil\\
=&\mathcal{O}\left(\log(\delta^{-1})F^{\prime-1}\right)=\mathcal{O}\left(\log(\delta^{-1})F^{-1}\right).
\end{aligned}
\end{equation}
Thus, the total number of copies of $\rho$ we need to achieve a success probability of at least $1-\delta$ is
\begin{equation}
\mathcal{O}\left(\log(\delta^{-1})\epsilon^{-1}F^{-2}\norm{H_{\mathcal{E}^{-1}}}_{\infty}^2\right).
\end{equation}

\subsection{Proof of Corollary~\ref{corollary:hardness_of_indnependnet_noise}}\label{app:hardness_of_indnependnet_noise}

In this section, let $M_{\mathbf{r}}:=\frac{1}{2}(\mathbb{I}_2+\mathbf{r}\cdot\boldsymbol{\sigma})$ be the Bloch representation of a two-dimensional Hermitian matrix, where $\mathbf{r}=(x,y,z)\in\mathbb{R}^3$ is the Bloch vector of $M_{\mathbf{r}}$, and $\boldsymbol{\sigma}=(\sigma_1,\sigma_2,\sigma_3)=(X,Y,Z)$. Additionally, $\norm{\mathbf{r}}=\sqrt{x^2+y^2+z^2}$ represents the norm of $\mathbf{r}$, $D^3=\{(x,y,z):x^2+y^2+z^2\le1\}$ represents the closed unit ball in $\mathbb{R}^3$, and $S^2=\{(x,y,z):x^2+y^2+z^2=1\}$ represents the unit sphere.

The action of a trace-preserving map in the Bloch representation can always be expressed as an affine map \cite{nielsen2010quantum}, which maps one ellipsoid to another ellipsoid. Note that affine maps preserve antipodal points of ellipsoids. We use the notation $\mathcal{E}$ to refer to both the CPTP noise channel and the corresponding affine map on $\mathbb{R}^3$. If $\mathcal{E}$ is not a unitary channel, it maps some pure state on $S^2$ to the interior of $D^3$. Consequently, $\mathcal{E}(D^3)\subsetneq D^3$, which implies that $D^3\subsetneq \mathcal{E}^{-1}(D^3)$.

In the following, we will discuss two cases separately: when $\mathcal{E}$ is a unital map, and when it is not a unital map.

If $\mathcal{E}$ is a unital map, then $(0,0,0)$ is the fixed point of $\mathcal{E}$. Since $D^3\subsetneq \mathcal{E}^{-1}(D^3)$, we can choose $M_{\pm\mathbf{r}}\in\mathcal{E}^{-1}(S^2)$ such that $\norm{\mathbf{r}}>1$. Besides, the two pure states, $\ketbra{\psi_{\pm}}{\psi_{\pm}}:=\mathcal{E}(M_{\pm\mathbf{r}})$, are orthogonal since their Bloch vectors are antipodal points on $S^2$. Let the spectral decomposition of $M_{\pm\mathbf{r}}$ be
\begin{equation}
M_{\pm\mathbf{r}}=\frac{1\pm\norm{\mathbf{r}}}{2}\ketbra{w_1}{w_1}+\frac{1\mp\norm{\mathbf{r}}}{2}\ketbra{w_2}{w_2},
\end{equation}
where $\ket{w_1}$ and $\ket{w_2}$ are two orthonormal pure states that are determined by $\mathbf{r}$. One can show that 
\begin{equation}
A_n:=\frac{1}{2}\ketbra{\psi_+}{\psi_+}^{\otimes n}-\frac{1}{2}\ketbra{\psi_-}{\psi_-}^{\otimes n}
\end{equation}
is an element of the feasible region $\mathscr{F}$, as $[M_{+\mathbf{r}},M_{-\mathbf{r}}]=0$. Acting $\mathcal{N}_{n}:=(\mathcal{E}^{-1})^{\otimes n}$ on $A_n$, we have
\begin{equation}
\begin{aligned}
&\mathcal{N}_{n}(A_n)=\frac{1}{2}(M_{+\mathbf{r}})^{\otimes n}-\frac{1}{2}(M_{-\mathbf{r}})^{\otimes n}\\
=&\frac{1}{2}\left[\left(\frac{1+\norm{\mathbf{r}}}{2}\right)^n-\left(\frac{1-\norm{\mathbf{r}}}{2}\right)^n\right]\ketbra{w_1}{w_1}^{\otimes n}\\
&+\frac{1}{2}\left[\left(\frac{1-\norm{\mathbf{r}}}{2}\right)^n-\left(\frac{1+\norm{\mathbf{r}}}{2}\right)^n\right]\ketbra{w_2}{w_2}^{\otimes n}\\
&+\cdots.
\end{aligned}
\end{equation}
Thus we have
\begin{equation}
\begin{aligned}
&R_*\ge R\left[\mathcal{N}_{n}\left(A_n\right)\right]\\
\ge&\frac{1}{2}\left[\left(\frac{1+\norm{\mathbf{r}}}{2}\right)^n-\left(\frac{1-\norm{\mathbf{r}}}{2}\right)^n\right]\\
&-\frac{1}{2}\left[\left(\frac{1-\norm{\mathbf{r}}}{2}\right)^n-\left(\frac{1+\norm{\mathbf{r}}}{2}\right)^n\right]\\
\ge&\Omega\left(\left(\frac{1+\norm{\mathbf{r}}}{2}\right)^n\right).
\end{aligned}
\end{equation}
Since $\frac{1+\norm{\mathbf{r}}}{2}>1$, Theorem~\ref{theorem:LowerBound} implies that the sample complexity $K$ for realizing $e^{-i\mathcal{N}_n(\rho)t}$ with accuracy $\epsilon$ has a lower bound of $K\ge\epsilon^{-1}2^{\Omega(n)}t^2$.

When $\mathcal{E}$ is not a unital map, then $\mathcal{E}^{-1}$ is also not a unital map. Denote $\mathcal{E}^{-1}(\frac{\mathbb{I}_2}{2})=M_{\mathbf{c}}$, then the center of the ellipsoid $\mathcal{E}^{-1}(D^3)$ is given by $\mathbf{c}$. The line passing through $(0,0,0)$ and $\mathbf{c}$ intersects the ellipsoid plane $\mathcal{E}^{-1}(S^2)$ at two points $\mathbf{s}$ and $\mathbf{t}$, which are antipodal points on $\mathcal{E}^{-1}(S^2)$ and satisfy $\mathbf{t}=-\frac{\norm{\mathbf{t}}}{\norm{\mathbf{s}}}\mathbf{s}$. Without loss of generality, we assume that $\norm{\mathbf{s}}>\norm{\mathbf{t}}$. Denote $\ketbra{\phi_{+}}{\phi_{+}}:=\mathcal{E}(M_{\mathbf{s}})$ and $\ketbra{\phi_{-}}{\phi_{-}}:=\mathcal{E}(M_{\mathbf{t}})$, they are antipodal pure states on $S^2$, thus $\braket{\phi_+}{\phi_-}=0$. Let the spectral decomposition of $M_{\mathbf{s}}$ and $M_{\mathbf{t}}$ be
\begin{equation}
\begin{aligned}
M_{\mathbf{s}}&=\frac{1+\norm{\mathbf{s}}}{2}\ketbra{v_1}{v_1}+\frac{1-\norm{\mathbf{s}}}{2}\ketbra{v_2}{v_2},\\
M_{\mathbf{t}}&=\frac{1+\norm{\mathbf{t}}}{2}\ketbra{v_2}{v_2}+\frac{1-\norm{\mathbf{t}}}{2}\ketbra{v_1}{v_1},
\end{aligned}
\end{equation}
where $\ket{v_1}$ and $\ket{v_2}$ are orthonormal pure states that are determined by $\mathbf{s}$. One can show that
\begin{equation}
B_n:=\frac{1}{2}\ketbra{\phi_+}{\phi_+}^{\otimes n}-\frac{1}{2}\ketbra{\phi_-}{\phi_-}^{\otimes n}
\end{equation}
is an element of the feasible region $\mathscr{F}$, as $[M_{\mathbf{s}},M_{\mathbf{t}}]=0$. Acting $\mathcal{N}_{n}$ on $B_n$, we have
\begin{equation}
\begin{aligned}
&\mathcal{N}_{n}(B_n)=\frac{1}{2}(M_{\mathbf{s}})^{\otimes n}-\frac{1}{2}(M_{\mathbf{t}})^{\otimes n}\\
=&\frac{1}{2}\left[\left(\frac{1+\norm{\mathbf{s}}}{2}\right)^n-\left(\frac{1-\norm{\mathbf{t}}}{2}\right)^n\right]\ketbra{v_1}{v_1}^{\otimes n}\\
&+\frac{1}{2}\left[\left(\frac{1-\norm{\mathbf{s}}}{2}\right)^n-\left(\frac{1+\norm{\mathbf{t}}}{2}\right)^n\right]\ketbra{v_2}{v_2}^{\otimes n}\\
&+\cdots.
\end{aligned}
\end{equation}
Thus we have
\begin{equation}
\begin{aligned}
R_*\ge R\left[\mathcal{N}_{n}\left(B_n\right)\right]\ge\Omega\left(\left(\frac{1+\norm{\mathbf{s}}}{2}\right)^n\right).
\end{aligned}
\end{equation}
Since $\frac{1+\norm{\mathbf{s}}}{2}>1$, Theorem~\ref{theorem:LowerBound} implies that the sample complexity $K$ for realizing $e^{-i\mathcal{N}_n(\rho)t}$ with accuracy $\epsilon$ has a lower bound of $K\ge\epsilon^{-1}2^{\Omega(n)}t^2$.

Finally, note that
\begin{equation}
\begin{aligned}
\mathrm{C}\text{-}e^{-i\mathcal{N}_n(\rho)t}=e^{-i\ketbra{1}{1}_c\otimes\mathcal{N}_n(\rho)t},
\end{aligned}
\end{equation}
and $\ketbra{1}{1}_c\otimes\mathcal{N}_n(\rho)$ has the same spectrum as $\mathcal{N}_n(\rho)$. Based on this observation, we can easily demonstrate that realizing controlled-$e^{-i\mathcal{N}_n(\rho)t}$ with accuracy $\epsilon$ also has a lower bound of $K\ge\epsilon^{-1}2^{\Omega(n)}t^2$.


\begin{thebibliography}{83}
\bibitem [{202()}]{2023complexity}%
	\BibitemOpen
	\href@noop {} {\bibinfo  {journal} {To be published}\ }\BibitemShut {NoStop}%
\end{thebibliography}

%

\end{document}